\documentclass[11pt,letterpaper]{article}

\usepackage{microtype}
\usepackage{geometry}
\usepackage{amsthm,amsmath,amssymb}
\usepackage{graphicx}
\usepackage{enumerate}
\usepackage[dvipsnames]{xcolor}
\definecolor{nblue}{RGB}{19, 39, 150}
\definecolor{dark-gray}{gray}{0.35}
\usepackage[colorlinks=true,citecolor=nblue,linkcolor=nblue,urlcolor=blue]{hyperref}
\usepackage[retainorgcmds]{IEEEtrantools}

\theoremstyle{plain}
\newtheorem{theorem}{Theorem}
\newtheorem{lemma}[theorem]{Lemma}
\newtheorem{corollary}[theorem]{Corollary}
\newtheorem{proposition}[theorem]{Proposition}

\theoremstyle{definition}
\newtheorem{definition}[theorem]{Definition}

\theoremstyle{remark}
\newtheorem{remark}[theorem]{Remark}

\usepackage{enumitem}

\usepackage{tikz}
\usetikzlibrary{arrows,automata,positioning}
\usepackage[latin1]{inputenc}
\usepackage{pgfplots}

\pgfplotsset{compat=1.13}

\usepackage[varg]{txfonts}
\usepackage{comment}
\usepackage{relsize}

\newcommand{\N}{\mathbb{N}}

\newcommand{\Z}{\mathbb{Z}}

\usepackage{ifthen}
% ...
\newenvironment{rtheorem}[3][]{%
\noindent\ifthenelse{\equal{#1}{}}{\bf #2 #3.}{\bf #2 #3 (#1)}%
\begin{it}}{\end{it}}

\newcommand{\mysetminusD}{\hbox{\tikz{\draw[line width=0.6pt,line cap=round] (3pt,0) -- (0,6pt);}}}
\newcommand{\mysetminusT}{\mysetminusD}
\newcommand{\mysetminusS}{\hbox{\tikz{\draw[line width=0.45pt,line cap=round] (2pt,0) -- (0,4pt);}}}
\newcommand{\mysetminusSS}{\hbox{\tikz{\draw[line width=0.4pt,line cap=round] (1.5pt,0) -- (0,3pt);}}}
\newcommand{\mysetminus}{\mathbin{\mathchoice{\mysetminusD}{\mysetminusT}{\mysetminusS}{\mysetminusSS}}}

\title{Rapid Mixing of the Switch Markov Chain for Strongly Stable Degree Sequences and 2-Class Joint Degree Matrices\thanks{This work is supported by NWO Gravitation Project NETWORKS, Grant Number
		024.002.003.}}

\author{Georgios Amanatidis \\
	\small Centrum Wiskunde \& Informatica (CWI)\\[-0.8ex]
	\small Amsterdam, The Netherlands\\
	\small\tt georgios.amanatidis@cwi.nl\\
	\and
	Pieter Kleer \\ 
	\small Centrum Wiskunde \& Informatica (CWI)\\[-0.8ex]
	\small Amsterdam, The Netherlands\\
	\small\tt kleer@cwi.nl
}

\begin{document}
	
\begin{titlepage}
	\clearpage\maketitle
	\thispagestyle{empty}
	
\begin{abstract}
The \emph{switch Markov chain} has been extensively studied as the most natural Markov Chain Monte Carlo approach for sampling graphs with prescribed degree sequences. We use comparison arguments with other---less natural but simpler to analyze---Markov chains, to show that the switch chain mixes rapidly in two different settings. We first study the classic problem of uniformly sampling simple undirected, as well as bipartite, graphs with a given degree sequence. We apply an embedding argument, involving a Markov chain defined by Jerrum and Sinclair (TCS, 1990) for sampling graphs that \emph{almost} have a given degree sequence, to show rapid mixing for degree sequences satisfying \emph{strong stability}, a notion closely related to $P$-stability. This results in a much shorter proof that unifies the currently known rapid mixing results of the switch chain and extends them up to sharp characterizations of $P$-stability. In particular, our work resolves an open problem posed by Greenhill (SODA, 2015).
		
Secondly, in order to illustrate the power of our approach, we study the problem of uniformly sampling graphs for which---in addition to the degree sequence---a \emph{joint degree distribution} is given. Although the problem was formalized over a decade ago, and despite its practical significance in generating synthetic network topologies, small progress has been made on the random sampling of such graphs. The case of a single degree class reduces to sampling of regular graphs, but beyond this almost nothing is known. We fully resolve the case of two degree classes, by showing that the switch Markov chain is \emph{always} rapidly mixing. Again, we first analyze an auxiliary chain for strongly stable instances on an augmented state space and then use an embedding argument.
	
%\bigskip\noindent \textbf{Keywords:} switch Markov chain; graph sampling; fixed degree sequence; joint degree distribution.
\end{abstract}

\end{titlepage}

\newpage

%%%%%%%%%%%%%%%%--to be removed--%%%%%%%%%%%%%%%%%%%%%%
\clearpage
\tableofcontents
\thispagestyle{empty}

\newpage

\setcounter{page}{1}
%%%%%%%%%%%%%%%%--to be removed--%%%%%%%%%%%%%%%%%%%%%%

%%%%%%%%%%%%%%%%%%%%%%%%%%%%%%%%%%%%%%%%%%%%%%%%%%%%%%%%%%%%%%%%%%%%%%%%
%%%%%%%%%%%%%%%%%%%%%%%%%%%%%%%%%%%%%%%%%%%%%%%%%%%%%%%%%%%%%%%%%%%%%%%%
\section{Introduction} 
%%%%%%%%%%%%%%%%%%%%%%%%%%%%%%%%%%%%%%%%%%%%%%%%%%%%%%%%%%%%%%%%%%%%%%%%
%%%%%%%%%%%%%%%%%%%%%%%%%%%%%%%%%%%%%%%%%%%%%%%%%%%%%%%%%%%%%%%%%%%%%%%%
The problem of (approximately) uniform sampling of simple graphs with a given degree sequence has received considerable attention, and has applications in domains as diverse as hypothesis testing in network structures \cite{Olding2014} and ecology, see, e.g., \cite{Rao1996} and references therein. We refer the interested  reader to \cite{Blitzstein2011} for more pointers to possible applications. 

Several variants of the problem have been studied, including sampling undirected, bipartite, connected, or directed graphs. In particular, there is an extensive line of work on Markov Chain Monte Carlo (MCMC) approaches, see, e.g., \cite{Jerrum1990,Kannan1999,Feder2006,Cooper2007,Greenhill2017journal,Erdos2013,Erdos2015decomposition,ErdosMMS2018,CooperDGH17,CarstensK18}.
In such an approach one studies a random walk on the space of all graphical realizations.\footnote{Sometimes the state space is augmented with a set of auxiliary states, as in \cite{Jerrum1990}.} This random walk is induced by making small local changes to a given realization using a probabilistic procedure that defines the Markov chain. The idea, roughly, is that after a sufficient number of steps, the so-called \emph{mixing time}, the resulting graph  corresponds to a sample from an almost uniform distribution over all graphical realizations of the given degree sequence. The goal is to show that the chain mixes \emph{rapidly}, meaning that one only needs to perform a polynomial (in the size of the graph) number of transitions  of the chain in order to obtain an approximately uniform sample.  One of the most well-known probabilistic procedures for making these small changes uses local operations called \emph{switches} (also known as \emph{swaps} or \emph{transpositions}); see, e.g., \cite{Rao1996} and Figure \ref{fig:switch_0} for an example.
The resulting \emph{switch Markov chain} has been shown to be rapidly mixing for various degree sequences \cite{Cooper2007,Greenhill2015, Greenhill2017journal,Erdos2013,Erdos2015decomposition,ErdosMMS2018}, but it is still open whether it is rapidly mixing for all degree sequences. %\yarem{Do Kannan et al.~actually conjecture this? I Can't find it in their paper. \pkblue{Good point. They hint at it, but never say it explicitly. However in some Erdos papers it is phrased as a conjecture...}} 

\begin{figure}[ht!]
	\centering
	\begin{tikzpicture}[scale=0.6]
	\coordinate (A1) at (0,0); 
	\coordinate (A2) at (0,1.5);
	\coordinate (M1) at (2,0);
	\coordinate (M2) at (2,1.5);
	
	\coordinate (P1) at (3,1);
	\coordinate (P2) at (4,1);
	%\draw [very thick,fill opacity=0.5] (A1) -- (A2) -- (M2) -- (T2)--(T1)--(M1)--(A1)--cycle;
	
	\node at (A1) [circle,scale=0.7,fill=black] {};
	\node (a1) [below=0.1cm of A1]  {$v$};
	\node at (A2) [circle,scale=0.7,fill=black] {};
	\node (a2) [above=0.1cm of A2]  {$w$};
	\node at (M1) [circle,scale=0.7,fill=black] {};
	\node (m1) [below=0.1cm of M1]  {$x$};
	\node at (M2) [circle,scale=0.7,fill=black] {};
	\node (m2) [above=0.1cm of M2]  {$y$};
	
	\path[every node/.style={sloped,anchor=south,auto=false}]
	(A1) edge[-,very thick] node {} (A2) 
	(M1) edge[-,very thick] node {} (M2)
	(P1) edge[->,very thick] node {} (P2);  
	\end{tikzpicture}
	\quad
	\begin{tikzpicture}[scale=0.6]
	\coordinate (A1) at (0,0); 
	\coordinate (A2) at (0,1.5);
	\coordinate (M1) at (2,0);
	\coordinate (M2) at (2,1.5);
	
	%\draw [very thick,fill opacity=0.5] (A1) -- (A2) -- (M2) -- (T2)--(T1)--(M1)--(A1)--cycle;
	
	\node at (A1) [circle,scale=0.7,fill=black] {};
	\node (a1) [below=0.1cm of A1]  {$v$};
	\node at (A2) [circle,scale=0.7,fill=black] {};
	\node (a2) [above=0.1cm of A2]  {$w$};
	\node at (M1) [circle,scale=0.7,fill=black] {};
	\node (m1) [below=0.1cm of M1]  {$x$};
	\node at (M2) [circle,scale=0.7,fill=black] {};
	\node (m2) [above=0.1cm of M2]  {$y$};
	
	\path[every node/.style={sloped,anchor=south,auto=false}]
	(A1) edge[-,very thick] node {} (M2) 
	(M1) edge[-,very thick] node {} (A2);  
	\end{tikzpicture}
	\caption{Example of a switch in which edges $\{v,w\},\{x,y\}$ are replaced by $\{v,y\},\{x,w\}$. Note that the degree sequence is preserved when applying a switch operation.}
	\label{fig:switch_0}
\end{figure}
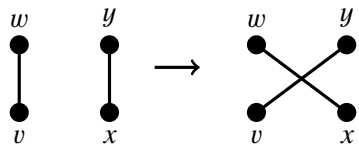

In this work, besides the problem of sampling undirected, as well as bipartite, simple graphs with a given degree sequence, we also focus on  the problem of sampling undirected simple graphs with a given \emph{joint degree distribution}. That is, in addition to the degrees, the number of edges between nodes of degree $i$ and degree $j$ is also specified for every pair $(i,j)$.\footnote{We refer the reader to Subsection \ref{sec:jdm_model} for an exact definition and also Appendix \ref{app:example} for an example.} The motivation for using such a metric is that this extra information restricts the space of possible realizations to graphs with more desirable structure.
This was first observed by Mahadevan et al.~\cite{MahadevanKFV06} who argued that the joint degree distribution is a much more reliable metric for a synthetic graph to resemble a real network topology, compared to just using the degree sequence. The joint degree matrix model of Amanatidis, Green, and Mihail \cite{Amanatidis2015} formalizes this approach.
Although there are polynomial-time algorithms that produce a graphical realization of a given joint degree distribution \cite{Amanatidis2015,AGM2018,StantonP11,CzabarkaDEM15,GjokaTM15}, it is not
known how to  uniformly sample such a realization efficiently. In particular, bounding the mixing time of the natural restriction of the switch Markov chain for this setting has been an open problem since the introduction of the model \cite{Amanatidis2015,StantonP11,Erdos2015decomposition}.

\vspace{-5pt}
%\pkrem{I moved the title of subsection up. Otherwise it reads a bit weird right after definition of JDM...}
\paragraph{Our Contribution.} The proofs of the results in \cite{Erdos2013,Greenhill2015,Erdos2015decomposition,ErdosMMS2018} for the analysis of the switch Markov chain in undirected (or bipartite) graphs are all using conceptually similar  ideas to the ones introduced by Cooper, Dyer and Greenhill \cite{Cooper2007} for the analysis of the switch chain for regular undirected graphs, and are based on the multicommodity flow method of Sinclair \cite{Sinclair1992}. The individual parts of this template can become quite technical and require long proofs.
In this work we take a different approach for proving that the switch chain is rapidly mixing. First we analyze some easier auxiliary Markov chain; such a chain  can be used to sample graphical realizations that \emph{almost} have a given fixed degree sequence or joint degree distribution. We  show that there exists an efficient multicommodity flow for the auxiliary chain when the given instance is \emph{strongly stable}, \footnote{In the case of sampling graphs with a given degree sequence, the existence of such a flow for the auxiliary chain we use was already claimed in \cite{Jerrum1990}, at least for the degree sequences satisfying a more restrictive analog of inequality \eqref{eq:stable1}. For completeness, we give a detailed proof in Appendix \ref{app:js}.}  
and then show how it can be transformed into an efficient multicommodity flow for the switch chain.
Note that this last step compares two Markov chains with different state spaces, as the auxiliary chain samples from a strictly larger set of graphs than the switch chain. Thus, for the flow transformation we use embedding arguments similar to those in Feder et al.~\cite{Feder2006}.\footnote{We also refer the reader to the survey of Dyer et al.~\cite{Dyer2006} for more on Markov chain comparison techniques.}
%\footnote{We also refer the reader to the survey of Dyer, Goldberg, Jerrum, and Martin~\cite{Dyer2006} for more on Markov chain comparison techniques.} %We also refer the reader to the survey of Dyer et al.~\cite{Dyer2006} for more on Markov chain comparison techniques. 
\medskip

\noindent Using the aforementioned approach we obtain the following two main results:%\pkrem{Maybe numbers here, instead of bulletins?}
%\vspace{-5pt}
\begin{enumerate}[label=\arabic*),leftmargin=*]
	\item We show rapid mixing of the switch chain for \emph{strongly stable} families of degree sequences (Theorem \ref{thm:switch}), thus providing a rather short proof that unifies and extends the results in \cite{Cooper2007,Greenhill2015} (Corollaries \ref{lem:stable_jerrum} and \ref{cor:stable}). 
	
	We introduce strong stability as a stricter version of the notion of $P$-stability \cite{Jerrum1990}. The strong stability condition is satisfied by the degree sequences in the works \cite{Kannan1999,Cooper2007,Greenhill2015,Erdos2013,Erdos2015decomposition,ErdosMMS2018} and by characterizations of $P$-stability \cite{Jerrum1989graphical}. These characterizations are employed to get Corollaries \ref{lem:stable_jerrum} and \ref{cor:stable}. In particular, investigating the mixing time for the degree sequences of Corollary \ref{cor:stable} was posed as an open question by Greenhill \cite{Greenhill2015}; here we resolve it in the positive. 
	
	\item We show that the switch chain restricted on the space of the graphic realizations of a given joint degree distribution with two degree classes is \emph{always} rapidly mixing (Theorem \ref{thm:stable_jdm1}). 
	
	Despite being for the case of two classes, this is the very first rapid mixing result for the problem. 
	Establishing the rapid mixing of the auxiliary chain in this case presents significant challenges. To attack this problem, we rely on ideas introduced by Bhatnagar, Randall, Vazirani and Vigoda \cite{Bhatnagar2008}. 
	At the core of this approach lies the \emph{mountain-climbing problem} \cite{Homma1952,Whittaker1966}.
	We should note that the auxiliary chain used here is analyzed in a much more general model than the joint degree matrix model. 
\end{enumerate}

Besides the above results we also unify the results in \cite{Kannan1999,Erdos2013,Erdos2015decomposition,ErdosMMS2018} for bipartite degree sequences and extend them for the  special case of an equally sized bipartition; see Corollaries \ref{cor:bipartite} and \ref{cor:bipartite2} in Appendix \ref{sec:bipartite}. 
We should note that the unification of the existing results mentioned so far is qualitative rather than quantitative, in the sense that our simpler, indirect approach provides weaker polynomial bounds for the mixing time. For examples of explicit mixing time bounds % on the switch chain for sampling undirected graphs, 
we refer the reader to \cite{Cooper2007,Cooper2012corrigendum,Greenhill2017journal}. 

Finally, as a byproduct of our analysis for the auxiliary Markov chain used to prove Theorem \ref{thm:stable_jdm1}, we obtain the first \emph{fully polynomial almost uniform generator} \cite{Sinclair1989} for sampling graphical realizations of certain sparse \emph{partition adjacency matrix} instances with two partition classes \cite{Czabarka14,ErdosHIM2017} (this is a generalization of the joint degree distribution problem; see Appendix \ref{sec:auxiliary} for definitions). See Corollary \ref{cor:first_sparse} in Appendix \ref{sec:stable_jdm}. 
%\pkrem{Something about FPRAS}

%Summarizing, based on the rapid mixing of the Jerrum-Sinclair chain, our work provides a rather short rapid mixing proof for the switch chain for a large range of degree sequences unifying a long line of work \cite{Kannan1999,Cooper2007,Greenhill2015,Erdos2013,Erdos2015decomposition,ErdosMMS2018} and extending these results to sharp characterizations of $P$-stability. To the best of our knowledge, no efficient approximately uniform %(either MCMC or non-MCMC)
%MCMC generator is known for non-trivial, %\footnote{For example, there exist classes of non-stable degree sequences where every sequence only has one graphical realization, see \cite{Jerrum1989graphical}.} 
%non-stable classes of degree sequences. 

\vspace{-5pt}
\paragraph{Related Work.}
Here the focus is on MCMC approaches. As such approaches have impractical mixing times in general, we should note that there is a line of work on non-MCMC sampling algorithms which, albeit having weaker properties, do have practical running times. See, e.g., \cite{BayatiKS10,GaoW18}  and references therein. 

Jerrum and Sinclair \cite{Jerrum1990} give a fully polynomial almost uniform generator for generating graphical realizations of degree sequences coming from any \emph{$P$-stable} family of  sequences (see preliminaries). The auxiliary chain we use to prove Theorem \ref{thm:switch}, henceforth referred to as Jerrum-Sinclair (JS) chain, is presented in \cite{Jerrum1990} as a more straightforward implementation of this generator. One drawback is that the algorithms in \cite{Jerrum1990} work with auxiliary states. Kannan, Tetali and Vempala \cite{Kannan1999} introduce the switch chain as a simpler and more direct generator that does not have to work with auxiliary states. They addressed the mixing time of such a switch-based Markov chain for the regular bipartite case. 
Cooper et al.~\cite{Cooper2007} then gave a rapid mixing proof for regular undirected graphs, and later Greenhill \cite{Greenhill2015} extended this result to certain ranges of irregular degree sequences (see also Greenhill and Sfragara \cite{Greenhill2017journal}). Mikl\'os, Erd\H{o}s and Soukup  \cite{Erdos2013} proved rapid mixing for the \emph{half-regular} bipartite case, and Erd\H{o}s, Mikl\'os and Toroczkai \cite{Erdos2015decomposition} for the \emph{almost half-regular} case. Very recently, Erd\H{o}s et al.~\cite{ErdosMMS2018} presented a range of bipartite degree sequences unifying and generalizing the results in \cite{Erdos2013,Erdos2015decomposition}.  

Switch-based Markov chain Monte Carlo approaches have also been studied for other graph sampling problems. Feder et al.~\cite{Feder2006} as well as Cooper et al.~\cite{CooperDGH17} study the mixing time of a Markov chain using a switch-like probabilistic procedure (called a \emph{flip}) for sampling connected graphs. For sampling perfect matchings, switch-based Markov chains have also been studied, see, e.g., the recent work of Dyer, Jerrum and M\"{u}ller \cite{Dyer2017} and references therein. It is interesting to note that Dyer et al.~\cite{Dyer2017} also use a lemma on the mountain-climbing problem that is very similar to Lemma \ref{lem:mountain}. 

The joint degree matrix model was first studied by Patrinos and Hakimi \cite{PatrinosH76}, albeit with a  different formulation and name, and was reintroduced in Amanatidis et al.~\cite{Amanatidis2015}. While it has been shown that the switch chain restricted on the space of the graphic realizations of any given joint degree distribution is irreducible \cite{Amanatidis2015,CzabarkaDEM15}, almost no progress has been made towards bounding its mixing time. 
Stanton and Pinar \cite{StantonP11} performed experiments based on the autocorrelation of each edge, suggesting that the switch chain mixes quickly. The only  relevant result is that of Erd\H{o}s et al.~\cite{Erdos2015decomposition} showing 
fast mixing for a related Markov chain over the severely restricted subset of the so-called \emph{balanced} joint degree matrix realizations; this special case, however, lacks several of the technical challenges that arise in the original problem. %\yarem{Maybe also add Bassler et al.~here?}
%that only contains realizations where the edges satisfying the joint degree requirements are as equally %distributed on the nodes of each class as possible. 

%\color{red}
%
%
%Our proof outline is using, to some extent, similar ideas as the seminal work of
%Kannan, Tetali and Vempala \cite{Kannan1999}.
%Roughly speaking, in \cite{Kannan1999} these ideas are analyzed through Tutte's construction, seemingly inspired by the analysis in \cite{Jerrum1990}. The analysis uses a mixture of perfect matching and switch-based arguments in Tutte's construction.
%However, as Cooper, Dyer and Greenhill \cite{Cooper2007} point out: \emph{``this does not appear complete in every respect, and it is certainly unclear how it extends to generating non-bipartite graphs"}. 
%The main conceptual and technical difference with our approach is that we split the analysis into two separate parts. We first study the JS chain, and we then compare it with the switch chain in terms of efficient multicommodity flows. 

%\medskip 

\vspace{-5pt}
\paragraph{Outline.} In Section \ref{sec:preliminaries} we give all the necessary preliminaries and we formally  
describe  the JS chain, the switch chain, and the restricted switch chain. Our first main result is Theorem \ref{thm:switch} in Section \ref{sec:main_result}, where we show that the switch chain is rapidly mixing for families of strongly stable degree sequences. Given the rapid mixing of the JS chain (Appendix \ref{app:js}), the proof of Theorem \ref{thm:switch} in Section \ref{sec:main_result} is self-contained. The corresponding result for the bipartite case is completely analogous and is deferred to Appendix \ref{sec:bipartite}. In Section \ref{sec:switch_jdm_main} we state our second main result, Theorem \ref{thm:stable_jdm1}. For the sake of presentation, we defer the proof of Theorem \ref{thm:stable_jdm1} to Appendices \ref{sec:auxiliary}, \ref{sec:stable_jdm} and \ref{sec:switch_for_jdm}, and we only include a short discussion of our approach in Section \ref{sec:switch_jdm_main}.

%\color{black}

%%%%%%%%%%%%%%%%%%%%%%%%%%%%%%%%%%%%%%%%
%%%%%%%%%%%%%%%%%%%%%%%%%%%%%%%%%%%%%%%%
\section{Preliminaries}\label{sec:preliminaries}
%%%%%%%%%%%%%%%%%%%%%%%%%%%%%%%%%%%%%%%%
%%%%%%%%%%%%%%%%%%%%%%%%%%%%%%%%%%%%%%%%
We begin with the  preliminaries regarding Markov chains and the multicommodity flow method of Sinclair \cite{Sinclair1992}. For Markov chain definitions not given here, see for example \cite{Levin2009}.

%\yarem{Shall we define MC, ergodic, etc.? Maybe add a short appendix for completeness. \pkblue{We never use the definitions, so not sure if we want to elaborate on it. Added reference to standard textbook for completeness. Good?}}

Let $\mathcal{M} = (\Omega,P)$ be an ergodic, time-reversible Markov chain over state space $\Omega$ with transition matrix $P$ and stationary distribution $\pi$. We write $P^t(x,\cdot)$ for the distribution over $\Omega$ at time step $t$ given that the initial state is $x \in \Omega$. 
The \emph{total variation distance} at time $t$ with initial state $x$ is
\[
\Delta_x(t) = \max_{S \subseteq \Omega} \big| P^t(x,S) - \pi(S)\big| = \frac{1}{2}\sum_{y \in \Omega} \big| P^t(x,y) - \pi(y)\big| \,,
\]
and the \emph{mixing time} $\tau(\epsilon)$ is defined as 
%\[
$\tau(\epsilon) = \max_{x \in \Omega}\left\{ \min\{ t : \Delta_x(t') \leq \epsilon \text{ for all } t' \geq t\}\right\}$. % \,.
%\]
Informally, $\tau(\epsilon)$  is the number of steps until the Markov chain is $\epsilon$-close to its stationary distribution. A Markov chain is said to be \emph{rapidly mixing} if the mixing time can be upper bounded by a function polynomial in $\ln(|\Omega|/\epsilon)$. 

It is well-known that, since the Markov chain is time-reversible, the matrix $P$ only has real eigenvalues $1 = \lambda_0 > \lambda_1 \geq \lambda_2 \geq \dots \geq \lambda_{|\Omega|-1} > -1$.  We may replace the transition matrix $P$ of the Markov chain by $(P+I)/2$, to make the chain \emph{lazy}, and hence  guarantee that all its eigenvalues are non-negative. It then follows that the second-largest eigenvalue of $P$ is $\lambda_1$. In this work we always consider the lazy versions of the Markov chains involved.
It follows directly from Proposition 1 in \cite{Sinclair1992} that
%\begin{equation*}
$\tau(\epsilon) \ 
\leq \ \frac{1}{1 - \lambda_1} \big(\ln(1/\pi_*) + \ln(1/\epsilon)\big)$, %  \,,
%\end{equation*}
where $\pi_* = \min_{x \in \Omega} \pi(x)$. For the special case where $\pi(\cdot)$ is the uniform distribution, as is the case here, the above bound becomes
$\tau(\epsilon) 
\leq \ln(|\Omega|/\epsilon)/(1 - \lambda_1)$.
%\begin{equation*}\label{eq:mixing_time}
%\tau(\epsilon) 
%\leq \frac{\ln(|\Omega|/\epsilon)}{1 - \lambda_1} \,.
%\end{equation*}
%\pkrem{Removed sentence about mixing time determined by spectral gap}%This roughly implies that the mixing time is determined  by the \emph{spectral gap} $(1 - \lambda_1)$, or rather its inverse, the \emph{relaxation time} $(1 - \lambda_1)^{-1}$. 
The quantity $(1 - \lambda_1)^{-1}$ can be upper bounded using the \emph{multicommodity flow method} of Sinclair  \cite{Sinclair1992}.

We define the state space graph of the chain $\mathcal{M}$ as the directed graph $\mathbb{G}$ with node set $\Omega$ that contains exactly the arcs $(x,y) \in \Omega \times \Omega$ for which $P(x,y) > 0$ and $x \neq y$. Let $\mathcal{P} = \cup_{x \neq y} \mathcal{P}_{xy}$, where $\mathcal{P}_{xy}$ is the set of simple paths between $x$ and $y$ in the state space graph $\mathbb{G}$.
A \emph{flow} $f$ in $\Omega$ is a function $\mathcal{P} \rightarrow [0,\infty)$ satisfying
$\sum_{p \in \mathcal{P}_{xy}} f(p) = \pi(x)\pi(y)$ for all $x,y \in \Omega, x \neq y$.
%\[
%\sum_{p \in \mathcal{P}_{xy}} f(p) = \pi(x)\pi(y) \quad \text{ for all } x,y \in \Omega, x \neq y \,.
%\]
The flow $f$ can be extended to a function on oriented edges of $\mathbb{G}$ by setting 
$f(e) =  \sum_{p \in \mathcal{P} : e \in p } f(p)$,
%\[
%f(e) = \! \sum_{p \in \mathcal{P} : e \in p } \! f(p) \,,
%\]
so that $f(e)$ is the total flow routed through $e\in E(\mathbb{G})$. Let $\ell(f) = \max_{p \in \mathcal{P} : f(p) > 0} |p|$ be the length of a longest flow carrying path, and let 
$
\rho(e) = f(e)/Q(e)
$
be the \emph{load} of the edge $e$, where $Q(e) = \pi(x)P(x,y)$ for $e = (x,y)$.
The maximum load of the flow is
%\[ 
$
\rho(f) = \max_{e\in E(\mathbb{G})} \rho(e).
$ %\]
Sinclair (\cite{Sinclair1992}, Corollary $6^{\,\prime}$) shows that 
%\begin{equation*}\label{eq:sinclair}
$(1 - \lambda_1)^{-1} \leq \rho(f)\ell(f)$.
%\end{equation*} 

We  use the following standard technique for bounding the maximum load of a flow in case the chain $\mathcal{M}$ has uniform stationary distribution $\pi$. 
Suppose $\theta$ is the smallest positive transition probability of the Markov chain between two distinct states.
If $b$ is such that $f(e) \leq b / |\Omega|$ for all $e\in E(\mathbb{G})$, then it follows that
$
\rho(f) \leq b/\theta
$.
Thus, we have
\begin{equation*}
\tau(\epsilon) \leq \frac{\ell(f)\cdot b}{\theta}\ln(|\Omega|/\epsilon)\,.
\end{equation*}
Now, if $\ell(f), b$ and $1/\theta$ can be bounded by a function polynomial in $\log(|\Omega|)$, it follows that
the Markov chain $\mathcal{M}$ is rapidly mixing. In this case, we say that $f$ is an \emph{efficient} flow.
Note that in this approach the transition probabilities do not play a role as long as $1/\theta$ is polynomially bounded.

%%%%%%%%%%%%%%%%%%%%%%%%%%%%%%%%%%%%%%%%%%%%%%%%%%%%%%%%%%%%%%%%%%%%%%%%
\subsection{Graphical Degree Sequences}\label{sec:degree_sequence}
%%%%%%%%%%%%%%%%%%%%%%%%%%%%%%%%%%%%%%%%%%%%%%%%%%%%%%%%%%%%%%%%%%%%%%%%
A sequence of non-negative integers $d = (d_1,\dots,d_n)$ is called a \emph{graphical degree sequence} if there exists a simple, undirected, labeled graph on $n$ nodes having degrees $d_1,\dots,d_n$; such a graph is called a \emph{graphical realization of $d$}. 
For a given  degree sequence $d$,  $\mathcal{G}(d)$ denotes the set of all graphical realizations of $d$. Throughout this work we only consider sequences $d$ with positive components, and for which $\mathcal{G}(d) \neq \emptyset$. 
Let $\mathcal{G}'(d) = \cup_{d'} \mathcal{G}(d')$ with $d'$ ranging over the set 
$\left\{ d' : d'_j \leq d_j \text{ for all } j \text{, and } \sum_{i=1}^n |d_i - d_i'| \leq 2\right\}$.
%\[
%\bigg\{ d' : d'_j \leq d_j \text{ for all } j \text{, and } \sum_{i=1}^n |d_i - d_i'| \leq 2\bigg\} \,.
%\]
That is, we have \emph{(i)} $d' = d$, or \emph{(ii)} there exist distinct $\kappa,\lambda$ such that
$d'_i = d_i - 1$ if $i \in \{\kappa,\lambda\}$ and $d'_i =d_i$ otherwise,
or \emph{(iii)} there exists a $\kappa$ so that
$d'_i = d_i - 2$ if $i = \kappa$ and $d'_i =d_i$ otherwise.
%That is, we have $d' = d$, or there exist distinct $\kappa,\lambda$ such that
%\[
%d'_i = \left\{ \begin{array}{ll} d_i - 1 & \text{ if } i \in \{\kappa,\lambda\}, \\
%d_i & \text{ otherwise}\,,
%\end{array}\right.
%\]
%or there exists a $\kappa$ so that
%\[
%d'_i = \left\{ \begin{array}{ll} d_i - 2 & \text{ if } i = \kappa, \\
%d_i & \text{ otherwise}\,.
%\end{array}\right.
%\]
%
%That is, we have 
%\[d' = d, \text{ or there exist distinct } \kappa,\lambda \text{ such that }
%d'_i = \left\{ \begin{array}{ll} d_i - 1 & \text{ if } i \in \{\kappa,\lambda\}, \\
%d_i & \text{ otherwise}\,,
%\end{array}\right.
%\text{ or there exists a } \kappa \text{ so that }
%d'_i = \left\{ \begin{array}{ll} d_i - 2 & \text{ if } i = \kappa, 
%d_i & \text{ otherwise}\,.
%\end{array}\right.
%\]
In the case (ii) we say that $d'$ has two nodes with degree deficit one, and in the case (iii) we say that $d'$ has one node with degree deficit two.
A family $\mathcal{D}$ of graphical degree sequences is called \emph{$P$-stable} \cite{Jerrum1990}, if there exists a polynomial $q(n)$ such that for all $d \in \mathcal{D}$ 
we have $|\mathcal{G}'(d)|/|\mathcal{G}(d)| \leq q(n)$,
%\begin{equation*}\label{eq:def_stable}
%\frac{|\mathcal{G}'(d)|}{|\mathcal{G}(d)|} \leq q(n) \,,
%\end{equation*}
where $n$ is the number of components of $d$.

Jerrum and Sinclair \cite{Jerrum1990} define the following Markov chain on $\mathcal{G}'(d)$, which will henceforth be referred to as the \emph{JS chain}.\footnote{A slightly different definition of stability is given by Jerrum, McKay and Sinclair \cite{Jerrum1989graphical}. Based on this variant, one could define the corresponding variant of the JS chain. Nevertheless, the definitions of stability in \cite{Jerrum1989graphical} and \cite{Jerrum1990} (and their corresponding definitions of strong stability) are equivalent. To avoid confusion, here we only use the definitions in \cite{Jerrum1990}.} Let $G \in \mathcal{G}'(d)$ be the current state of the chain:
\begin{itemize}[topsep=3pt,itemsep=3pt,parsep=0pt,partopsep=20pt]
	\item With probability $1/2$, do nothing.
	\item Otherwise, select an ordered pair $i,j$ of nodes uniformly at random and 
	\begin{itemize}[topsep=3pt,itemsep=3pt,parsep=0pt,partopsep=20pt]
		\item if $G \in \mathcal{G}(d)$ and $(i,j)$ is an edge of $G$, then delete $(i,j)$ from $G$,
		\item if $G \notin \mathcal{G}(d)$, the degree of $i$ in $G$ is less than $d_i$, and $(i,j)$ is not an edge of $G$, then add $(i,j)$ to $G$. If the new degree of $j$  exceeds $d_j$, then select an edge $(j,k)$ uniformly at random and delete it.
	\end{itemize}
\end{itemize}
The graphs $G,G' \in \mathcal{G}'(d)$ are %said to be 
\emph{JS adjacent} if $G$ can be obtained from $G'$ with positive probability in one transition of the JS chain and vice versa.  The following properties of the JS chain are easy to check.
\begin{theorem}[Follows by \cite{Jerrum1990}] 
	The JS chain is irreducible, aperiodic and symmetric, and, hence, has uniform stationary distribution over $\mathcal{G}'(d)$. Moreover, $P(G,G')^{-1} \leq 2n^3$ for all JS adjacent $G,G' \in \mathcal{G}'(d)$, and also 
	the maximum in- and out-degrees of the state space graph of the JS chain are bounded by $n^3$.
\end{theorem}

We say that two graphs $G, G'$ are within distance $r$ in the JS chain if there exists a path of at most length $r$ from $G$ to $G'$ in the state space graph of the JS chain. By $\text{dist}(G,d)$ we denote the  minimum distance of $G$ from an element in $\mathcal{G}(d)$.
The following parameter will play a central role in this work.  Let 
\begin{equation}\label{eq:distance}
k_{JS}(d) = \max_{G \in \mathcal{G}'(d)} \text{dist}(G,d) \,.
\end{equation}
Based on the parameter $k_{JS}(d)$, we define the notion of \emph{strong stability}. The simple observation in Proposition \ref{prop:strong_stable} justifies the terminology. For the different settings studied in this work, i.e., for sampling bipartite graphs or joint degree matrix realizations, the definition of $k_{JS}$ is adjusted accordingly (see Appendices \ref{sec:bipartite} and \ref{sec:auxiliary}).

\begin{definition}[Strong stability]\label{def:strong_stable}
	A family of graphical degree sequences $\mathcal{D}$ is called \emph{strongly stable} if there exists a constant $\ell$ such that $k_{JS}(d) \leq \ell$ for all $d \in \mathcal{D}$.
\end{definition}

\begin{proposition}\label{prop:strong_stable}
	If $\mathcal{D}$ is strongly stable, then  it is $P$-stable.
\end{proposition}
\begin{proof}
	Suppose $\mathcal{D}$ is strongly stable with respect to the constant $\ell$.
	Let $d \in \mathcal{D}$ be a degree sequence with $n$ components. For every $G \in \mathcal{G}'(d) \mysetminus \mathcal{G}(d)$ choose some $\varphi(G) \in \mathcal{G}(d)$ within distance $k = k_{JS}(d)$ of $G$. As the in-degree of any node in the state space graph of the JS chain is bounded by $n^3$, the number of paths with length at most $k$ that end up at any particular graph in $\mathcal{G}(d)$ is upper bounded by $(n^3)^k$. Therefore, 
	$|\mathcal{G}'(d)|/|\mathcal{G}(d)| \leq n^{3k} \le  n^{3\ell}$,
%	\[
%	\frac{|\mathcal{G}'(d)|}{|\mathcal{G}(d)|} \leq n^{3k},
%	\]
	meaning that $\mathcal{D}$ is stable, since $\ell$ is constant.
\end{proof}

Finally, the lazy version of the \emph{switch chain on $\mathcal{G}(d)$} is defined as follows (see, e.g., \cite{Kannan1999,Cooper2007}). 
Let $G\in \mathcal{G}(d)$ be the current state of the  chain: 
\begin{itemize}[topsep=3pt,itemsep=3pt,parsep=0pt,partopsep=20pt]
	\item With probability $1/2$, do nothing.
	\item Otherwise, select two edges $\{a,b\}$ and $\{x,y\}$ uniformly at random, and select a perfect matching $M$ on nodes $\{x,y,a,b\}$ uniformly at random (there are three possible options). If $M \cap E(G) = \emptyset$, then delete $\{a,b\}, \{x,y\}$ from $E(G)$ and add the edges of $M$. This local operation is called a \emph{switch}.
\end{itemize}
The graphs $G,G' \in \mathcal{G}(d)$ are \emph{switch adjacent} if $G$ can be obtained from $G'$ with positive probability in one transition of this chain and vice versa. 
It is well-known that the switch chain is  irreducible, aperiodic and symmetric (e.g., \cite{Greenhill2017journal} and references therein), and, thus, has uniform stationary distribution over $\mathcal{G}(d)$. 
Furthermore, is it a matter of simple counting that $P(G,G')^{-1} \leq 6n^4$ for all switch adjacent $G,G' \in \mathcal{G}(d)$, and  
the maximum in- and out-degrees of the state space graph of the switch chain are bounded by $n^4$.

  %%%%%%%%%%%%%%%%%%%%%%%%%%%%%%%%%%%%%%%%%%%%%%%%%%%%%%%%%%%%%%%%%%%%%%%%
  \subsection{Joint Degree Matrix Model}\label{sec:jdm_model}
  %%%%%%%%%%%%%%%%%%%%%%%%%%%%%%%%%%%%%%%%%%%%%%%%%%%%%%%%%%%%%%%%%%%%%%%%
  Here in addition to the degrees, we would also like to specify the number of edges between nodes of degree $i$ and nodes of degree $j$ for every pair $(i,j)$. 
  %The motivation for this is that this extra information restricts the space of possible realizations to graphs that more closely resemble real-world networks. 
  Let $V = \{1,\dots,n\}$ be a set of nodes. An instance of the joint degree matrix (JDM) model is given by a partition $V_1 \cup V_2 \cup \dots \cup V_q$ of $V$ into pairwise disjoint (degree) classes,  a symmetric \emph{joint degree  matrix} $c = (c_{ij})_{i,j \in [q]}$  of non-negative integers, and a sequence $d = (d_1,\dots,d_q)$ of non-negative integers.\footnote{This is  shorthand notation for the degree sequence. Alternatively, we could write $\hat{d} = \Big(d_1^1,\dots,d_1^{|V_1|},\dots,d_q^1,\dots,d_q^{|V_q|}\Big)$ corresponding to the definition of a graphical degree sequence. In such a case, $d_i^{j}=d_i$ for $i\in V$ and $j \in \{1,\ldots,|V_i|\}$.} 
  We say that the tuple $((V_i)_{i \in q},c,d)$ (or just $(c,d)$ when it is clear what the partition is) is graphical, if there exists a simple, undirected, labeled graph $G = (V,E)$ on the nodes in $V$ such that all nodes in  $V_i$ have degree $d_i$ and there are precisely $c_{ij}$ edges between nodes in $V_i$ and $V_j$. 
  %This is denoted by $E[V_i,V_j] = c_{ij}$. 
  Such a $G$ is called a graphical realization of the tuple. We let $\mathcal{G}((V_i)_{i \in q},c,d)$ (or just $\mathcal{G}(c,d)$) denote the set of all graphical realizations of  $((V_i)_{i \in q},c,d)$. %We often write $\mathcal{G}(c,d)$ instead of $\mathcal{G}((V_i)_{i \in [q]},c,d)$ but it will always be clear what the partition is. %For general instances, it is not known if an initial state can be computed in time polynomial in $n$. It is conjectured to be NP-hard \cite{ErdosHIM2017} in general. 
  %
  %One special case of the partition adjacency model is when the degrees are component-wise regular, i.e., when $d_a = d_b$ whenever $a,b \in V_i$ for $i = 1,\dots,q$. This corresponds to the \emph{joint degree matrix} model \cite{Amanatidis2015}.
 We focus on the case of $q=2$, i.e., when two degree classes are given.  
  
While switches maintain the degree sequence, this is no longer true for the joint degree constraints. However, some switches do respect these constraints as well, e.g., if $w, y$ in Figure \ref{fig:switch_0} are in the same degree class. Thus, we are interested in the following (lazy) \emph{restricted switch Markov chain} for sampling graphical realizations of $\mathcal{G}(c,d)$. Let $G\in \mathcal{G}(c,d)$ be the current state of the  chain: 
  \begin{itemize}[topsep=3pt,itemsep=3pt,parsep=0pt,partopsep=20pt]
  	\item With probability $1/2$, do nothing.
  	\item Otherwise, try to perform a \emph{switch move}: select two edges $\{a,b\}$ and $\{x,y\}$ uniformly at random, and select a perfect matching $M$ on nodes $\{x,y,a,b\}$ uniformly at random. If $M \cap E(G) = \emptyset$ and $G + M - (\{a,b\} \cup \{x,y\}) \in \mathcal{G}(c,d)$, then  delete $\{a,b\}, \{x,y\}$ from $E(G)$ and add the edges of $M$. %This local operation is called a \emph{switch}.
  	%\item With probability $1/4$ perform a \emph{double switch move}: \pkblue{to do...}
  \end{itemize}
  This chain is  irreducible, aperiodic and symmetric \cite{Amanatidis2015,CzabarkaDEM15}.
  Like the switch chain defined above, $P(G,G')^{-1} \leq n^4$ for all adjacent $G, G' \in \mathcal{G}'(c,d)$, and also the maximum in- and out-degrees of the state space graph are less than $n^4$.
Since bounding the mixing time of this chain on $\mathcal{G}(c,d)$ has been elusive \cite{Amanatidis2015,StantonP11,Erdos2015decomposition}, we follow a similar approach as in the case of the switch chain for undirected and bipartite graphs. In Appendix \ref{sec:auxiliary} the simpler \emph{hinge flip Markov chain} is defined on a strictly larger state space (and even for a more general model than JDM). This also gives rise to the corresponding strong stability definition. As the whole analysis of this auxiliary chain takes place in the appendix, we defer any further definitions there.

%%%%%%%%%%%%%%%%%%%%%%%%%%%%%%%%%%%%%%%%%%%%%%%%%%%%%%%%%%%%%%%%%%%%%%%%
%%%%%%%%%%%%%%%%%%%%%%%%%%%%%%%%%%%%%%%%%%%%%%%%%%%%%%%%%%%%%%%%%%%%%%%%
\section{Sampling Undirected Graphs}\label{sec:main_result}
%%%%%%%%%%%%%%%%%%%%%%%%%%%%%%%%%%%%%%%%%%%%%%%%%%%%%%%%%%%%%%%%%%%%%%%%
%%%%%%%%%%%%%%%%%%%%%%%%%%%%%%%%%%%%%%%%%%%%%%%%%%%%%%%%%%%%%%%%%%%%%%%%
The result in Theorem \ref{thm:switch} below is our main result regarding the mixing time of the switch chain for strongly stable degree sequences. Its proof is divided in two parts. 
First, in Section \ref{sec:js_flow}, by giving an efficient multicommodity flow, we show that for any $d$ in a family of strongly stable degree sequences the JS chain is rapidly mixing on $\mathcal{G}'(d)$. Then, in Section \ref{sec:js_to_switch}, we show that such an efficient  flow for the JS chain  on $\mathcal{G}'(d)$ can be transformed into an efficient flow for the switch chain on $\mathcal{G}(d)$. This  yields the following theorem. 

\begin{theorem}\label{thm:switch}
	Let $\mathcal{D}$ be a strongly stable family of degree sequences with respect to some constant $k$. Then there exists a polynomial $q(n)$ such that, for any $0 < \epsilon < 1$, the mixing time $\tau_{\mathrm{sw}}$ of the switch chain for a graphical sequence $d = (d_1,\dots,d_n) \in \mathcal{D}$ satisfies
	\[
	\tau_{\mathrm{sw}}(\epsilon) \leq q(n)^k \ln(1/\epsilon) \,.
	\]
\end{theorem}

We  discuss two direct corollaries of Theorem \ref{thm:switch}. 
Both corollaries are consequences of the corresponding results in \cite{Jerrum1989graphical}, where it is shown that the families of sequences satisfying \eqref{eq:stable1} and \eqref{eq:stable2}, respectively, are (strongly) stable. We work with a slightly different definition of stability here than the one used in \cite{Jerrum1989graphical}. The reason why the results from \cite{Jerrum1989graphical} carry over to the definition used here (which was introduced in \cite{Jerrum1990}) is explained in Appendix \ref{app:open_question}. The equivalence of these definitions of $P$-stability is also claimed in \cite{Jerrum1989graphical}.

Corollary \ref{lem:stable_jerrum} below extends the rapid mixing results in \cite{Cooper2007,Greenhill2017journal}. In particular, the condition of \cite{Greenhill2017journal} for rapid mixing is $\delta\ge 1$ and $3\le \Delta \le \frac{1}{3}\sqrt{2m}$, which is a special case of the condition \eqref{eq:stable1} below.\footnote{The condition in Theorem 1.1 in \cite{Greenhill2017journal} is a special case of the condition of Theorem 3.1 in \cite{Jerrum1990} which in turn is a special case of the condition of Corollary \ref{lem:stable_jerrum}. See also the remark after Theorem 8 in \cite{Jerrum1989graphical}.} 

\begin{corollary}\label{lem:stable_jerrum}
	Let $\mathcal{D} = \mathcal{D}(\delta,\Delta,m)$ be the set of all graphical degree sequences $d = (d_1,\dots,d_n)$ satisfying 
	\begin{equation}\label{eq:stable1}
	(2m - n \delta)(n\Delta - 2m) \leq (\Delta - \delta)\big[(2m - n\delta)(n - \Delta - 1) + (n\Delta - 2m)\delta\big]
	\end{equation}
	where $\delta$ and $\Delta$ are the minimum and maximum component of $d$, respectively, and $m = \frac{1}{2} \sum_{i = 1}^n d_i$. For any $d \in \mathcal{D}$, we have
	$
	k_{JS}(d) \leq 6.
	$
	Hence, the switch chain is rapidly mixing for sequences in $\mathcal{D}$. 
\end{corollary}

The next corollary is a special case of the  Corollary \ref{lem:stable_jerrum} and answers an open question posed in \cite{Greenhill2017journal}. It is a result of similar flavor, but the corresponding condition is stated only in terms of $\delta$ and $\Delta$.

\begin{corollary}\label{cor:stable}
	Let $\mathcal{D} = \mathcal{D}(\delta,\Delta)$ be the set of all graphical degree sequences $d = (d_1,\dots,d_n)$ satisfying 
	\begin{equation}\label{eq:stable2}
	(\Delta - \delta + 1)^2 \leq 4\delta(n - \Delta - 1)
	\end{equation}
	where $\delta$ and $\Delta$ are the minimum and maximum component of $d$, respectively. For any $d \in \mathcal{D}$, we have 
	$
	k_{JS}(d)  \leq 6.
	$
	Hence, the switch chain is rapidly mixing for sequences in $\mathcal{D}$.
\end{corollary}

Explicit families satisfying these conditions are given in \cite{Jerrum1989graphical}. For instance, all sequences $d = (d_1,\dots,d_n)$  with (i) $\delta(d) \geq 1$ and $\Delta(d) \leq 2\sqrt{n}-2$, or (ii) $\delta(d) \geq \frac{1}{4}n$ and $\Delta(d) \leq \frac{3}{4}n - 1$ satisfy \eqref{eq:stable2} but not necessarily the conditions in \cite{Greenhill2017journal}. We refer the reader to \cite{Jerrum1989graphical,Salas2016} for more examples.  
The bounds in Corollaries \ref{lem:stable_jerrum} and \ref{cor:stable}  are in a sense  best possible with respect to the graph parameters involved. Namely, there exist non-stable degree sequence families the members of which only slightly violate \eqref{eq:stable2}; see  the discussion in \cite{Jerrum1989graphical} for details.

%%%%%%%%%%%%%%%%%%%%%%%%%%%%%%%%%%%%%%%%%%%%%%%%%%%%%%%%%%%%%%%%%%%%%%%%
\subsection{Flow for the Jerrum-Sinclair Chain}\label{sec:js_flow}
%%%%%%%%%%%%%%%%%%%%%%%%%%%%%%%%%%%%%%%%%%%%%%%%%%%%%%%%%%%%%%%%%%%%%%%%
Jerrum and Sinclair \cite{Jerrum1990} claim, without proof,
that the JS chain is rapidly mixing for (some) families of stable degree sequences.  
For completeness, we prove in Theorem \ref{thm:js_mixing} that the chain is rapidly mixing for any family of strongly stable degree sequences. For the proof of the theorem see Appendix \ref{app:js}.

\begin{theorem}[\cite{Jerrum1990}]\label{thm:js_mixing}
	Let $\mathcal{D}$ be a strongly stable family of degree sequences with respect to some constant $k$. Then there exist polynomials $p(n)$ and $r(n)$ such that for any $d = (d_1,\dots,d_n)\in \mathcal{D}$ there exists an efficient multicommodity flow $f$ for the JS chain on $\mathcal{G}'(d)$   	
	satisfying $\max_e f(e) \leq p(n)/ |\mathcal{G}'(d)|$  and  $\ell(f) \leq r(n)$.
\end{theorem}

Our proof of Theorem \ref{thm:js_mixing} uses conceptually similar arguments to the ones used in \cite{Cooper2007} for the analysis of the switch chain on regular undirected graphs. However, the analysis done here  for the JS chain is, in our opinion, easier and cleaner than the corresponding analysis for the switch chain. In particular, the so-called \emph{circuit processing procedure} is simpler in our setting, as it only involves altering edges in the symmetric difference of two graphical realizations in a straightforward fashion. In the switch chain analyses  \cite{Cooper2007,Greenhill2017journal,Erdos2013,Erdos2015decomposition,ErdosMMS2018}  one also has to temporarily alter edges that are not in the symmetric difference and this significantly complicates things. 
Moreover, for the analysis of the JS chain, we can rely on arguments used (in a somewhat different context) by Jerrum and Sinclair \cite{Jerrum1989} for the analysis of a Markov chain for  sampling (near) perfect matchings of a given graph. This usage of arguments in \cite{Jerrum1989} was suggested by Jerrum and Sinclair \cite{Jerrum1990} for showing that the JS chain is rapidly mixing for stable degree sequences.

%%%%%%%%%%%%%%%%%%%%%%%%%%%%%%%%%%%%%%%%%%%%%%%%%%%%%%%%%%%%%%%%%%%%%%%%
\subsection{Flow Transformation}\label{sec:js_to_switch}
%%%%%%%%%%%%%%%%%%%%%%%%%%%%%%%%%%%%%%%%%%%%%%%%%%%%%%%%%%%%%%%%%%%%%%%%
Next we show that, when $d$ comes from a family of strongly stable degree sequences, an efficient multicommodity flow for the JS chain on $\mathcal{G}'(d)$ can be transformed into an efficient multicommodity flow for the switch chain on $\mathcal{G}(d)$. In combination with Theorem \ref{thm:js_mixing} this implies that if $\mathcal{D}$ is  strongly stable, then for any sequence in $\mathcal{D}$ there exists an efficient flow for the switch chain. For the sake of simplicity, we did not attempt to optimize the bounds in the proof of Theorem \ref{thm:transformation}.

\begin{theorem}%[Flow transformation]
	\label{thm:transformation}
	Let $\mathcal{D}$ be a strongly stable family of degree sequences with respect to some constant $k$, and let    $p(n)$ and $r(n)$ be polynomials such that for any $d = (d_1,\dots,d_n)\in \mathcal{D}$ there exists an efficient multicommodity flow $f_d$ for the JS chain on $\mathcal{G}'(d)$ with the property that $max_{e} f(e) \leq p(n) / |\mathcal{G}'(d)|$ and $\ell(f) \leq r(n)$.

	Then there exists a polynomial $t(n)$ such that for all $d = (d_1,\dots,d_n) \in \mathcal{D}$ there is a feasible multicommodity flow $g_d$ for the switch chain on $\mathcal{G}(d)$ 
	with \emph{(i)} $\ell(g_d) \leq 2k \cdot \ell(f_d)$, and \emph{(ii)} for every edge $e$ of the state space graph of the switch chain, we have 
\vspace{-10pt}
	\begin{equation}\label{eq:good_bound}
	g_d(e) \leq t(n)^k \cdot \frac{p(n)}{|\mathcal{G}(d)|} \,.
	\end{equation}
\end{theorem}
\begin{proof}
	Let $d \in \mathcal{D}$. For simplicity we will write $f$ and $g$ instead of $f_d$ and $g_d$ respectively.
	Since there are two Markov chains involved in the proof, each with a different state space graph, we should clarify that $\mathcal{P}_{xy}$ refers to the set of simple paths between $x$ and $y$ in the state space graph of the \textit{JS chain}.
	We first introduce some additional notation. 
	
	For every pair $(x,y) \in \mathcal{G}'(d) \times \mathcal{G}'(d)$ with $x \neq y$, and for any $p \in \mathcal{P}_{xy}$, we write 
	$\alpha(p)= f(p)|\mathcal{G}'(d)|^2$.
%	\[
%	\alpha(p)= f(p)|\mathcal{G}'(d)|^2 \,.
%	\]
	Recall that since the stationary distribution of the JS chain is uniform on $\mathcal{G}'(d)$ we have $\sum_{p \in \mathcal{P}_{xy}} f(p) = |\mathcal{G}'(d)|^{-2}$.  Thus, $\sum_{p \in \mathcal{P}_{xy}} \alpha(p) = 1$. Moreover, we define $\alpha(e) = \sum_{p \in \mathcal{P}_{xy} : e \in p}  \alpha(p) = f(e) |\mathcal{G}'(d)|^2$. 
	
	Now, for every $G \in \mathcal{G}'(d) \mysetminus \mathcal{G}(d)$ choose some $\varphi(G) \in \mathcal{G}(d)$ that is within distance $k$ of $G$ in the JS chain, and take $\varphi(G) = G$ for $G \in \mathcal{G}(d)$. Based on the arguments in the proof of Proposition \ref{prop:strong_stable}, it follows that for any $H \in \mathcal{G}(d)$,
	\begin{equation}\label{eq:neighborhood_bound}
	|\varphi^{-1}(H)| \leq n^{3k} \,,			
	\end{equation}						
	using that the maximum in-degree of any element in the state 				
	space graph of the JS chain is upper bounded by $n^3$. 
	In particular, this implies that
	\begin{equation}\label{eq:ratio}
	\frac{|\mathcal{G}'(d)|}{|\mathcal{G}(d)|} \leq n^{3k} \,.
	\end{equation}
	Let the flow $h$ be defined as follows for any given pair $(x,y)$. If $(x,y) \in \mathcal{G}(d) \times \mathcal{G}(d)$, take 
	$h(p) = \alpha(p)/|\mathcal{G}(d)|^2$
%	\[
%	h(p) = \alpha(p)/|\mathcal{G}(d)|^2
%	\]
	for all $p \in \mathcal{P}_{xy}$. If either $x$ or $y$ is not contained in $\mathcal{G}(d)$, take $h(p) = 0$ for every $p \in \mathcal{P}_{xy}$. Note that $h$ is a multicommodity flow that routes $1/|\mathcal{G}(d)|^2$ units of flow between any pair $(x,y) \in \mathcal{G}(d) \times \mathcal{G}(d)$, and zero units of flow between any other pair of states in $\mathcal{G}'(d)$.
	
	Note that
	\begin{equation}\label{eq:traverse1}
	h(e) \le \frac{|\mathcal{G}'(d)|^2}{|\mathcal{G}(d)|^2} \cdot f(e) \leq \frac{|\mathcal{G}'(d)|^2}{|\mathcal{G}(d)|^2} \frac{p(n)}{|\mathcal{G}'(d)|}  = \frac{p(n)}{|\mathcal{G}(d)|} \frac{|\mathcal{G}'(d)|}{|\mathcal{G}(d)|}  \leq  n^{3k} \cdot \frac{p(n)}{|\mathcal{G}(d)|} \,,
	\end{equation}
	using the definition of $h$ in the first inequality, the assumption on $f$ in the second inequality, and the upper bound of (\ref{eq:ratio}) in the last one. 
	
	Next, we merge the ``auxiliary states'' in $\mathcal{G}'(d) \mysetminus \mathcal{G}(d)$, i.e., the states not reached by the switch chain, with the elements of $\mathcal{G}(d)$. Informally speaking, for every $H \in \mathcal{G}(d)$ we merge all the nodes in $\varphi^{-1}(H)$ into a \emph{supernode}. Self-loops created in this process are removed, and parallel arcs between states are merged into one arc that gets all the flow of the parallel arcs. Formally, we consider the graph $\Gamma$ where $V(\Gamma) = \mathcal{G}(d)$ and $e = (H,H') \in E(\Gamma)$ if and only if $H$ and $H'$ are switch adjacent or if there exist $G \in \varphi^{-1}(H)$ and $G' \in \varphi^{-1}(H')$ such that $G$ and $G'$ are JS adjacent.  Moreover, for a given $h$-flow carrying path $(G_1,G_2,\ldots,G_q) = p \in \mathcal{P}_{xy}$, let 
	$p'_{\Gamma} = (\varphi(G_1),\varphi(G_2),\dots,\varphi(G_q))$
%	\[
%	p'_{\Gamma} = (\varphi(G_1),\varphi(G_2),\dots,\varphi(G_q))
%	\]
	be the corresponding (possibly non-simple) directed path in $\Gamma$. Any self-loops and cycles can be removed from $p'_{\Gamma}$ and let $p_{\Gamma}$ be the resulting simple path in $\Gamma$. Over $p_{\Gamma}$ we route $h_{\Gamma}(p_{\Gamma})= h(p)$ units of flow.  Note that $h_{\Gamma}$ is a flow that routes $1/|\mathcal{G}(d)|^2$ 			
	units of flow between any pair of states $(x,y) \in \mathcal{G}(d) \times \mathcal{G}(d)$ in the graph $\Gamma$ and that $\ell(h_{\Gamma}) \leq \ell(f)$. Furthermore, the flow $h_{\Gamma}$ on an edge $(H, H') \in E(\Gamma)$ is then bounded by
	\begin{equation}\label{eq:up_bound_h_Gamma}
	h_{\Gamma}(H,H') \le \!\! \sum_{\substack{(G,G') \in \varphi^{-1}(H) \times \varphi^{-1}(H') \\ G \text{ and } G' \text{ are JS adjacent}}} \!\! h(G,G') \,,
	\end{equation}
	where the inequality (instead of an equality) follows from the fact that when we map a path $p \in \mathcal{P}_{xy}$ to the corresponding path $p_{\Gamma}$, some edges of the intermediate path $p'_{\Gamma}$ may be deleted.
	Using (\ref{eq:neighborhood_bound}), it follows that 
	$|\varphi^{-1}(H) \times \varphi^{-1}(H')| \leq n^{3k} \cdot n^{3k} = n^{6k}$
%	\[
%	|\varphi^{-1}(H) \times \varphi^{-1}(H')| \leq n^{3k} \cdot n^{3k} = n^{6k}
%	\]
	and therefore, in combination with \eqref{eq:traverse1} and \eqref{eq:up_bound_h_Gamma}, we have that
	\begin{equation}\label{eq:traverse2}
	h_{\Gamma}(e) \leq n^{3k} \cdot  n^{6k} \cdot \frac{p(n)}{|\mathcal{G}(d)|} \,.
	\end{equation}
	
	Now recall how $E(\Gamma)$ was defined. An edge $(H,H')$ might have been added because: \emph{(i)} $H$ and $H'$ are switch adjacent (we call these edges of $\Gamma$ \emph{legal}), or \emph{(ii)} $H$ and $H'$ are \emph{not} switch adjacent but there exist $G \in \varphi^{-1}(H)$ and $G' \in \varphi^{-1}(H')$ that are JS adjacent (we call these edges of $\Gamma$ \emph{illegal}).
	The final step of the proof consists of showing that the flow on every illegal edge in $E(\Gamma)$ can be rerouted over a ``short'' path consisting only of legal edges. In particular, for every flow carrying path $p$ using $e$, we are going to show that the flow $h_{\Gamma}(p)$ is rerouted over some legal detour, the length of which is bounded by a multiple of $k$. Doing this iteratively for every remaining  illegal edge on $p$, we obtain a directed path $p''$ only using legal edges, i.e., edges of the state space graph of the switch chain. Of course, $p''$ might not be simple, so any self-loops and cycles can be removed, as before, to obtain the simple legal path $p'$. Figure \ref{fig:rerouting} illustrates this procedure for a path with a single illegal edge. Note that deleting self-loops and cycles only decreases the amount of flow on an edge.

	\begin{figure}[t!]
		\centering
		\scalebox{0.6}{
			\begin{tikzpicture}[
			->,
			>=stealth',
			shorten >=0.5pt,
			auto,
			node distance=1cm,
			semithick,
			every state/.style={circle,text=black,inner sep=5pt,minimum size=1pt},
			]
			\begin{scope}
			\node[state]  (M1)               					 {x};
			\node[state]  (M2) [right=2cm of M1] 				 {};
			\node[state]  (M3) [right=2cm of M2]  			 {};
			\node[state]  (M4) [right=2cm of M3]			     {};
			\node[state]  (M5) [right=2cm of M4] 					{y};
			\node[state]  (T1) [above=1.5cm of M2] 				{};
			\node[state]  (T2) [above=1.5cm of M3]				 {};
			\node[state]  (B1) [below=1.5cm of M2]				 {};       
			\node[state]  (B2) [below=1.5cm of M3] 				 {}; 
			
			\path[every node/.style={sloped,anchor=south,auto=false}]
			(M1) edge[line width=1.5pt] 	node {} (M2)
			(M2) edge[line width=1.5pt] 	node {} (M3)
			(M3) edge[dashed,line width=1.5pt] 	node {} (M4)
			(M4) edge[line width=1.5pt] 	node {} (M5)
			(M3) edge[line width=4pt] 	node {} (B2)
			(B2) edge[line width=4pt] 	node {} (B1)
			(B1) edge[line width=4pt] 	node {} (M2)
			(M2) edge[line width=4pt] 	node {} (T1)
			(T1) edge[line width=4pt] 	node {} (T2)
			(T2) edge[line width=4pt] 	node {} (M4);
			
			\end{scope}
			\end{tikzpicture}
			\ \ \ \ \ \ \ \ \ 
			\quad
			\ \ \ \ \ \ \ \ \ 
			\begin{tikzpicture}[
			->,
			>=stealth',
			shorten >=0.5pt,
			auto,
			node distance=1cm,
			semithick,
			every state/.style={circle,text=black,inner sep=5pt,minimum size=1pt},
			]
			\begin{scope}
			\node[state]  (M1)               					 {x};
			\node[state]  (M2) [right=2cm of M1] 				 {};
			\node[state]  (M3) [right=2cm of M2]  			 {};
			\node[state]  (M4) [right=2cm of M3]			     {};
			\node[state]  (M5) [right=2cm of M4] 					{y};
			\node[state]  (T1) [above=1.5cm of M2] 				{};
			\node[state]  (T2) [above=1.5cm of M3]				 {};
			\node[state]  (B1) [below=1.5cm of M2]				 {};       
			\node[state]  (B2) [below=1.5cm of M3] 				 {}; 
			
			\path[every node/.style={sloped,anchor=south,auto=false}]
			(M1) edge[line width=2pt] 	node {} (M2)
			%(M2) edge[dashed,line width=1.5pt] 	node {} (M3)
			%(M3) edge[line width=1.5pt] 	node {} (M4)
			(M4) edge[line width=2pt] 	node {} (M5)
			%(M3) edge[line width=4pt] 	node {} (B2)
			%(B2) edge[line width=4pt] 	node {} (B1)
			%(B1) edge[line width=4pt] 	node {} (M2)
			(M2) edge[line width=2pt] 	node {} (T1)
			(T1) edge[line width=2pt] 	node {} (T2)
			(T2) edge[line width=2pt] 	node {} (M4);
			
			\end{scope}
			\end{tikzpicture}}
		\caption{The dashed edge on the left represents an illegal edge, and the bold path represents a ``short'' detour. The shortcutted path on the right is the result of removing any loops and cycles.}
		\label{fig:rerouting}
	\end{figure}
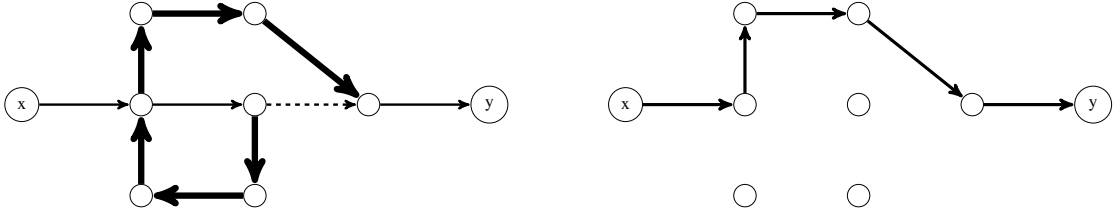

	The crucial observation here is that if $(H,H') \in E(\Gamma)$, then 
	$|E(H) \triangle  E(H')| \leq 4k$.
%	\[
%	|E(H) \triangle  E(H')| \leq 4k  \,.
%	\]
	That is, even though $H$ and $H'$ might not be switch adjacent, they are not too far apart. 
	To see this, first note that the symmetric difference of any two JS adjacent graphs has size at most 2. Moreover, if one of any two JS adjacent graphs is in $\mathcal{G}(d)$, then  their symmetric difference has size 1.
	In particular, for any $G^* \in \mathcal{G}'(d)$, we have $|E(G^*) \triangle E(\varphi(G^*))|\le 2k - 1$.
	
	Clearly, if $(H,H') \in E(\Gamma)$ is legal, then $|E(H) \triangle  E(H')| = 4 \leq 4k$. 
	Assume $(H,H') \in E(\Gamma)$ is illegal. Then there exist JS adjacent $G \in \varphi^{-1}(H)$ and $G'\in \varphi^{-1}(H')$ and according to the above we have 
	\begin{eqnarray*}
	|E(H) \triangle E(H')| & \leq & |E(H) \triangle E(G)| + |E(G) \triangle E(G')| + |E(G') \triangle E(H')|\nonumber \\ 
	& \leq & 2k-1 + 2 + 2k - 1 \leq 4k \,. \nonumber
	\end{eqnarray*}
	Moreover, this implies that we can go from $H$ to $H'$ in a ``small'' number of moves in the switch chain. This easily follows from most results showing that the state space of the switch chain is connected, e.g., from \cite{Taylor81}.\footnote{To be precise, we can focus on the subgraph induced by the nodes with positive degree in the symmetric difference. Taylor's proof on the connectivity of the state space of the switch chain \cite{Taylor81} implies that we can find $O(k^2)$ switches to get from $H$ to $H'$, only using edges in this induced subgraph.}  Specifically, here we use the following result of Erd{\H{o}}s, Kir{\'a}ly, and Mikl{\'o}s \cite{Erdos2013swap} which implies that we can go from $H$ to $H'$ in $2k$ switches. %\pkrem{Proof? Check reference}
	\begin{theorem}[follows from Theorem 3.6 in \cite{Erdos2013swap}] \label{thm:swap}
		Let $d= (d_1,\dots,d_n)$ be a degree sequence. For any two graphs $H, H' \in \mathcal{G}(d)$, $H$ can be transformed into $H'$ using at most $\frac{1}{2}|E(H) \triangle E(H')|$ switches.
	\end{theorem}
	
	For every \emph{illegal} edge $e \in E(\Gamma)$, we choose such a (simple) path from $H$ to $H'$ with at most $2k$ transitions and reroute the flow of $e$ over this path. Note that for any legal edge $e \in E(\Gamma)$, the number of illegal edge detours that use $e$ for this rerouting procedure, is at most $( n^4)^{2 k}\cdot ( n^4)^{2 k} =  n^{16 k}$,
%	\[
%	\left( n^4\right)^{2 k}\cdot \left( n^4\right)^{2 k} =  n^{16 k}  \,,
%	\]
	using the fact that in the state space graph of the switch chain the maximum degree of an element  is at most $n^4$  and  any illegal edge using  $e$ in its rerouting procedure must lie within distance $2 k$ of $e$. Combining this with (\ref{eq:traverse2}), we see that the resulting flow, $g$, satisfies 
	\begin{equation*}\label{eq:traverse3}
	g(e) \leq \frac{p(n) \cdot n^{9k} + p(n)  \cdot n^{16 k}}{|\mathcal{G}(d)|}  \,.
	\end{equation*}
	
	Note that the $\ell(g) \leq 2k \ell(h_{\Gamma})$. This holds because every illegal edge on a flow-carrying path gives rise to at most $2k$ additional edges as a result of rerouting the flow over legal edges, and the removal of loops and cycles from any resulting non-simple path can only decrease its length. Combining this inequality with $\ell(h_{\Gamma}) \leq \ell(f)$ (as we noted above), we get $\ell(g) \leq 2k \cdot \ell(f)$.
	This completes the proof of (\ref{eq:good_bound}), as we have now constructed a feasible multicommodity flow $g$ in the state space graph of the switch chain with the desired properties.  
\end{proof}

%\begin{remark}[Almost strong stability]
%	The result in Theorem \ref{thm:transformation} remains true under the slightly 
%	weaker condition that a family of degree sequences is \emph{almost strongly stable}, 
%	meaning that there exists a constant $k$ and a polynomial $q(\cdot)$, so that for every $d \in \mathcal{D}$ 
%	with $n$ components there exists a set $S$ with $|S| \leq q(n)$ such that
%	\[
%	\max_{G \in \mathcal{G}'(d) \mysetminus S} \text{dist}(G,d) \leq k \,.
%	\]
%\end{remark}

%%%%%%%%%%%%%%%%%%%%%%%%%%%%%%%%%%%%%%%%%%%%%%%%%%%%%%%%%%%%%%%%%%%%%%%%
%%%%%%%%%%%%%%%%%%%%%%%%%%%%%%%%%%%%%%%%%%%%%%%%%%%%%%%%%%%%%%%%%%%%%%%%
\section{Sampling Graphs with a Given JDM}\label{sec:switch_jdm_main}
%%%%%%%%%%%%%%%%%%%%%%%%%%%%%%%%%%%%%%%%%%%%%%%%%%%%%%%%%%%%%%%%%%%%%%%%
%%%%%%%%%%%%%%%%%%%%%%%%%%%%%%%%%%%%%%%%%%%%%%%%%%%%%%%%%%%%%%%%%%%%%%%%
We may use a similar high level approach to that in Section \ref{sec:main_result}
to show that the (restricted) switch chain defined in Subsection \ref{sec:jdm_model} is \emph{always} rapidly mixing for JDM instances with two degree classes. %\yarem{maybe add some text about this being the first result of the sort}

\begin{theorem}\label{thm:stable_jdm1}
	Let $\mathcal{D}$ be the family of instances of the joint degree matrix model with two degree classes. Then the switch chain is rapidly mixing for  instances in $\mathcal{D}$.
\end{theorem}

In analogy to the JS chain we first analyze a simpler Markov chain, called the \emph{hinge flip chain}, that adds and removes (at most) one edge at a time. Very much like the JS chain, the hinge flip chain might slightly violate the degree constraints.  Now, however, the joint degree constraints might be violated as well. The definition of \emph{strong stability} is appropriately adjusted to account for both deviations from the original requirements. Finally, we use a similar embedding argument as in Theorem \ref{thm:transformation}. The relevant definitions, as well as the analysis of this auxiliary chain are deferred to Appendix \ref{sec:auxiliary} due to space constraints. Here we present a high level outline of the proof of Theorem \ref{thm:stable_jdm1}. 

\paragraph{Rapid Mixing of the Hinge Flip Chain.} The first step of the proof is to show that the hinge flip chain defined on a \emph{strict superset} of the target state space mixes rapidly for strongly stable instances. Appendix \ref{sec:auxiliary} is dedicated to this  step. The fact that we do not want to deviate by more than a constant from the joint degree constraints % (in particular, we always require that the number of edges across the two degree classes differs at most 1 from $c_{12}$) 
makes the analysis much more challenging than the one for the JS chain presented in Appendix \ref{app:js}.
In order to overcome the difficulties that arise due to this fact, %the number of edges across the two degree classes being almost constant, 
we rely on ideas introduced by Bhatnagar et al.~\cite{Bhatnagar2008} for uniformly sampling  bichromatic matchings.  % in a graph that has its edges partitioned into two color classes. 
In particular, in the circuit processing part of the proof, we process a circuit at multiple places \emph{simultaneously} in case there is only one circuit in the canonical decomposition of a pairing, or we process multiple circuits \emph{simultaneously} in case the decomposition yields multiple circuits.  At the core of this approach lies a variant of the \emph{mountain-climbing problem} \cite{Homma1952,Whittaker1966}. In our case the analysis is more involved than that of \cite{Bhatnagar2008}, and we therefore use different arguments in various parts of the proof.

It is interesting to note that the analysis of the hinge flip chain is not carried out in the JDM model but in the more general Partition Adjacency Matrix (PAM) model \cite{Czabarka14,ErdosHIM2017}. The difference from the JDM model is that in each class $V_i$ the nodes need not have the same constant degree but rather follow a given degree sequence of size $|V_i|$. Given that small deviations from the prescribed degrees cannot be directly handled---by definition---by the JDM model, the PAM model is indeed a more natural choice for this step.

\paragraph{Strong Stability of JDM Instances.} Next we show that for any JDM instance, any graph in the state space of the hinge flip chain (i.e., graphs that satisfy or \emph{almost} satisfy the joint degree requirements) can be transformed to a graphical realization of the original instance within 6 hinge flips at most. That is, the set of JDM instances is a strongly stable family of instances of the PAM model and thus the hinge flip chain mixes rapidly for JDM instances. See Theorem \ref{cor:stable_jdm} in Appendix \ref{sec:stable_jdm}.

\paragraph{Flow Transformation.} The final step is an embedding argument, along the lines of the argument of Subsection \ref{sec:js_to_switch}, for transforming the efficient flow for the hinge flip chain to an efficient flow for the switch chain. As an intermediate step we need an analog of Theorem \ref{thm:swap}, but this directly follows from the proof of irreducibility of the switch chain in \cite{Amanatidis2015}. See Appendix \ref{sec:switch_for_jdm}.

%\color{red}
%%%%%%%%%%%%%%%%%%%%%%%%%%%%%%%%%%%%%%%%%%%%%%%%%%%%%%%%%%%%%%%%%%%%%%%%
%%%%%%%%%%%%%%%%%%%%%%%%%%%%%%%%%%%%%%%%%%%%%%%%%%%%%%%%%%%%%%%%%%%%%%%%
\section{Discussion}
%%%%%%%%%%%%%%%%%%%%%%%%%%%%%%%%%%%%%%%%%%%%%%%%%%%%%%%%%%%%%%%%%%%%%%%%
%%%%%%%%%%%%%%%%%%%%%%%%%%%%%%%%%%%%%%%%%%%%%%%%%%%%%%%%%%%%%%%%%%%%%%%%
We believe that our ideas can  be also used to simplify the switch chain analyses in settings where there is some given forbidden edge set, the elements of which cannot be used in any (bipartite) graphical realization \cite{Greenhill2011,Greenhill2017journal,Erdos2015,ErdosMMS2018}. This is an interesting direction for future work, as it  captures the case of sampling directed graphs.

Moreover, we suspect it can be shown that the condition in (\ref{eq:bip_same}) is essentially best possible in terms of $\delta_{r}$ and $\Delta_{r}$ in a similar sense as described in \cite{Jerrum1989graphical} for the results in Corollaries \ref{lem:stable_jerrum}  and  \ref{cor:stable}.\footnote{We suggest using the same ideas as the proof of Theorem 6 in \cite{Jerrum1989graphical} based on a non-stable family of bipartite degree sequences presented in \cite{Kannan1999}.} While this is an interesting question,  our goal here is not to give a full bipartite analogue of \cite{Jerrum1989graphical}. Even so, a deeper understanding of when a family of bipartite degree sequences is strongly stable is missing. 
In particular, is it possible to unify the results of Corollaries \ref{cor:bipartite} and \ref{cor:bipartite2} under a single condition similar to (\ref{eq:stable1})?

Further, it is not clear whether there exist  degree sequence families---bipartite or not---that are $P$-stable but not strongly stable. For instance, in a recent work by Gao and Wormald \cite{GaoW18}, who provide a very efficient non-MCMC approximate sampler for certain power-law degree sequences,  it is argued that these power-law degree sequences are $P$-stable. Is it the case these sequences are strongly stable as well? Theorem \ref{thm:switch} would then directly imply that the switch chain is rapidly mixing for this family.

A central open question is how to go beyond (strong) stability. We suspect that the proof template of \cite{Cooper2007} cannot be used for proving rapid mixing of the switch chain for general families of degree sequences. The intuition is that it relies on the fact that there is a set of auxiliary states that is not much larger than the set of actual graphical realizations for a given degree sequence; this property seems very closely related to $P$-stability, and also arises explicitly in the analysis of the bipartite case in \cite{Kannan1999}. This observation suggests the need for a novel approach for studying the mixing of the switch chain on non-stable degree families. 

Finally, the problem of sampling graphic realizations of a given joint degree distribution with three or more degree classes is also open. Although our proof breaks down for more than two classes, we hope that our high level approach can facilitate progress on the problem.

%
%\color{black}

\section*{Acknowledgements} We are grateful to P{\'{e}}ter Erd{\H{o}}s, Tam{\'{a}}s Mezei and Istv{\'{a}}n Mikl{\'{o}}s for their useful comments.

\newpage
\appendix

%%%%%%%%%%%%%%%%%%%%%%%%%%%%%%%%%%%%%%%%%%%%%%%%%%%%%%%%%%%%%%%%%%%%%%%%
%%%%%%%%%%%%%%%%%%%%%%%%%%%%%%%%%%%%%%%%%%%%%%%%%%%%%%%%%%%%%%%%%%%%%%%%
\section{Missing Material from Section \ref{sec:main_result}}
%%%%%%%%%%%%%%%%%%%%%%%%%%%%%%%%%%%%%%%%%%%%%%%%%%%%%%%%%%%%%%%%%%%%%%%%
%%%%%%%%%%%%%%%%%%%%%%%%%%%%%%%%%%%%%%%%%%%%%%%%%%%%%%%%%%%%%%%%%%%%%%%%

%%%%%%%%%%%%%%%%%%%%%%%%%%%%%%%%%%%%%%%%%%%%%%%%%%%%%%%%%%%%%%%%%%%%%%%%
\subsection{On the Proofs of Corollaries \ref{lem:stable_jerrum} and \ref{cor:stable}}\label{app:open_question}
%%%%%%%%%%%%%%%%%%%%%%%%%%%%%%%%%%%%%%%%%%%%%%%%%%%%%%%%%%%%%%%%%%%%%%%%

Since a slightly different notion of stability (used in \cite{Jerrum1989graphical}) is defined below, to avoid confusion,  we should clarify that whenever \emph{strong} stability is mentioned it is meant in the sense of Definition \ref{def:strong_stable}.

We first introduce some notation, using the same terminology as in \cite{Jerrum1989graphical}. Let $G = (V,E)$ be an undirected graph. For distinct $u,v \in V$ we say that $u,v$ are co-adjacent if $\{u,v\} \notin E$, and $\{u,v\}$ is called a co-edge. An alternating path of length $q$ in $G$ is a sequence of nodes $v_0,v_1,\dots,v_q$ such that (i) $\{v_i,v_{i+1}\}$ is an edge when $i$ is even, and a co-edge if $i$ is odd, or (ii) $\{v_i,v_{i+1}\}$ is a co-edge when $i$ is even, and an edge if $i$ is odd. The path is called a cycle if $v_0 = v_q$. We will always specify if a path is of type $(i)$ or $(ii)$.

The definition of $P$-stability in \cite{Jerrum1989graphical}---which we here call $\pm P$-stability---is based on the following definition. Let $\mathcal{G}''(d) = \cup_{d'} \mathcal{G}(d')$ with $d'$ ranging over the set
\[
\{d\} \cup \bigg\{ d' : \sum_{i=1}^n d_i = \sum_{i=1}^n d_i' \ \text{ and } \ \sum_{i=1}^n |d_i - d_i'| = 2\bigg\} \,.
\]
That is, either $d = d'$, or $d'$ has one node with deficit one and one node with surplus one. A family of sequences is $\pm P$-stable if there exists a polynomial $p$ such that $|\mathcal{G}''(d)|/|\mathcal{G}(d)| \leq p(n)$ for every $d = (d_1,\dots,d_n)$ in the family. In order to show  in \cite{Jerrum1989graphical} that the families corresponding to the conditions in Corollaries \ref{lem:stable_jerrum} and \ref{cor:stable} are $\pm P$-stable, it is shown that for any $G \in \mathcal{G}''(d)$ there always exists a \emph{short alternating path} of even length, starting with an edge and ending with a co-edge, connecting the node with surplus one to the node with deficit one. The connection with our definition of strong stability is as follows.

\begin{lemma}\label{lem:app_A}
	Let $\mathcal{D}$ be a family of sequences, and assume that there exists a constant $k_0$ such that for all $d \in \mathcal{D}$ and any $G \in \mathcal{G}''(d) \setminus \mathcal{G}(d)$, there exists an alternating path of length at most $k_0$, starting with an edge and ending with a co-edge, connecting the node with surplus one to the node with deficit one.
	Then the family $\mathcal{D}$ is strongly stable with respect to $k_0/2 + 1$.
\end{lemma}
\begin{proof}
	Let $H \in \mathcal{G}'(d)$ be such that there are two nodes $u \neq v$ with deficit one (the case of one node with deficit two is very similar). Since $u$ has deficit one, there is some $x$ so that $\{u,x\}$ is a co-edge. Then the graph $H + \{u,x\} \in \mathcal{G}''(d)$ as now $x$ has surplus one, and $v$ still has deficit one (note that if $x = v$ we are immediately done).
	
	The assumption now implies that there is an alternating path of length at most $k_0$ starting at $x$ with an edge, and ending at $v$ with a co-edge. But this implies that in $H$, there exists an alternating path starting at $u$ and ending at $v$ of length at most $k_0 + 1$, where both the first and last edge on this path are co-edges. This certainly implies that $k_{JS}(d) \leq k_0/2 + 1$ by using a sequence of (see beginning of Appendix \ref{app:js} below for definitions) Type 1 transitions, and finally a Type 2 transition.
\end{proof}

The following two corollaries now explain why the families presented in Corollaries \ref{lem:stable_jerrum} and \ref{cor:stable} are strongly stable. The value of $k_0 = 10$ used below follows directly from the proofs of Theorems 8 and 2 in \cite{Jerrum1989graphical}.

\begin{corollary}
	Let $\mathcal{D} = \mathcal{D}(\delta,\Delta,m)$ be the set of all graphical degree sequences $d = (d_1,\dots,d_n)$ satisfying  
	\begin{equation*}
	(2m - n \delta)(n\Delta - 2m) \leq (\Delta - \delta)\big[(2m - n\delta)(n - \Delta - 1) + (n\Delta - 2m)\delta\big] \tag{\ref{eq:stable1}}
	\end{equation*}
	where $\delta$ and $\Delta$ are the minimum and maximum component of $d$, respectively, and $m = \frac{1}{2} \sum_{i = 1}^n d_i$.
	For all $d \in \mathcal{D}$ and any $G \in \mathcal{G}''(d) \setminus \mathcal{G}(d)$, there exists an alternating path of length at most $10$, starting with an edge and ending with a co-edge, connecting the node with surplus one to the node with deficit one. Hence, $\mathcal{D}$ is strongly stable.
\end{corollary}

\begin{corollary}
	Let $\mathcal{D} = \mathcal{D}(\delta,\Delta)$ be the set of all graphical degree sequences $d = (d_1,\dots,d_n)$ satisfying  
	\begin{equation*}
	(\delta_{\max} - \delta_{\min} + 1)^2 \leq 4\delta_{\min}(n - \delta_{\max} - 1)  \tag{\ref{eq:stable2}}
	\end{equation*}
	where $\delta$ and $\Delta$ are the minimum and maximum component of $d$, respectively.
	For all $d \in \mathcal{D}$ and any $G \in \mathcal{G}''(d) \setminus \mathcal{G}(d)$, there exists an alternating path of length at most $10$, starting with an edge and ending with a co-edge, connecting the node with surplus one to the node with deficit one. Hence, $\mathcal{D}$ is strongly stable.
\end{corollary}

%%%%%%%%%%%%%%%%%%%%%%%%%%%%%%%%%%%%%%%%%%%%%%%%%%%%%%%%%%%
\subsection{Proof of Theorem \ref{thm:js_mixing}}\label{app:js}
%%%%%%%%%%%%%%%%%%%%%%%%%%%%%%%%%%%%%%%%%%%%%%%%%%%%%%%%%%%

For the reader's convenience, we repeat the description of the JS chain \cite{Jerrum1990} here (ignoring the lazy part).  We also introduce some shorthand terminology for the type of moves defining the chain: in state $G \in \mathcal{G}'(d)$, select an ordered pair $i,j$ of nodes uniformly at random and then 
\begin{enumerate}[label=(\roman*)]
	\item if $G \in \mathcal{G}(d)$ and $(i,j)$ is an edge of $G$, delete $(i,j)$ from $G$ (\emph{Type 0 transition}),
	\item if $G \notin \mathcal{G}(d)$ and the degree of $i$ in $G$ is less than $d_i$, and $(i,j)$ is not an edge of $G$, add $(i,j)$ to $G$; if this causes the degree of $j$ to exceed $d_j$, select an edge $(j,k)$ uniformly at random and delete it (\emph{Type 1 transition}).
\end{enumerate}
In case the degree of $j$ does not exceed $d_j$ in (ii), we call this a \emph{Type 2 transition}.\bigskip

\begin{rtheorem}{Theorem}{\ref{thm:js_mixing}}
	Let $\mathcal{D}$ be a strongly stable family of degree sequences with respect to some constant $k$. Then there exist polynomials $p(n)$ and $r(n)$ such that for any $d = (d_1,\dots,d_n)\in \mathcal{D}$ there exists an efficient multicommodity flow $f$ for the JS chain on $\mathcal{G}'(d)$   	
	satisfying $\max_e f(e) \leq p(n)/ |\mathcal{G}'(d)|$  and  $\ell(f) \leq r(n)$.
\end{rtheorem}

\bigskip

We will use the following idea from \cite{Jerrum1989}---used in a different setting---in order to restrict ourselves to establishing flow between states in $\mathcal{G}(d)$, rather than between all states in $\mathcal{G}'(d)$. Assume that $d$ is is a degree sequence with $n$ components that is a member of a strongly stable family of degree sequences (with respect to some $k$).

\begin{lemma}\label{lem:flow_simplification}
	Let $f'$ be a flow that routes $1/|\mathcal{G}'(d)|^2$ units of flow between any pair of states in $\mathcal{G}(d)$ in the JS chain, so that $f'(e) \leq b/|\mathcal{G}'(d)|$ for all $e$ in the state space graph of the JS chain. Then $f'$ can be extended to a flow $f$ that routes $1/|\mathcal{G}'(d)|^2$ units of flow between any pair of states in $\mathcal{G}'(d)$ with the property that for all $e$ in the state space graph of the JS chain
	\[
	f(e) \leq q(n) \frac{b}{|\mathcal{G}'(d)|} \,,
	\]
	where $q(\cdot)$ is a polynomial whose degree only depends on $k_{JS}(d)$. Moreover, $\ell(f) \le \ell(f') + 2k_{JS}(d)$.\footnote{We omit the proof of Lemma \ref{lem:flow_simplification} as the lemma is actually not needed for proving Theorem \ref{thm:switch}. Careful consideration of the proof of Theorem \ref{thm:transformation} shows that we can only focus on flow between states in $\mathcal{G}(d)$, since the flow $h$ given in the proof of Theorem \ref{thm:transformation} only has positive flow between states corresponding to elements in $\mathcal{G}(d)$. That is, when defining the flow $h$, we essentially forget about all flow in $f$ between any pair of states where at least one state is an auxiliary state, i.e., an element of $\mathcal{G}'(d) \mysetminus \mathcal{G}(d)$. Said differently, in Theorem \ref{thm:transformation} we could start with the assumption that $f$ routes $1/|\mathcal{G}'(d)|^2$ units of flow between any pair of states in $\mathcal{G}(d)$ in the state space graph of the JS chain, and then the transformation still works. However, the formulations of Theorems \ref{thm:js_mixing} and \ref{thm:transformation} are more natural describing a comparison between the JS and switch chains.}
\end{lemma}

We now continue with the proof of Theorem \ref{thm:js_mixing}. It consists of four parts following, in a conceptual sense, the proof template in \cite{Cooper2007} developed for proving rapid mixing of the switch chain for regular graphs. Certain parts use similar ideas as in \cite{Jerrum1989} where a Markov chain for sampling (near)-perfect matchings is studied. Whenever we refer to \cite{Jerrum1989}, the reader is referred to Section 3 of \cite{Jerrum1989}.

\begin{proof}[Proof of Theorem \ref{thm:js_mixing}] 
	We only need to define a flow $f'$ as in Lemma \ref{lem:flow_simplification} so that $b \leq p_1(n)$ and $\ell(f') \le p_2(n)$ for some polynomials $p_1(\cdot), p_2(\cdot)$ whose degrees may only depend on $k= k_{JS}(d)$. Actually, we are going to show that we may use $p_1(n)=p_2(n)=n^2$. Then the theorem follows from the lemma and the fact that $\ln(|\mathcal{G}'(d)|)$ is upper bounded by a polynomial in $n$. The latter follows from 
	Equation (1) of McKay and Wormald \cite{McKayW91} that implies that 
	\[|\mathcal{G}(d')| \le n^{n^2}\]
	for any degree sequence $d'$ with $n$ components (see also \cite{Greenhill2017journal}). So, by the definition of  $|\mathcal{G}'(d)|$ we have 
	\[|\mathcal{G'}(d)| \le \left( \frac{n(n-1)}{2}+n+1\right) n^{n^2} \,,\]
	and thus $\ln(|\mathcal{G}'(d)|)\le 3n^3$.
	
	Before we define $f'$, we first introduce some basic terminology similar to that in \cite{Cooper2007}. Let  $V$ be a set of labeled vertices, let $\prec_{E}$ be a fixed total order on the set $\{ \{v,w\} : v,w \in V\}$ of edges, and let $\prec_{\mathcal{C}}$ be a total order on all \emph{circuits} on the complete graph $K_V$, i.e., $\prec_{\mathcal{C}}$ is  a total order on the closed walks in $K_V$ that visit every edge at most once. We fix for every circuit one of its vertices where the walk begins and ends.
	
	For given $G,G \in \mathcal{G}(d)$, let $H = G \triangle G'$ be their symmetric difference. We refer to the edges in $G \mysetminus G'$ as \emph{blue}, and the edges in $G' \mysetminus G$ as \emph{red}.
	A \emph{pairing} of red and blue edges in $H$ is a bijective mapping that, for each node $v \in V$, maps every red edge adjacent to $v$, to a blue edge adjacent to $v$. The set of all pairings is denoted by  $\Psi(G,G')$, and, with $\theta_v$ the number of red edges adjacent to $v$ (which is the same as the number of blue edges adjacent to $v$), we have 
	$
	|\Psi(G,G')| = \displaystyle{\Pi_{v \in V}} \theta_v!.
	$ 
	
%	\medskip
	\paragraph{Canonical Path Description \cite{Cooper2007}.} Similar to the approach in \cite{Cooper2007}, the goal is to construct for each pairing $\psi \in \Psi(G,G')$ a canonical path from $G$ to $G'$ that carries a $|\Psi(G,G')|^{-1}$ fraction of the total flow from $G$ to $G'$ in $f'$. For notational convenience, for the remaining of the proof we write $uv$ instead of $\{u, v\}$ to denote an edge. For a given pairing $\psi$ and the total order $\prec_E$ given above, we first decompose $H$ into the edge-disjoint union of circuits in a canonical way. We start with the lexicographically smallest edge $w_0w_1$ in $E_H$ and follow the pairing $\psi$ until we reach the edge $w_kw_0$ that was paired with $w_0w_1$. This defines the circuit $C_1$. If $C_1 = E_H$, we are done. Otherwise, we pick the lexicographically smallest edge in $H \mysetminus C_1$ and repeat this procedure. We continue generating circuits until $E_H = C_1 \cup \dots \cup C_s$. Note that all circuits have even length and alternate between red and blue edges, and that they are pairwise edge-disjoint. We form a path 
	\[
	G = Z_0,Z_1,\dots,Z_M = G'
	\]
	from $G$ to $G'$ in the state space graph of the JS chain, by processing the circuits $C_i$ in turn according to the total order $\prec_{\mathcal{C}}$.
	The \emph{processing of a circuit $C$} is the procedure during which all blue edges on $C$ are deleted, and all red edges of $C$ are added to the current graphical realization, using the three types of transitions in the JS chain mentioned at the beginning of this section. All other edges of the current graphical realization remain unchanged. In general, this can be done similarly to the circuit processing procedure in \cite{Jerrum1989}.\footnote{This is the main difference between the switch chain analyses \cite{Cooper2007,Greenhill2017journal,Erdos2013,Erdos2015,Erdos2015decomposition,ErdosMMS2018} and our analysis. The processing of a circuit is much more complicated if performed directly in the switch chain.} 
	
	\paragraph{Circuit Processing \cite{Jerrum1989}.} Let $C = vx_1x_2\dots x_qv$ be a circuit with start node $v$. We may assume, without loss of generality, that $vx_1$ is the lexicographically smallest blue edge adjacent to the starting node $v$. We first perform a type 0 transition in which we remove the blue edge $vx_1$. Then we perform a sequence of $\frac{q-1}{2}$ type 1 transitions in which we add the red edge $x_ix_{i+1}$ and remove the blue edge $x_{i-1}x_{i}$ for $i = 1,3,\dots,q$. Finally we perform a type 2 transition in which we add the red edge $vx_q$. In particular, this means that the elements on the canonical path right before and after the processing of a circuit belong to $\mathcal{G}(d)$. It is easy to see that all the intermediate elements that we visit during the processing of the circuit $C$ belong to $\mathcal{G}'(d) \mysetminus \mathcal{G}(d)$, i.e., every element has either precisely two nodes with degree deficit one, or one node with degree deficit two. This is illustrated in Figures \ref{fig:step1_2}, \ref{fig:step3_4} and \ref{fig:step5_6} for the circuit in Figure \ref{fig:circuit}.

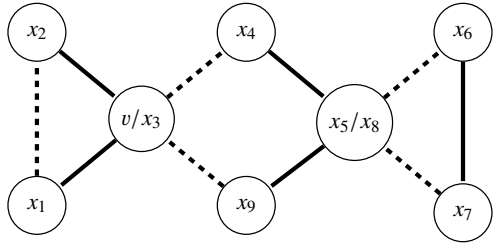
\begin{figure}[ht!]
	\centering
	\scalebox{0.8}{
		\begin{tikzpicture}[
		->,
		>=stealth',
		shorten >=1pt,
		auto,
		%node distance=2cm,
		semithick,
		every state/.style={circle,radius=0.1pt,text=black},
		]
		\begin{scope}
		\node[state]  (v)               					 		{$v/x_3$};
		\node[state]  (x1) [below left=0.7 and 1cm of v] 	 		{$x_1$};
		\node[state]  (x2) [above left=0.7 and 1cm of v]  			 		{$x_2$};
		\node[state]  (x4) [above right=0.7 and 1cm of v]					{$x_4$};
		\node[state]  (x5) [below right=0.7 and 1cm of x4] 			{$x_5/x_8$};
		\node[state]  (x6) [above right=0.7 and 1cm of x5] 			{$x_6$};
		\node[state]  (x7) [below right=0.7 and 1cm of x5] 	{$x_7$};
		\node[state]  (x9) [below right=0.7 and 1cm of v]				{$x_9$};       
		
		\path[every node/.style={sloped,anchor=south,auto=false}]
		%(v) edge[line width=1.5pt, bend left=25] 	node {} (x1)        
		(v) edge[-, line width=2pt] 	node {} (x1)            
		(x1) edge[-,dashed,line width=2pt] 	node {} (x2) 
		(x2) edge[-,line width=2pt] 	node {} (v) 
		(v) edge[-,dashed,line width=2pt] 	node {} (x4) 
		(x4) edge[-,line width=2pt] 	node {} (x5) 
		(x5) edge[-,dashed,line width=2pt] 	node {} (x6) 
		(x6) edge[-,line width=2pt] 	node {} (x7) 
		(x7) edge[-,dashed,line width=2pt] 	node {} (x5) 
		(x5) edge[-,line width=2pt] 	node {} (x9) 
		(x9) edge[-,dashed,line width=2pt] 	node {} (v);
		
		\end{scope}
		\end{tikzpicture}}
	\caption{The circuit $C = vx_1x_2x_3x_4x_5x_6x_7x_8x_9v$ with $v = x_3$ and $x_5 = x_8$. The blue edges are represented by the solid edges, and the red edges by the dashed edges.} 
	\label{fig:circuit}
\end{figure}

	\begin{figure}[ht!]
		\centering
		\scalebox{0.8}{
			\begin{tikzpicture}[
			->,
			>=stealth',
			shorten >=1pt,
			auto,
			%node distance=2cm,
			semithick,
			every state/.style={circle,radius=0.1pt,text=black},
			]
			\begin{scope}
			\node[state]  (v)               					 		{$v/x_3$};
			\node[state]  (x1) [below left=0.7 and 1cm of v] 	 		{$x_1$};
			\node[state]  (x2) [above left=0.7 and 1cm of v]  			 		{$x_2$};
			\node[state]  (x4) [above right=0.7 and 1cm of v]					{$x_4$};
			\node[state]  (x5) [below right=0.7 and 1cm of x4] 			{$x_5/x_8$};
			\node[state]  (x6) [above right=0.7 and 1cm of x5] 			{$x_6$};
			\node[state]  (x7) [below right=0.7 and 1cm of x5] 	{$x_7$};
			\node[state]  (x9) [below right=0.7 and 1cm of v]				{$x_9$};       
			
			\node (1) [right=0.1cm of x1] {$-1$};  
			\node (3) [right=0.1cm of v] {$-1$};  
			
			\path[every node/.style={sloped,anchor=south,auto=false}]
			%(v) edge[line width=1.5pt, bend left=25] 	node {} (x1)        
			%(v) edge[-, line width=2pt] 	node {} (x1)            
			(x1) edge[-,dashed,line width=3pt] 	node {} (x2) 
			(x2) edge[-,line width=3pt] 	node {} (v) 
			(v) edge[-,dashed,line width=3pt] 	node {} (x4) 
			(x4) edge[-,line width=3pt] 	node {} (x5) 
			(x5) edge[-,dashed,line width=3pt] 	node {} (x6) 
			(x6) edge[-,line width=3pt] 	node {} (x7) 
			(x7) edge[-,dashed,line width=3pt] 	node {} (x5) 
			(x5) edge[-,line width=3pt] 	node {} (x9) 
			(x9) edge[-,dashed,line width=3pt] 	node {} (v);
			
			\end{scope}
			\end{tikzpicture}}
		%\caption{Edge $vx_1$ is removed.} 
	%\label{fig:braess_5}
\ \ \ \ \ \ \ \ \ \quad \ \ \ \ \ \ \ \ \ 
\scalebox{0.8}{
\begin{tikzpicture}[
->,
>=stealth',
shorten >=1pt,
auto,
%node distance=2cm,
semithick,
every state/.style={circle,radius=0.1pt,text=black},
]
\begin{scope}
\node[state]  (v)               					 		{$v/x_3$};
\node[state]  (x1) [below left=0.7 and 1cm of v] 	 		{$x_1$};
\node[state]  (x2) [above left=0.7 and 1cm of v]  			 		{$x_2$};
\node[state]  (x4) [above right=0.7 and 1cm of v]					{$x_4$};
\node[state]  (x5) [below right=0.7 and 1cm of x4] 			{$x_5/x_8$};
\node[state]  (x6) [above right=0.7 and 1cm of x5] 			{$x_6$};
\node[state]  (x7) [below right=0.7 and 1cm of x5] 	{$x_7$};
\node[state]  (x9) [below right=0.7 and 1cm of v]				{$x_9$};       

\node (3) [right=0.1cm of v] {$-2$};  

\path[every node/.style={sloped,anchor=south,auto=false}]
%(v) edge[line width=1.5pt, bend left=25] 	node {} (x1)        
%(v) edge[-, line width=2pt] 	node {} (x1)            
(x1) edge[-,line width=3pt] 	node {} (x2) 
%(x2) edge[-,dashed,line width=0.5pt] 	node {} (v) 
(v) edge[-,dashed,line width=3pt] 	node {} (x4) 
(x4) edge[-,line width=3pt] 	node {} (x5) 
(x5) edge[-,dashed,line width=3pt] 	node {} (x6) 
(x6) edge[-,line width=3pt] 	node {} (x7) 
(x7) edge[-,dashed,line width=3pt] 	node {} (x5) 
(x5) edge[-,line width=3pt] 	node {} (x9) 
(x9) edge[-,dashed,line width=3pt] 	node {} (v);

\end{scope}
\end{tikzpicture}}
		\caption{The edge $vx_1$ is removed using a Type 0 transition (left). The edge $x_1x_2$ is added and $x_2x_3 = x_2v$ is removed using a Type 1 transition (right). We have also indicated the non-zero degree deficits.} 
		\label{fig:step1_2}
	\end{figure}
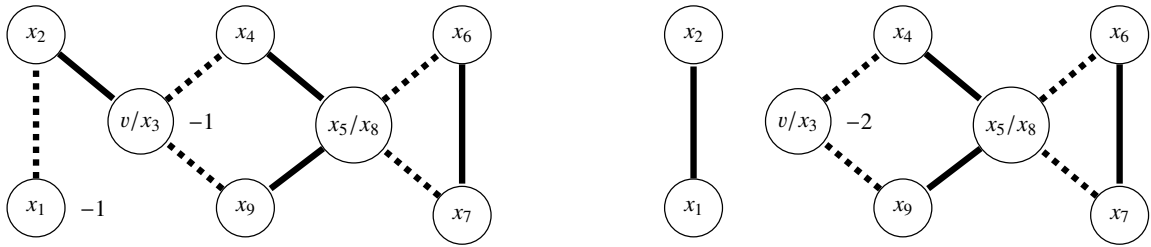
	
	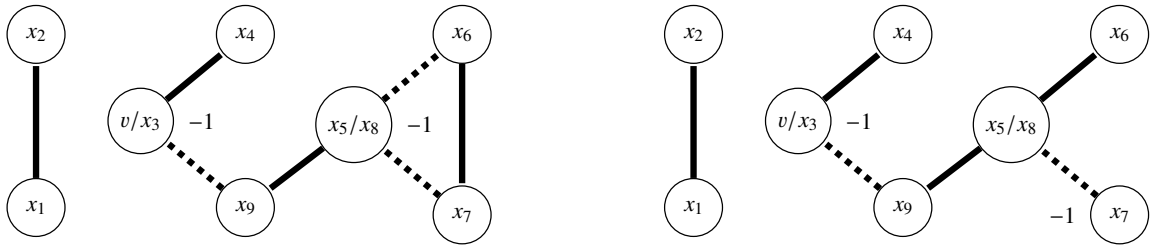
\begin{figure}[ht!]
		\centering
		\scalebox{0.8}{
			\begin{tikzpicture}[
			->,
			>=stealth',
			shorten >=1pt,
			auto,
			%node distance=2cm,
			semithick,
			every state/.style={circle,radius=0.1pt,text=black},
			]
			\begin{scope}
			\node[state]  (v)               					 		{$v/x_3$};
			\node[state]  (x1) [below left=0.7 and 1cm of v] 	 		{$x_1$};
			\node[state]  (x2) [above left=0.7 and 1cm of v]  			 		{$x_2$};
			\node[state]  (x4) [above right=0.7 and 1cm of v]					{$x_4$};
			\node[state]  (x5) [below right=0.7 and 1cm of x4] 			{$x_5/x_8$};
			\node[state]  (x6) [above right=0.7 and 1cm of x5] 			{$x_6$};
			\node[state]  (x7) [below right=0.7 and 1cm of x5] 	{$x_7$};
			\node[state]  (x9) [below right=0.7 and 1cm of v]				{$x_9$};       
			
			\node (3) [right=0.1cm of v] {$-1$};  
			\node (5) [right=0.1cm of x5] {$-1$};  
			
			\path[every node/.style={sloped,anchor=south,auto=false}]
			%(v) edge[line width=1.5pt, bend left=25] 	node {} (x1)        
			%(v) edge[-, line width=2pt] 	node {} (x1)            
			(x1) edge[-,line width=3pt] 	node {} (x2) 
			%(x2) edge[-,dashed,line width=0.5pt] 	node {} (v) 
			(v) edge[-,line width=3pt] 	node {} (x4) 
			%(x4) edge[-,line width=3pt] 	node {} (x5) 
			(x5) edge[-,dashed,line width=3pt] 	node {} (x6) 
			(x6) edge[-,line width=3pt] 	node {} (x7) 
			(x7) edge[-,dashed,line width=3pt] 	node {} (x5) 
			(x5) edge[-,line width=3pt] 	node {} (x9) 
			(x9) edge[-,dashed,line width=3pt] 	node {} (v);

			\end{scope}
			\end{tikzpicture}}
		\ \ \ \ \ \ \ \ \ \quad \ \ \ \ \ \ \ \ \ 
		\scalebox{0.8}{
			\begin{tikzpicture}[
			->,
			>=stealth',
			shorten >=1pt,
			auto,
			%node distance=2cm,
			semithick,
			every state/.style={circle,radius=0.1pt,text=black},
			]
			\begin{scope}
			\node[state]  (v)               					 		{$v/x_3$};
			\node[state]  (x1) [below left=0.7 and 1cm of v] 	 		{$x_1$};
			\node[state]  (x2) [above left=0.7 and 1cm of v]  			 		{$x_2$};
			\node[state]  (x4) [above right=0.7 and 1cm of v]					{$x_4$};
			\node[state]  (x5) [below right=0.7 and 1cm of x4] 			{$x_5/x_8$};
			\node[state]  (x6) [above right=0.7 and 1cm of x5] 			{$x_6$};
			\node[state]  (x7) [below right=0.7 and 1cm of x5] 	{$x_7$};
			\node[state]  (x9) [below right=0.7 and 1cm of v]				{$x_9$};       
			
			\node (3) [right=0.1cm of v] {$-1$};  
			\node (7) [left=0.1cm of x7] {$-1$};  
			
			\path[every node/.style={sloped,anchor=south,auto=false}]
			%(v) edge[line width=1.5pt, bend left=25] 	node {} (x1)        
			%(v) edge[-, line width=2pt] 	node {} (x1)            
			(x1) edge[-,line width=3pt] 	node {} (x2) 
			%(x2) edge[-,dashed,line width=0.5pt] 	node {} (v) 
			(v) edge[-,line width=3pt] 	node {} (x4) 
			%(x4) edge[-,line width=3pt] 	node {} (x5) 
			(x5) edge[-,line width=3pt] 	node {} (x6) 
			%(x6) edge[-,line width=3pt] 	node {} (x7) 
			(x7) edge[-,dashed,line width=3pt] 	node {} (x5) 
			(x5) edge[-,line width=3pt] 	node {} (x9) 
			(x9) edge[-,dashed,line width=3pt] 	node {} (v);

			\end{scope}
			\end{tikzpicture}}
		\caption{The edge $x_3x_4$ is added and $x_4x_5$ is removed using a Type 1 transition (left). The edge $x_5x_6$ is added and $x_6x_7$ is removed using a Type 1 transition (right).}
		\label{fig:step3_4}
	\end{figure}

	\begin{figure}[ht!]
		\centering
		\scalebox{0.8}{
			\begin{tikzpicture}[
			->,
			>=stealth',
			shorten >=1pt,
			auto,
			%node distance=2cm,
			semithick,
			every state/.style={circle,radius=0.1pt,text=black},
			]
			\begin{scope}
			\node[state]  (v)               					 		{$v/x_3$};
			\node[state]  (x1) [below left=0.7 and 1cm of v] 	 		{$x_1$};
			\node[state]  (x2) [above left=0.7 and 1cm of v]  			 		{$x_2$};
			\node[state]  (x4) [above right=0.7 and 1cm of v]					{$x_4$};
			\node[state]  (x5) [below right=0.7 and 1cm of x4] 			{$x_5/x_8$};
			\node[state]  (x6) [above right=0.7 and 1cm of x5] 			{$x_6$};
			\node[state]  (x7) [below right=0.7 and 1cm of x5] 	{$x_7$};
			\node[state]  (x9) [below right=0.7 and 1cm of v]				{$x_9$};       
			
			\node (3) [right=0.1cm of v] {$-1$};  
			\node (9) [right=0.1cm of x9] {$-1$};  
			
			\path[every node/.style={sloped,anchor=south,auto=false}]
			%(v) edge[line width=1.5pt, bend left=25] 	node {} (x1)        
			%(v) edge[-, line width=2pt] 	node {} (x1)            
			(x1) edge[-,line width=3pt] 	node {} (x2) 
			%(x2) edge[-,dashed,line width=0.5pt] 	node {} (v) 
			(v) edge[-,line width=3pt] 	node {} (x4) 
			%(x4) edge[-,line width=3pt] 	node {} (x5) 
			(x5) edge[-,line width=3pt] 	node {} (x6) 
			%(x6) edge[-,line width=3pt] 	node {} (x7) 
			(x7) edge[-,line width=3pt] 	node {} (x5) 
			%(x5) edge[-,line width=3pt] 	node {} (x9) 
			(x9) edge[-,dashed,line width=3pt] 	node {} (v);

			\end{scope}
			\end{tikzpicture}}
		\ \ \ \ \ \ \ \ \ \quad \ \ \ \ \ \ \ \ \ 
		\scalebox{0.8}{
			\begin{tikzpicture}[
			->,
			>=stealth',
			shorten >=1pt,
			auto,
			%node distance=2cm,
			semithick,
			every state/.style={circle,radius=0.1pt,text=black},
			]
			\begin{scope}
			\node[state]  (v)               					 		{$v/x_3$};
			\node[state]  (x1) [below left=0.7 and 1cm of v] 	 		{$x_1$};
			\node[state]  (x2) [above left=0.7 and 1cm of v]  			 		{$x_2$};
			\node[state]  (x4) [above right=0.7 and 1cm of v]					{$x_4$};
			\node[state]  (x5) [below right=0.7 and 1cm of x4] 			{$x_5/x_8$};
			\node[state]  (x6) [above right=0.7 and 1cm of x5] 			{$x_6$};
			\node[state]  (x7) [below right=0.7 and 1cm of x5] 	{$x_7$};
			\node[state]  (x9) [below right=0.7 and 1cm of v]				{$x_9$};       
			
			%\node (3) [below=0.1cm of v] {$-1$};  
			%\node (9) [below=0.1cm of x9] {$-1$};  
			
			\path[every node/.style={sloped,anchor=south,auto=false}]
			%(v) edge[line width=1.5pt, bend left=25] 	node {} (x1)        
			%(v) edge[-, line width=2pt] 	node {} (x1)            
			(x1) edge[-,line width=3pt] 	node {} (x2) 
			%(x2) edge[-,dashed,line width=0.5pt] 	node {} (v) 
			(v) edge[-,line width=3pt] 	node {} (x4) 
			%(x4) edge[-,line width=3pt] 	node {} (x5) 
			(x5) edge[-,line width=3pt] 	node {} (x6) 
			%(x6) edge[-,line width=3pt] 	node {} (x7) 
			(x7) edge[-,line width=3pt] 	node {} (x5) 
			%(x5) edge[-,line width=3pt] 	node {} (x9) 
			(x9) edge[-,line width=3pt] 	node {} (v);

			\end{scope}
			\end{tikzpicture}}
		\caption{The edge $x_7x_8 = x_5x_8$ is added and $x_5x_9 = x_8x_9$ is removed using a Type 1 transition (left). The edge $vx_9$ is added using a Type 2 transition (right).}
		\label{fig:step5_6}
	\end{figure}
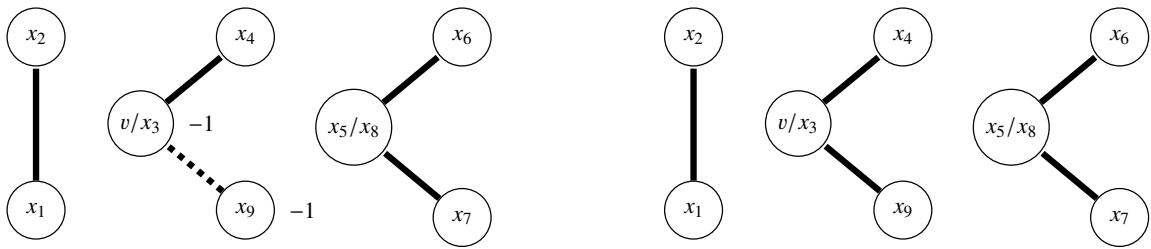\bigskip
	
	For the next part, we define the notion of an encoding that can be used to bound the congestion of an edge in the state space graph of the JS chain using an injective mapping argument.

	%\newpage
	\paragraph{Encoding \cite{Jerrum1989}.} Let $t = (Z,Z')$ be a given transition of the Markov chain.
	Suppose two graphs $G$ and $G'$ use the transition $t$ over some canonical path for some pairing $\psi \in \Psi(G,G')$. Let $H = G \triangle G'$. We define the encoding
	\[
	L_t(G,G') = \left\{ \begin{array}{ll}
	(H \triangle (Z \cup Z')) - e_{H,t} & \text{ if $t$ is a Type 1 transition, } \\
	H \triangle (Z \cup Z') & \text{ otherwise, }
	\end{array}\right.
	\]
	where $e_{H,t}$ is the first blue edge on the circuit that is currently being processed on the canonical path from $G$ to $G'$ (for the given pairing $\psi$).
	This encoding is of a similar nature as the encoding used in \cite{Jerrum1989}. An example is given in Figures \ref{fig:encoding1}, \ref{fig:encoding2} and \ref{fig:encoding3}. We also refer the reader to Figure 1 in \cite{Jerrum1989} for a detailed example.\footnote{Although the perfect matching setting might seem different at first glance, it is actually closely related to our setting, with the only difference that the symmetric difference of two perfect matchings is the union of \emph{node-disjoint cycles}, whereas in our setting the symmetric difference of two graphical realizations is the union of \emph{edge-disjoint circuits}. This is roughly why the notion of pairings is needed, as they allow us to uniquely determine the circuits. That is, the edge-disjoint circuits determined by the pairing are the analogue of the node-disjoint cycles in the perfect matching setting in \cite{Jerrum1989}.} The following lemma is crucial for the analysis. \bigskip
	
	\begin{lemma}\label{lem:recovery}
		Given $t = (Z,Z')$, $L$, and $\psi$, we can uniquely recover $G$ and $G'$. That is, if $L$ is such that  $L_t = L_t(G,G')$ for some pair $(G,G')$, then $(G,G')$ is the unique pair for which this is the case, given $t$, $L$, $\psi$.
	\end{lemma}
	\begin{proof}
		We give the proof for when $t$ is a Type 1 transition. The cases of the two other types are similar, and arguably somewhat easier. The proof uses the arguments in \cite{Jerrum1989} interpreted in our setting. First note that
		$
		L \triangle (Z \cup Z')
		$
		is a graph in which there are precisely two nodes with odd degree. In particular, the edge $e_{H,t}$ is the unique edge (having as endpoints these odd degree nodes) that has to be added to $L \triangle (Z \cup Z')$ to obtain $H = G \triangle G'$. That is, we have $(L \triangle(Z \cup Z')) + e_{H,t} = H$. The pairing $\psi$ then yields a unique circuit decomposition of $E(H)$ as explained at the beginning of the proof. From the transition $t$ it can be inferred which circuit is currently being processed, and, moreover, we can infer which edges of that circuit belong to $G$ and which to $G'$. Furthermore, the global ordering $\prec_{\mathcal{C}}$ on all circuits can then be used to determine for every other circuit whether it has been processed already or not. For every such circuit, we can then infer which edges on it belong to $G$ and which to $G'$ by comparing with $Z$ (or $Z'$). Therefore, $G$ and $G'$ can be uniquely recovered from $t$, $L$ and $\psi$. 
	\end{proof}
	
	%\newpage
	\begin{figure}[ht!]
		\centering
		\scalebox{0.8}{
			\begin{tikzpicture}[
			->,
			>=stealth',
			shorten >=1pt,
			auto,
			%node distance=2cm,
			semithick,
			every state/.style={circle,radius=0.1pt,text=black},
			]
			\begin{scope}
			\node[state]  (a1)               					 		{$a_1$};
			\node[state]  (a2) [above=2cm of a1]            			{$a_2$};
			\node[state]  (a3) [right=2cm of a2]                		{$a_3$};
			\node[state]  (a4) [below=2cm of a3]               	    {$a_4$};
			
			\node[state]  (x1) [right=2cm of a4]  	 				{$x_1$};
			\node[state]  (v)  [above right=0.75 and 1cm of x1]          {$v/x_3$};
			\node[state]  (x2) [above left=0.75 and 1cm of v]  			{$x_2$};
			\node[state]  (x4) [above right=0.75 and 1cm of v]			{$x_4$};
			\node[state]  (x5) [below right=0.75 and 1cm of x4] 			{$x_5/x_8$};
			\node[state]  (x6) [above right=0.75 and 1cm of x5] 			{$x_6$};
			\node[state]  (x7) [below right=0.75 and 1cm of x5] 			{$x_7$};
			\node[state]  (x9) [below right=0.75 and 1cm of v]			{$x_9$};       
			
			\node[state]  (b1) [right=2cm of x7]          			{$b_1$};
			\node[state]  (b2) [above=2cm of b1]            			{$b_2$};
			\node[state]  (b3) [right=2cm of b2]                		{$b_3$};
			\node[state]  (b4) [below=2cm of b3]               	    {$b_4$};
			\path[every node/.style={sloped,anchor=south,auto=false}]
			%(v) edge[line width=1.5pt, bend left=25] 	node {} (x1)        
			(v) edge[-, line width=2pt] 	node {} (x1)            
			(x1) edge[-,dashed,line width=2pt] 	node {} (x2) 
			(x2) edge[-,line width=2pt] 	node {} (v) 
			(v) edge[-,dashed,line width=2pt] 	node {} (x4) 
			(x4) edge[-,line width=2pt] 	node {} (x5) 
			(x5) edge[-,dashed,line width=2pt] 	node {} (x6) 
			(x6) edge[-,line width=2pt] 	node {} (x7) 
			(x7) edge[-,dashed,line width=2pt] 	node {} (x5) 
			(x5) edge[-,line width=2pt] 	node {} (x9) 
			(x9) edge[-,dashed,line width=2pt] 	node {} (v)
			%Path for the a and b cycles
			(a1) edge[-, line width=2pt] 	node {} (a2)   
			(a2) edge[-,dashed, line width=2pt] 	node {} (a3)   
			(a3) edge[-, line width=2pt] 	node {} (a4)   
			(a4) edge[-,dashed, line width=2pt] 	node {} (a1)   
			(b1) edge[-, line width=2pt] 	node {} (b2)   
			(b2) edge[-,dashed, line width=2pt] 	node {} (b3)   
			(b3) edge[-, line width=2pt] 	node {} (b4)   
			(b4) edge[-,dashed, line width=2pt] 	node {} (b1);
			
			\end{scope}
			\end{tikzpicture}}
		\caption{Symmetric difference $H = G \triangle G'$ where the solid edges represent the edges $G$ and the dashed edges the edges of $G'$. From left to right the circuit are numbered $C_1, C_2$ and $C_3$, and assume that this is also the order in which they are processed.} 
		\label{fig:encoding1}
	\end{figure}
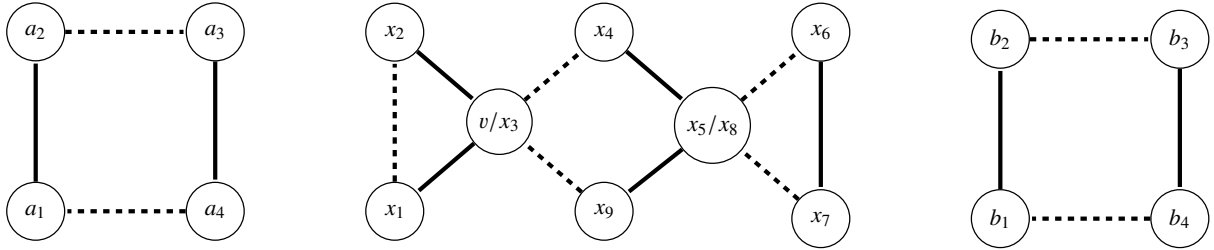 
	
	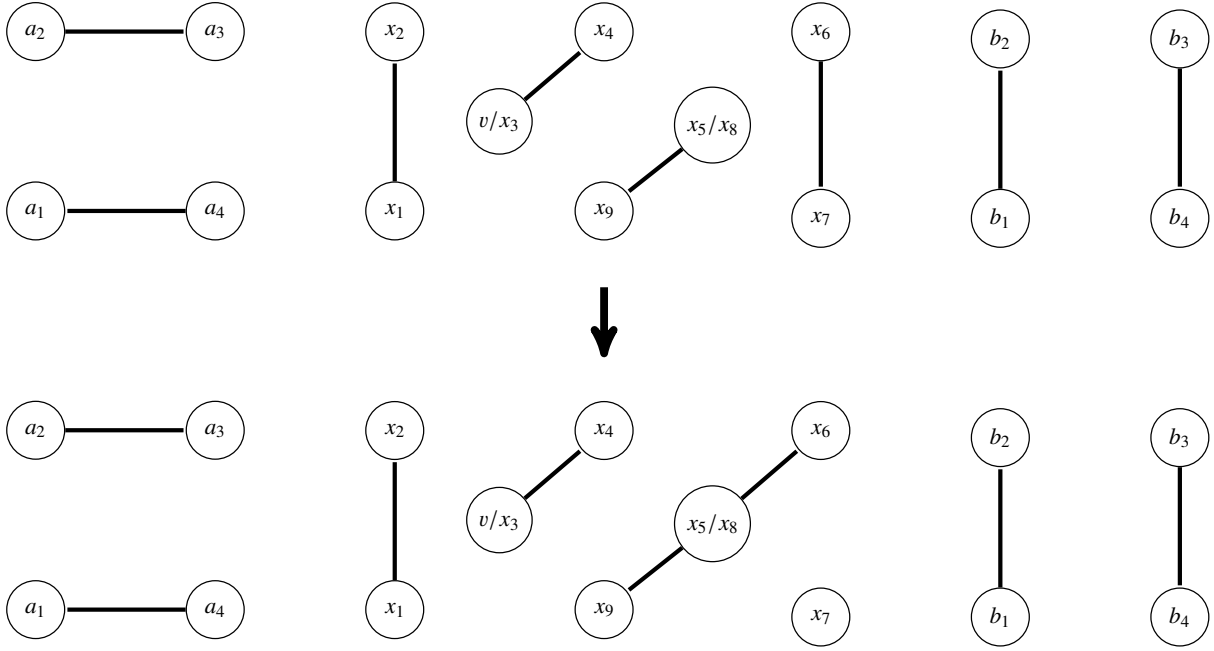
\begin{figure}[ht!]
		\centering
		\scalebox{0.8}{
			\begin{tikzpicture}[
			->,
			>=stealth',
			shorten >=1pt,
			auto,
			%node distance=2cm,
			semithick,
			every state/.style={circle,radius=0.1pt,text=black},
			]
			\begin{scope}
			\node[state]  (a1)               					 		{$a_1$};
			\node[state]  (a2) [above=2cm of a1]            			{$a_2$};
			\node[state]  (a3) [right=2cm of a2]                		{$a_3$};
			\node[state]  (a4) [below=2cm of a3]               	    {$a_4$};
			
			\node[state]  (x1) [right=2cm of a4]  	 				{$x_1$};
			\node[state]  (v)  [above right=0.75 and 1cm of x1]          {$v/x_3$};
			\node[state]  (x2) [above left=0.75 and 1cm of v]  			{$x_2$};
			\node[state]  (x4) [above right=0.75 and 1cm of v]			{$x_4$};
			\node[state]  (x5) [below right=0.75 and 1cm of x4] 			{$x_5/x_8$};
			\node[state]  (x6) [above right=0.75 and 1cm of x5] 			{$x_6$};
			\node[state]  (x7) [below right=0.75 and 1cm of x5] 			{$x_7$};
			\node[state]  (x9) [below right=0.75 and 1cm of v]			{$x_9$};       
			
			\node[state]  (b1) [right=2cm of x7]          			{$b_1$};
			\node[state]  (b2) [above=2cm of b1]            			{$b_2$};
			\node[state]  (b3) [right=2cm of b2]                		{$b_3$};
			\node[state]  (b4) [below=2cm of b3]               	    {$b_4$};

			\node (arrow1) [below=0.5cm of x9]               	    {};
			\node (arrow2) [below=1.2cm of arrow1]               	  {};
			\node (arrow3) [below=0.1cm of arrow2]               	  {};
			
			%\node (Z) [above left=0.5 and 0.5cm of a1]               	  {$Z:$};
			
			\path[every node/.style={sloped,anchor=south,auto=false}]
			%(v) edge[line width=1.5pt, bend left=25] 	node {} (x1)        
			%(v) edge[-, line width=2pt] 	node {} (x1)            
			(x1) edge[-,line width=2pt] 	node {} (x2) 
			%(x2) edge[-,line width=2pt] 	node {} (v) 
			(v) edge[-,line width=2pt] 	node {} (x4) 
			%(x4) edge[-,line width=2pt] 	node {} (x5) 
			%(x5) edge[-,dashed,line width=2pt] 	node {} (x6) 
			(x6) edge[-,line width=2pt] 	node {} (x7) 
			%(x7) edge[-,dashed,line width=2pt] 	node {} (x5) 
			(x5) edge[-,line width=2pt] 	node {} (x9) 
			%(x9) edge[-,dashed,line width=2pt] 	node {} (v)
			%Path for the a and b cycles
			%(a1) edge[-, line width=2pt] 	node {} (a2)   
			(a2) edge[-, line width=2pt] 	node {} (a3)   
			%(a3) edge[-, line width=2pt] 	node {} (a4)   
			(a4) edge[-, line width=2pt] 	node {} (a1)   
			(b1) edge[-, line width=2pt] 	node {} (b2)   
			%(b2) edge[-,dashed, line width=2pt] 	node {} (b3)   
			(b3) edge[-, line width=2pt] 	node {} (b4)   
			%(b4) edge[-,dashed, line width=2pt] 	node {} (b1);
			(arrow1) edge[line width=3.7pt] 	node {} (arrow2);     
			\end{scope}
			\end{tikzpicture}}
		\quad
		\scalebox{0.8}{
			\begin{tikzpicture}[
			->,
			>=stealth',
			shorten >=1pt,
			auto,
			%node distance=2cm,
			semithick,
			every state/.style={circle,radius=0.1pt,text=black},
			]
			\begin{scope}
			\node[state]  (a1)               					 		{$a_1$};
			\node[state]  (a2) [above=2cm of a1]            			{$a_2$};
			\node[state]  (a3) [right=2cm of a2]                		{$a_3$};
			\node[state]  (a4) [below=2cm of a3]               	    {$a_4$};
			
			\node[state]  (x1) [right=2cm of a4]  	 				{$x_1$};
			\node[state]  (v)  [above right=0.75 and 1cm of x1]          {$v/x_3$};
			\node[state]  (x2) [above left=0.75 and 1cm of v]  			{$x_2$};
			\node[state]  (x4) [above right=0.75 and 1cm of v]			{$x_4$};
			\node[state]  (x5) [below right=0.75 and 1cm of x4] 			{$x_5/x_8$};
			\node[state]  (x6) [above right=0.75 and 1cm of x5] 			{$x_6$};
			\node[state]  (x7) [below right=0.75 and 1cm of x5] 			{$x_7$};
			\node[state]  (x9) [below right=0.75 and 1cm of v]			{$x_9$};       
			
			\node[state]  (b1) [right=2cm of x7]          			{$b_1$};
			\node[state]  (b2) [above=2cm of b1]            			{$b_2$};
			\node[state]  (b3) [right=2cm of b2]                		{$b_3$};
			\node[state]  (b4) [below=2cm of b3]               	    {$b_4$};
			\path[every node/.style={sloped,anchor=south,auto=false}]
			%(v) edge[line width=1.5pt, bend left=25] 	node {} (x1)        
			%(v) edge[-, line width=2pt] 	node {} (x1)            
			(x1) edge[-,line width=2pt] 	node {} (x2) 
			%(x2) edge[-,line width=2pt] 	node {} (v) 
			(v) edge[-,line width=2pt] 	node {} (x4) 
			%(x4) edge[-,line width=2pt] 	node {} (x5) 
			(x5) edge[-,line width=2pt] 	node {} (x6) 
			%(x6) edge[-,line width=2pt] 	node {} (x7) 
			%(x7) edge[-,dashed,line width=2pt] 	node {} (x5) 
			(x5) edge[-,line width=2pt] 	node {} (x9) 
			%(x9) edge[-,dashed,line width=2pt] 	node {} (v)
			%Path for the a and b cycles
			%(a1) edge[-, line width=2pt] 	node {} (a2)   
			(a2) edge[-, line width=2pt] 	node {} (a3)   
			%(a3) edge[-, line width=2pt] 	node {} (a4)   
			(a4) edge[-, line width=2pt] 	node {} (a1)   
			(b1) edge[-, line width=2pt] 	node {} (b2)   
			%(b2) edge[-,dashed, line width=2pt] 	node {} (b3)   
			(b3) edge[-, line width=2pt] 	node {} (b4);   
			%(b4) edge[-,dashed, line width=2pt] 	node {} (b1);
			
			\end{scope}
			\end{tikzpicture}}
		\caption{The transition $t = (Z,Z')$ that removes the edge $x_6x_7$ and adds the edge $x_5x_6$ as part of the processing of $C_2$. Note that $C_1$ has already been processed. The edges in $(E(G) \cup E(G')) \setminus E(H)$ are left out.} 
		\label{fig:encoding2}
	\end{figure}

	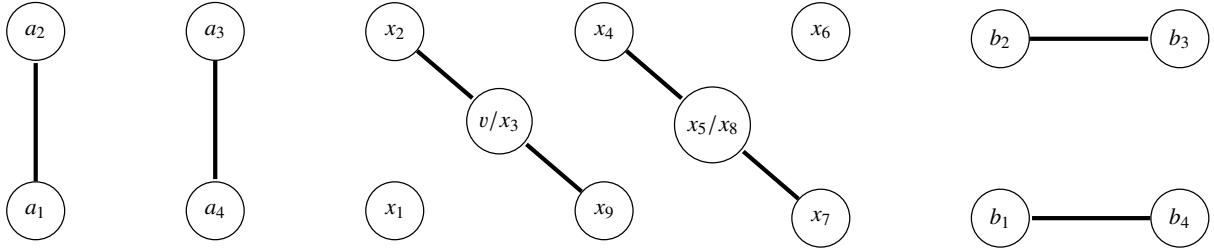
\begin{figure}[ht!]
		\centering
		\scalebox{0.8}{
			\begin{tikzpicture}[
			->,
			>=stealth',
			shorten >=1pt,
			auto,
			%node distance=2cm,
			semithick,
			every state/.style={circle,radius=0.1pt,text=black},
			]
			\begin{scope}
			\node[state]  (a1)               					 		{$a_1$};
			\node[state]  (a2) [above=2cm of a1]            			{$a_2$};
			\node[state]  (a3) [right=2cm of a2]                		{$a_3$};
			\node[state]  (a4) [below=2cm of a3]               	    {$a_4$};
			
			\node[state]  (x1) [right=2cm of a4]  	 				{$x_1$};
			\node[state]  (v)  [above right=0.75 and 1cm of x1]          {$v/x_3$};
			\node[state]  (x2) [above left=0.75 and 1cm of v]  			{$x_2$};
			\node[state]  (x4) [above right=0.75 and 1cm of v]			{$x_4$};
			\node[state]  (x5) [below right=0.75 and 1cm of x4] 			{$x_5/x_8$};
			\node[state]  (x6) [above right=0.75 and 1cm of x5] 			{$x_6$};
			\node[state]  (x7) [below right=0.75 and 1cm of x5] 			{$x_7$};
			\node[state]  (x9) [below right=0.75 and 1cm of v]			{$x_9$};       
			
			\node[state]  (b1) [right=2cm of x7]          			{$b_1$};
			\node[state]  (b2) [above=2cm of b1]            			{$b_2$};
			\node[state]  (b3) [right=2cm of b2]                		{$b_3$};
			\node[state]  (b4) [below=2cm of b3]               	    {$b_4$};
			\path[every node/.style={sloped,anchor=south,auto=false}]
			%(v) edge[line width=1.5pt, bend left=25] 	node {} (x1)        
			%(v) edge[-, line width=2pt] 	node {} (x1)            
			%(x1) edge[-,dashed,line width=2pt] 	node {} (x2) 
			(x2) edge[-,line width=2pt] 	node {} (v) 
			%(v) edge[-,dashed,line width=2pt] 	node {} (x4) 
			(x4) edge[-,line width=2pt] 	node {} (x5) 
			%(x5) edge[-,dashed,line width=2pt] 	node {} (x6) 
			%(x6) edge[-,line width=2pt] 	node {} (x7) 
			(x7) edge[-,line width=2pt] 	node {} (x5) 
			%(x5) edge[-,line width=2pt] 	node {} (x9) 
			(x9) edge[-,line width=2pt] 	node {} (v)
			%Path for the a and b cycles
			(a1) edge[-, line width=2pt] 	node {} (a2)   
			%(a2) edge[-,dashed, line width=2pt] 	node {} (a3)   
			(a3) edge[-, line width=2pt] 	node {} (a4)   
			%(a4) edge[-,dashed, line width=2pt] 	node {} (a1)   
			%(b1) edge[-, line width=2pt] 	node {} (b2)   
			(b2) edge[-, line width=2pt] 	node {} (b3)   
			%(b3) edge[-, line width=2pt] 	node {} (b4)   
			(b4) edge[-, line width=2pt] 	node {} (b1);
			
			\end{scope}
			\end{tikzpicture}}
		\caption{The encoding $L = L_t(G,G')$, where again the edges in $(E(G) \cup E(G')) \setminus E(H)$ are left out. Note that in this case  $e_{H,t}=vx_1$ and that $L$ is itself an element of $\mathcal{G}'(d)$.} 
		\label{fig:encoding3}
	\end{figure}

	%\newpage
	
	\paragraph{Injective Mapping Argument \cite{Cooper2007,Jerrum1989}.} We complete the proof by using an injective mapping argument 
	to bound the congestion of the flow $f'$ on the edges of the state space graph of the JS chain. The arguments used are a combination of ideas from \cite{Jerrum1989}  	
	and the proof of Lemma 2.5 in \cite{Cooper2007} (see also Lemma 1 in  \cite{Cooper2012corrigendum}).\footnote{Lemma 2.5 in \cite{Cooper2007} contained a flaw for which the corrigendum \cite{Cooper2012corrigendum} was published.}
	We again focus on Type 1 transitions $t$ as the proofs for the other two types are similar but simpler.

	For a tuple $(G,G',\psi)$, let $p_{\psi}(G,G')$ denote the canonical path from $G$ to $G'$ for pairing $\psi$.
	Let 
	\[
	\mathcal{L}_t =\{L_t(G,G') \,|\, (G,G',\psi) \in \mathcal{F}_t\} 
	\]
	be the set of all (distinct) encodings $L_t$, where 
	\[
	\mathcal{F}_t = \big\{(G,G',\psi) : t \in p_{\psi}(G,G')\big\}
	\] 
	is the set of all tuples $(G,G',\psi)$ such that the canonical path from $G$ to $G'$ under pairing $\psi$ uses the transition $t$. 
	A crucial observation is that every encoding $L_t(G,G')$ itself is an element of $\mathcal{G}'(d)$ (see Figure \ref{fig:encoding3} for an example). This implies that
	\begin{equation}\label{eq:injective1}
	|\mathcal{L}_t| \leq |\mathcal{G}'(d)|.
	\end{equation}

	Moreover, with $H = G \triangle G'$ and $L=L_t(G,G')$, the pairing $\psi$ has the property that it pairs up the edges of $E(H)\mysetminus E(L)$ and $E(H)\cap E(L)$ in such a way that for every node $v$ (with the exception of at most two nodes) each edge in $E(H)\mysetminus E(L)$ that is incident to $v$ is paired up with an edge in $E(H)\cap E(L)$ that is incident to $v$. However, there are either two nodes for which the incident edges in $E(H)\mysetminus E(L)$ exceed by 2 the incident edges in $E(H)\cap E(L)$, or one node for which the incident edges in $E(H)\mysetminus E(L)$ exceed by 4 the incident edges in $E(H)\cap E(L)$. These are exactly the two nodes with degree deficit 1 or the one node with degree deficit 2 in $L$; for the example in Figure \ref{fig:encoding3} these are nodes $x_1$ and $x_6$. There $\psi$ pairs up each edge of $E(H)\cap E(L)$ to an edge of $E(H)\mysetminus E(L)$ but also two edges of $E(H)\mysetminus E(L)$ with each other; or in the case of one node with degree deficit 2 $\psi$ pairs up each edge of $E(H)\cap E(L)$ to an edge of $E(H)\mysetminus E(L)$ but also makes two pairs out of the remaining 4 edges in $E(H)\mysetminus E(L)$. 
	Let $\Psi'(L)$ be the set of all pairings with this property.\footnote{Remember that we do not need to know $G$ and $G'$ in order to determine the set $H$. It can be found based on $L$ and the transition $t = (Z,Z')$, as described in the proof of Lemma \ref{lem:recovery}.} 
	Note that not every such pairing has to correspond to a tuple $(G,G',\psi)$ for which $t \in p_{\psi}(G,G')$. 
	
	By  simply counting, we can upper bound $|\Psi'(L)|$ in terms of $|\Psi(H)|$. We show the calculation for the case where $L$ has two nodes with degree deficit 1. The case of one node with degree deficit 2 is very similar and the same upper bound  works there as well. Suppose that $u, w$ are the two   nodes of $L$ with degree deficit 1. Then 
	\begin{equation}\label{eq:injective2}
	|\Psi'(L)| = \left( \Pi_{v \in V\mysetminus\{u, w\}}\theta_v! \right) \cdot \frac{(\theta_u+1)!}{2} \cdot \frac{(\theta_w+1)!}{2} = |\Psi(H)| \cdot \frac{(\theta_u+1)(\theta_w+1)}{4} \leq n^2 \cdot |\Psi(H)| \,.
	\end{equation}
	Putting everything together, we have
	\begin{eqnarray*}
	|\mathcal{G}'(d)|^2 f'(e)  & = & \sum_{(G,G')} \sum_{\psi \in \Psi(G,G')} \mathbf{1}(e \in p_{\psi}(H)) |\Psi(H)|^{-1}  \nonumber \\
	& \leq & \sum_{L \in \mathcal{L}_t} \sum_{\psi' \in \Psi'(L)}  |\Psi(H)|^{-1} \ \ \ \ \ (\text{using Lemma } \ref{lem:recovery}) \nonumber \\
	& \leq & n^2 \sum_{L \in \mathcal{L}_t} 1 \ \  \ \ \ \ \ \ \ \ \ \ \ \ \ \ \ \ (\text{using } (\ref{eq:injective2})) \nonumber \\
	& \leq & n^2\cdot|\mathcal{G}'(d)|  \ \  \ \ \ \ \ \ \ \ \ \ \ \ (\text{using } (\ref{eq:injective1}))
	\end{eqnarray*}
	The usage of Lemma \ref{lem:recovery} for the first inequality works as follows. Every tuple $(G,G',\psi)\in \mathcal{F}_t$ with encoding $L_t(G, G')$ generates a unique tuple in $\{L_t(G, G')\} \times \Psi'(L_t(G, G'))$. But since, by Lemma \ref{lem:recovery}, we can uniquely recover $G$ and $G'$ from $L$, $t$ and $\psi$, we have that $\sum_{L \in \mathcal{L}_t} |\{L\} \times \Psi'(L)| = \sum_{L \in \mathcal{L}_t} \sum_{\psi' \in \Psi'(L)} 1$ is an upper bound on the number of canonical paths that use $t$. 
	
	By rearranging \eqref{eq:injective2} we get the upper bound for  $f'$ required in Lemma \ref{lem:flow_simplification}. What is left to show is that $\ell(f')$ is not too large. This, however, is determined by the way we defined the canonical paths. It is easy to see that for any canonical path between any two graphs $G, G' \in \mathcal{G}(d)$ has length at most $\frac{3}{4} |E(G \triangle G')|$ and, therefore, $\ell(f')\le n^2$.
\end{proof}

\begin{remark}[Bipartite case]\label{rem:app_bip}
	The proof for the bipartite case is very similar. The only difference is that in the circuit processing procedure there will never be an auxiliary state where one node has degree deficit two. For this to occur there necessarily has to be a simple cycle of odd length in a circuit, see, e.g., Figure \ref{fig:step1_2}. This explains the adjusted definition of $\mathcal{G}'(r,c)$ used in the definition of (strong) stability for the bipartite case in Appendix \ref{sec:bipartite}. Moreover, the exact same encoding and injective mapping arguments can be used.
\end{remark}

\newpage

%%%%%%%%%%%%%%%%%%%%%%%%%%%%%%%%%%%%%%%%%%%%%%%%%%%%%%%%%%%%%%%%%%%%%%%%
%%%%%%%%%%%%%%%%%%%%%%%%%%%%%%%%%%%%%%%%%%%%%%%%%%%%%%%%%%%%%%%%%%%%%%%%
\section{Sampling Bipartite Graphs}\label{sec:bipartite}
%%%%%%%%%%%%%%%%%%%%%%%%%%%%%%%%%%%%%%%%%%%%%%%%%%%%%%%%%%%%%%%%%%%%%%%%
%%%%%%%%%%%%%%%%%%%%%%%%%%%%%%%%%%%%%%%%%%%%%%%%%%%%%%%%%%%%%%%%%%%%%%%%
A similar analysis as in Appendix \ref{app:js} holds for the bipartite case, under some slightly adjusted definitions.			
Throughout this section, we use $m$ and $n$ to denote the number of vertices in the two independent sets of a given bipartition. 

For a given tuple $(r,c)$ with $r = (r_1,\dots,r_m)$ and $c = (c_1,\dots,c_n)$, we define $\mathcal{G}(r,c)$ as the set of all bipartite graphs $G = (V \cup U, E)$ with $V = \{1,\dots,m\}$ and $U = \{1,\dots,n\}$, where the nodes in $V$ have degree sequence $r$, and the nodes in $U$ degree sequence $c$. The set $\mathcal{G}'(r,c)$ is defined as $\cup_{(r',c')} \mathcal{G}(r',c')$ where we either have $(r',c') = (r,c)$, or, there exist  $x\in \{1,\dots,m\}$ and $y \in \{1,\dots,n\}$ such that 
\[
r'_i = \left\{ \begin{array}{ll} r_i - 1 & \text{ if } i = x, \\
r_i & \text{ otherwise,}
\end{array}\right. \ \  \text{ and } \ \  
c'_j = \left\{ \begin{array}{ll} c_j - 1 & \text{ if } j = y, \\
c_j & \text{ otherwise.}
\end{array}\right.
\]
Note that the case in which there is one node with degree deficit two cannot occur here.
The definitions of the JS and switch chain, as well as that of the parameter in (\ref{eq:distance}), denoted by $k_{JS}(r,c)$ here, and the notions of  stability and strong stability, are very similar. In particular, in the JS chain the only difference is that we now pick $i \in U$ and $j \in V$ uniformly at random instead of picking  $i, j \in U \cup V$ uniformly at random. 

Crucially, as is stated in Theorem \ref{thm:switch_bipartite} below, the result in Theorem \ref{thm:switch} carries over. The flow transformation in Theorem \ref{thm:transformation} is completely analogous for the bipartite case. The \emph{merging} and \emph{rerouting} procedures can be carried out in exactly the same fashion, and instead of Theorem \ref{thm:swap} one can use the corresponding result for the bipartite case, see, e.g., \cite{Rao1996,Erdos2013swap}.	
Theorem \ref{thm:js_mixing} is also true for the bipartite case; this is discussed in Remark \ref{rem:app_bip} in Appendix \ref{app:js}. 

\begin{theorem}[Bipartite case]\label{thm:switch_bipartite}
	Let $\mathcal{D}$ be a strongly stable family of bipartite degree sequences with respect to some constant $k$. Then there exists a polynomial $q(n,m)$ such that, for any $0 < \epsilon < 1$, the mixing time $\tau_{\mathrm{sw}}$ of the switch chain for a graphical sequence $(r,c) \in \mathcal{D}$, with $r = (r_1,\dots,r_m)$ and $c = (c_1,\dots,c_n)$, satisfies
	\[
	\tau_{\mathrm{sw}}(\epsilon) \leq q(n,m)^k \ln(1/\epsilon) \,.
	\]
\end{theorem}

Similarly, we can obtain a constant upper bound on $k_{JS}(r,c)$ for certain degree families. We first present a range of sequences that can be considered a bipartite counterpart of Corollary \ref{cor:stable}.

\begin{corollary}\label{cor:bipartite}
	Let $\mathcal{D} = \mathcal{D}(\delta_{r},\Delta_{r},\delta_{c},\Delta_{c})$ be the set of all graphical bipartite degree sequences $(r,c)$ on $m$ and $n$ nodes respectively, satisfying 
	\begin{equation}\label{eq:bip_stable}
	(\Delta_{r} - \delta_{c})^2 \leq 4\delta_{c}(n - \Delta_{r}) \ \ \ \text{ and } \ \ \ (\Delta_{c} - \delta_{r})^2 \leq 4\delta_{r}(m - \Delta_{c})
	\end{equation}
	where $\delta_{r},\Delta_{r}$ are the minimum and maximum component of $r$, and $\delta_{c},\Delta_{c}$ the minimum and maximum component of $c$ respectively. For any $(r,c) \in \mathcal{D}$, we have $k_{JS}(r,c) \leq 8$. Hence the switch chain is rapidly mixing for sequences in $\mathcal{D}$.
\end{corollary}

The proof is  a bipartite analogue of Lemma 1 in \cite{Jerrum1989graphical}.
In particular, we use the notions of co-edges and alternating paths as defined in Appendix \ref{app:open_question}. 
\begin{proof}
	Suppose that $(r',c')$ is such that there exist precisely one $s \in V$ and $t \in U$ with degree deficit one, and all other nodes have no deficit. Let $G$ be a realization for the sequence $(r',c')$ if $\mathcal{G}(r',c') \neq \emptyset$.

	As both $s$ and $t$ have a deficit, there exist $a \in U$ and $b \in V$ such that $\{s,a\}$ and $\{t,b\}$ are co-edges in $G$. If either $a = t$ or $b = s$, we are done. Therefore, we may assume that $\{a,b\}$ is a co-edge in $G$. Now consider the graph $G' = G + \{s,a\} + \{t,b\}$. Let 
	\[
	X_1 = \{ v\ |\ \{a,v\} \in G' \} \subseteq V
	\] 
	be the set of all neighbors of $a$, and similarly $X_2$ the set of all neighbors of $b$ in $U$. We may assume that the sets $X_1$ and $X_2$ form a complete bipartite graph in $G'$, otherwise there is an alternating path of length $3$ and we are done. Define $Y_i$ as the set of nodes with at least one co-edge in $X_{i+1}$. We may assume that the sets $Y_1$ and $Y_2$ form an independent set in $G'$, otherwise there is an alternating path of length $5$ and we are done.
	Furthermore, let $Z_i$ be the set of nodes not in $X_i$ with a neighbor in $Y_{i+1}$. Also here, it must be the case that $Z_1$ and $Z_2$ form a complete bipartite graph, otherwise there is an alternating path of length $7$ and we are done.  Moreover, by definition of $Y_i$, all nodes in $Z_i$ are connected to all nodes in $X_{i+1}$, and so the sets $K_1 = X_1 \cup Z_1$ and $K_2 = X_2 \cup Z_2$ form a complete bipartite graph. Finally, we let $R_1 = V \mysetminus (K_1 \cup Y_1)$ and $R_2 = U \mysetminus (K_2 \cup Y_2)$.

	Before we continue, let us introduce some additional notation. For subsets $A \subseteq V$ and $B \subseteq U$ we write $E(A,B)$ for the set of edges with one endpoint in $A$ and one endpoint in $B$.  Moreover, for all sets that were introduced, lower case letters are used to represent the cardinalities of these sets, e.g., $y_1 = |Y_1|$. We will upper and lower bound the quantity $|E(K_1,R_2 \cup Y_2)|$.

	Now, by definition all neighbors of nodes in $Y_2$ are part of $K_1$, so that
	\[
	|E(K_1,Y_2)| \geq \delta_{c} \cdot y_2 \,.
	\]
	Also, by definition all nodes in $R_1$ are adjacent to all nodes in $X_1$, and, in combination with the fact that $|X_1| \geq \delta_{c} + 1$ (as $t$ has a degree surplus of one), this implies that
	\[
	|E(K_1,R_2)| \geq |E(X_1,R_2)| \geq (\delta_{c}+1)\cdot r_2 \,.
	\]
	Combining these results, we find that
	\begin{equation}\label{eq:bip1}
	|E(K_1,R_2 \cup Y_2)| > \delta_{c} (r_2 + y_2) = \delta_{c} (n - k_2) \,.
	\end{equation}
	We now continue with an upper bound. As $K_1$ and $K_2$ form a complete bipartite graph, we have
	\begin{equation}\label{eq:bip2}
	|E(K_1,R_2 \cup Y_2)| \leq k_1(\Delta_{r} - k_2)  \,.
	\end{equation}
	Combining (\ref{eq:bip1}) and (\ref{eq:bip2}) we find
	\begin{equation}\label{eq:bip3}
	\delta_{c} (n - k_2) < k_1(\Delta_{r} - k_2) \,.
	\end{equation}
	Similarly, by interchanging the roles of $U$ and $V$, we find
	\begin{equation}\label{eq:bip4}
	\delta_{r} (m - k_1) < k_2(\Delta_{c} - k_1) \,.
	\end{equation}
	We now make a case distinction depending on the sizes of $k_1$ and $k_2$.
	
	\textbf{Case 1: $k_1 \leq k_2$.} Then \eqref{eq:bip3} implies that
	\[
	\delta_{c} (n - k_2) < k_2(\Delta_{r} - k_2),
	\]
	which can be rewritten as
	\[
	f(k_2) := k_2^2 - (\Delta_{r} + \delta_{c})k_2 + \delta_{c} n < 0.
	\]
	This implies that the function $f(k)$ has two null points, and therefore its discriminant must be strictly greater than zero. However, its discriminant is equal to $(\Delta_{r} + \delta_{c})^2 - 4\delta_{c}n = (\Delta_{r} - \delta_{c})^2 - 4\delta_{c}(n - \Delta_{r}) \leq 0$ which is a contradiction. Recall that the last inequality holds by assumption. 
	
	\textbf{Case 2: $k_2 \leq k_1$.} We can repeat the analysis of Case 1 but based on (\ref{eq:bip4}) instead of \eqref{eq:bip3}. 
\end{proof}

If $\Delta_{r} = \Delta_{c}$, $\delta_{r} = \delta_{c}$ and $m = n$, the condition in (\ref{eq:bip_stable}) reduces to
\begin{equation}\label{eq:bip_same}
(\Delta_{r} - \delta_{r})^2 \leq 4\delta_{r}(m - \Delta_{r}).
\end{equation}
In particular, condition \eqref{eq:bip_same} is satisfied if $\delta_{r} \geq m/4$ and $\Delta_{r} \leq 3m/4 - 1$. This strictly improves the range of $\delta_{r} \geq m/3$ and $\Delta_{r} \leq 2m/3$ for the biparte case recently obtained in \cite{ErdosMMS2018}. 
It should be noted that Corollary \ref{cor:bipartite} does not capture the (almost) half-regular case  \cite{Erdos2013,Erdos2015decomposition}.\footnote{A pair $(r,c)$ is \emph{half-regular} if either $r$ of $c$ is regular, and \emph{almost half-regular} if either $\Delta_{r} \leq \delta_{r} + 1$ or $\Delta_{c} \leq \delta_{c} + 1$.} Those cases can be captured by the next result. In \cite{ErdosMMS2018} it was recently shown that the switch chain  is rapidly mixing for the condition in (\ref{eq:ErdosMMS2018}). We now provide an alternative proof for this fact by showing that these bipartite degree sequences are strongly stable.

\begin{corollary}\label{cor:bipartite2}
	Let $\mathcal{D} = \mathcal{D}(\delta_{r},\Delta_{r},\delta_{c},\Delta_{c})$ be the set of all graphical bipartite degree sequences $(r,c)$  on $m$ and $n$ nodes respectively, satisfying 
	\begin{equation}\label{eq:ErdosMMS2018}
	(\Delta_{c} - \delta_{c} - 1)(\Delta_{r} - \delta_{r} - 1) < 1 +  \max\big\{\delta_{c}(n - \Delta_{r}),\delta_{r}(m - \Delta_{c})\big\}
	\end{equation}	
	where $\delta_{r},\Delta_{r}$ are the minimum and maximum component of $r$, and $\delta_{c},\Delta_{c}$ the minimum and maximum component of $c$ respectively. For any $(r,c) \in \mathcal{D}$, we have $k_{JS}(r,c) \leq 8$. Hence the switch chain is rapidly mixing for sequences in $\mathcal{D}$.
\end{corollary}

\begin{proof}
	The proof bears similarities with the the proof of Corollary \ref{cor:bipartite}, 
	especially in the beginning. In fact, let $s , t, G, a, b , G', X_1, X_2, Y_1$ and $Y_2$ be exactly as in the proof of Corollary \ref{cor:bipartite}.
	Also, let $R_1 =  V \mysetminus (X_1 \cup Y_1 \cup \{b\})$ and $R_2 = U  \mysetminus (X_2 \cup Y_2 \cup \{a\})$ and note that, by the definition of $Y_1$, $R_1$ and $X_2$ form a complete bipartite graph in $G'$; the same is true for $X_1$ and $R_2$.  
	
	Before we continue, we introduce some additional notation. For subsets $A \subseteq V$ and $B \subseteq U$ we write $\bar{E}(A,B)$ for the set of co-edges in $G'$ with one endpoint in $A$ and one endpoint in $B$.  Like in the proof of Corollary \ref{cor:bipartite}, for all sets that were introduced, lower case letters are used to represent the cardinalities of these sets, e.g., $y_1 = |Y_1|$. 
	
	We will upper and lower bound the quantity $|\bar{E}(Y_1,U)|$. The fact that the degree of each node in $Y_1$ is at least $\delta_{r}$, gives us an obvious upper bound: 
	\begin{equation}\label{eq:xbip1}
	|\bar{E}(Y_1,U)| \le y_1 (n-\delta_{r}) \,.
	\end{equation}
	Since $Y_1$ and $Y_2$ form an independent set and, by definition, no node in $Y_1$ is adjacent to $a$, we have
	\begin{equation}\label{eq:xbip2}
	|\bar{E}(Y_1,U)| = y_1 (y_2 + 1) + |\bar{E}(Y_1,U - Y_2 - \{a\})| \,.
	\end{equation}	
	So, we need to lower bound the quantity $|\bar{E}(Y_1,U - Y_2 - \{a\})|= |\bar{E}(Y_1,X_2 \cup R_2)|$.	Using straightforward upper and lower bounds and the fact that $|\bar{E}(X_1,R_2)| = |\bar{E}(X_1 \cup R_1 \cup \{b\}, X_2)| = 0$,\footnote{This follows from the definitions of the sets $X, Y$ and $R$, and the fact that $X_1$ and $X_2$ form a complete bipartite graph.} we have
	\begin{IEEEeqnarray}{rCl}
		|\bar{E}(Y_1,X_2 \cup R_2)| & = & |\bar{E}(V,X_2 \cup R_2)| - |\bar{E}(X_1 \cup R_1 \cup \{b\},X_2 \cup R_2)|       \nonumber                  \\
		& \ge & (m- \Delta_{c})(n-y_2-1) - |\bar{E}(X_1 \cup R_1 \cup \{b\}, X_2)| - |\bar{E}(X_1,R_2)| - |\bar{E}(R_1 \cup \{b\},R_2)| \nonumber \\ 
		& \ge &  (m- \Delta_{c})(n-y_2-1) - (r_1+1)r_2\,. \label{eq:xbip3}
	\end{IEEEeqnarray}
	Now we may combine \eqref{eq:xbip1}, \eqref{eq:xbip2} and \eqref{eq:xbip3} to get
	\begin{equation*}
	y_1 (n-\delta_{r}) \ge y_1 (y_2 + 1) + (m- \Delta_{c})(n-y_2-1) - (r_1+1)r_2 \,,
	\end{equation*}
	and by rearranging terms
	\begin{equation}\label{eq:xbip4}
	(r_1+1)r_2 \ge (y_2 + 1 - n) (y_1 - m + \Delta_{c}) + \delta_{r} y_1  \,.
	\end{equation}
	
	Next we derive bounds on the sizes of the different sets involved in the proof so far. First note that $a$ has at least $\delta_{c} + 1$ neighbors as it has degree surplus of one. Thus, $x_1\ge \delta_{c} + 1$. Similarly, $x_2\ge \delta_{r} + 1$. 
	To bound $y_1$, notice that a node in $X_2$ has at least $m - \Delta_{c}$ non-neighbors in $V$ and all of them must be in $Y_1$. Thus, $y_1\ge m - \Delta_{c}$. 
	On the other hand,  $Y_1 \cup R_1$ contains the non-neighbors of $a$ (except from $b$) and these can be at most $m - \delta_{c} - 2$. Thus, $y_1\le m - \delta_{c} - r_1 - 2$. Similarly, $n - \Delta_{r} \le y_2 \le n - \delta_{r} - r_2 - 2$.

	We are going to combine the latter bounds with \eqref{eq:xbip4}:
	\begin{equation*}
	(n - \delta_{r} - 2 - y_2)(m - \delta_{c} - 1 - y_1)   \ge   (y_2 + 1 - n) (y_1 - m + \Delta_{c}) + \delta_{r} y_1  \,.
	\end{equation*}
	By multiplying everything out and canceling several terms, we get
	\begin{equation*}
	-n\delta_{c} - n - m \delta_{r} + \delta_{r}\delta_{c} +\delta_{r}-m+2\delta_{c} +2+y_1+\delta_{c}y_2 +y_2  \ge   \Delta_{c}(y_2 + 1 - n)\,,
	\end{equation*}
	and by rearranging 
	\begin{equation*}
	( n  -y_2- 1) (\Delta_{c} - \delta_{c}-1 ) +( y_1 - m + \delta_{c}+1)  \ge \delta_{r} (m - \delta_{c}-1) \,.
	\end{equation*}
	Next we may use $y_2\ge n - \Delta_{r}$ and  $y_1\le m - \delta_{c}  - 2$ to finally get
	\begin{equation*}
	(\Delta_{r} - 1)(\Delta_{c} - \delta_{c} - 1) -1 \geq  \delta_{r} (m - \Delta_{c}) + \delta_{r} (\Delta_{c} - \delta_{c}-1)\,,
	\end{equation*}
	which is equivalent to
	\begin{equation*}
	(\Delta_{r} - \delta_{r} - 1)(\Delta_{c} - \delta_{c} - 1) \geq 1 + \delta_{r} (m - \Delta_{c})\,.
	\end{equation*}
	Using the exact same arguments, we can work with $|\bar{E}(V,Y_2)|$ to get 
	\begin{equation*}
	(\Delta_{c} - \delta_{c} - 1)(\Delta_{r} - \delta_{r} - 1) \geq 1 +  \delta_{c}(n - \Delta_{r}) \,.
	\end{equation*}
	The last two inequalities contradict the choice of $(r, c)$.
\end{proof}

%\newpage
\bigskip

%%%%%%%%%%%%%%%%%%%%%%%%%%%%%%%%%%%%
%%%%%%%%%%%%%%%%%%%%%%%%%%%%%%%%%%%%
\section{The JDM Model: An Example}\label{app:example}
%%%%%%%%%%%%%%%%%%%%%%%%%%%%%%%%%%%%
%%%%%%%%%%%%%%%%%%%%%%%%%%%%%%%%%%%%
In this section we provide an example of the joint degree matrix model with two degree classes. We let $V = \{1,\dots,11\}$ and consider the partition given by $V_1 = \{1,\dots,6\}$ and $V_2 = \{7,\dots,11\}$, i.e., $q = 2$. We let
\[
c = \begin{pmatrix} 7 & 4 \\ 4 & 8
\end{pmatrix} \ \ \text{ and } \ \ d = (3,4).
\]
This means that the nodes in $V_1$ have degree three, the nodes in $V_2$ degree four, and there are in total four edges between the nodes of $V_1$ and $V_2$. In Figure \ref{fig:jdm_example} below  we  give a possible graphical realization of this tuple $(c,d)$.

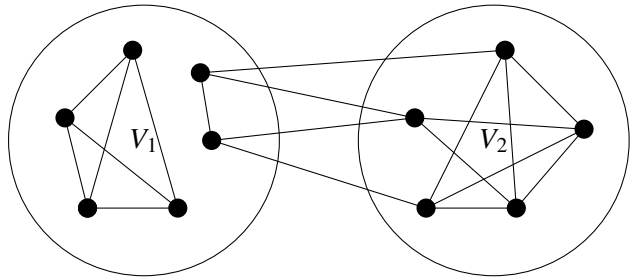
\begin{figure}[ht!]
\centering
\begin{tikzpicture}[scale=0.6]
\coordinate (a) at (-0.5,0); 
\coordinate (b) at (1.5,0);
\coordinate (c) at (-1,2); 
\coordinate (d) at (2.25,1.5); 
\coordinate (f) at (0.5,3.5); 
\coordinate (e) at (2,3); 

\coordinate (a2) at (7,0); 
\coordinate (b2) at (9,0);
\coordinate (c2) at (6.75,2); 
\coordinate (d2) at (10.5,1.75); 
\coordinate (e2) at (8.75,3.5);

\node at (a) [circle,scale=0.7,fill=black] {};

%\node (A) [below=0.1cm of a]  {};
\node at (a) [circle,scale=0.7,fill=black] {};
\node at (b) [circle,scale=0.7,fill=black] {};
\node at (c) [circle,scale=0.7,fill=black] {};
\node at (d) [circle,scale=0.7,fill=black] {};
\node at (f) [circle,scale=0.7,fill=black] {};
\node at (e) [circle,scale=0.7,fill=black] {};

\node at (a2) [circle,scale=0.7,fill=black] {};
\node at (b2) [circle,scale=0.7,fill=black] {};
\node at (c2) [circle,scale=0.7,fill=black] {};
\node at (d2) [circle,scale=0.7,fill=black] {};
\node at (e2) [circle,scale=0.7,fill=black] {};

\path[every node/.style={sloped,anchor=south,auto=false}]
(a) edge[-,] node {} (b)
(a) edge[-,] node {} (c)
(a) edge[-,] node {} (f)
(b) edge[-,] node {} (c)
(b) edge[-,] node {} (f)
(c) edge[-,] node {} (f)
(d) edge[-,] node {} (e)
(a2) edge[-,] node {} (b2)
(a2) edge[-,] node {} (d2)
(a2) edge[-,] node {} (e2)
(b2) edge[-,] node {} (c2)
(b2) edge[-,] node {} (d2)
(b2) edge[-,] node {} (e2)
(c2) edge[-,] node {} (d2)
%(c2) edge[-,] node {} (e2)
(d2) edge[-,] node {} (e2)
(e) edge[-,] node {} (c2)
(e) edge[-,] node {} (e2)
(d) edge[-,] node {} (a2)
(d) edge[-,] node {} (c2);

\draw (0.75,1.5) circle (3cm) node{$V_1$};
\draw (8.5,1.5) circle (3cm) node{$V_2$};
\end{tikzpicture}
\caption{An example of a graphical realization for $c$ and $d$ as given above.}
\label{fig:jdm_example}
\end{figure}

\newpage
%%%%%%%%%%%%%%%%%%%%%%%%%%%%%%%%%%%%%%%%
%%%%%%%%%%%%%%%%%%%%%%%%%%%%%%%%%%%%%%%%
\section{The PAM Model and the Hinge Flip Markov Chain}\label{sec:auxiliary}
%%%%%%%%%%%%%%%%%%%%%%%%%%%%%%%%%%%%%%%%
%%%%%%%%%%%%%%%%%%%%%%%%%%%%%%%%%%%%%%%%
We start with a general description of the Partition Adjacency Matrix (PAM) model as a generalization of the JDM model. Let $V = \{1,\dots,n\}$ be a given set. An instance of the partition adjacency matrix model is given by a partition $V_1 \cup V_2 \cup \dots \cup V_q$ of $V$ into pairwise disjoint classes. Moreover, we are given a symmetric \emph{partition adjacency matrix} $c = (c_{ij})_{i,j \in [q]}$  of non-negative integers, and a sequence $d = (d_1,\dots,d_n)$ of non-negative integers. We say that the tuple $((V_i)_{i \in q},c,d)$ is graphical if there exists a simple, undirected, labelled graph $G = (V,E)$ on the nodes in $V$ with node $i \in V$ having degree $d_i$, and so that there are precisely $c_{ij}$ edges between endpoints in $V_i$ and $V_j$. This is denoted by $E[V_i,V_j] = c_{ij}$. The graph $G$ is called a graphical realization of the tuple. We let $\mathcal{G}((V_i)_{i \in q},c,d)$ denote the set of all graphical realizations of the tuple $((V_i)_{i \in q},c,d)$. We often write $\mathcal{G}(c,d)$ instead of $\mathcal{G}((V_i)_{i \in [q]},c,d)$ when it is  clear what the partition is. 

%One special case of the partition adjacency model is when the degrees are component-wise regular, i.e., when $d_a = d_b$ whenever $a,b \in V_i$ for $i = 1,\dots,q$. This corresponds to the \emph{joint degree matrix} model \cite{Amanatidis2015}.
In this work we focus on the case of a partition into two classes $V_1$ and $V_2$, and, without loss of generality, assume that $1 \leq c_{12} \leq |V_1|\cdot|V_2| - 1$.\footnote{It is not hard to see that the cases $c_{12} \in \{0,|V_1|\cdot|V_2|\}$ reduce to the single class case.} For the case of two classes an initial state can be computed in polynomial time \cite{ErdosHIM2017}.\footnote{For general instances, it is not known if an initial state can be computed in time polynomial in $n$. It is conjectured to be NP-hard  in general \cite{ErdosHIM2017}, see also \cite{CzabarkaSTW2017}.} 
We let $\mathcal{G}'(c,d) = \cup_{(c',d')}\mathcal{G}'(c',d')$ with $(c',d')$ ranging over tuples satisfying
%\[
%\bigg\{ (\gamma',d') : \sum_{i=1}^n |d_i - d_i'| \in \{0,2,4\} \text{ and } \gamma' \in \{\gamma'-1,\gamma',\gamma'+1\} \bigg\} \,.
%\]
\begin{enumerate}[label=(\roman*)]
\item $\sum_{i=1}^n d_i - d_i' = 0$,
\item $\sum_{i=1}^n |d_i - d_i'| \in \{0,2,4\}$,
\item $c_{12}' \in \{c_{12}-1,c_{12},c_{12}+1\}$. %\ \ (or $\gamma' \in \{\gamma - 1,\gamma,\gamma+1\}$).
\end{enumerate}

We call elements in $\mathcal{G}'(c,d) \setminus \mathcal{G}(c,d)$ \emph{perturbed (auxiliary) states}. For any $G \in \mathcal{G}'(c,d)$ the \emph{perturbation at node $v \in V$} is defined as $\alpha_v = d_v - d_v'$ where $d'$ is the degree sequence of $G$. We say that the node $v$ has a degree surplus if $\alpha_v > 0$ and a degree deficit if $\alpha_v < 0$. Moreover, the total degree surplus is defined as $\sum_{v : \alpha_v > 0} \alpha_v$, and the total degree deficit as $- \sum_{v : \alpha_v < 0} \alpha_v$. Note that
\[
\sum_{v : \alpha_v > 0} \alpha_v  - \sum_{v : \alpha_v < 0} \alpha_v = \sum_{i=1}^n |d_i - d_i'|.
\]
Finally, we say that a tuple $(c',d')$ is \emph{edge-balanced} if $c' = c$ (but possibly $d' \neq d$). From the conditions defining $\mathcal{G}'(c,d)$, we may infer the following properties. %In particular, the second property implies that if the number of cut edges in a graphical realization is off by at most one edge, then this is also the case for the number of edges with both endpoints in $V_1$, and the number of edges with both endpoints in $V_2$.

\begin{proposition}\label{prop:basic_properties} For any $G \in \mathcal{G}'(c,d)$, for some tuple $(c',d')$, it holds that
\begin{enumerate}[label=(\alph*)]
\item the perturbation at node $v$ satisfies $\alpha_{v} \in \{-2,-1,0,1,2\}$ for any $v \in V$, 
\item $\max_{i,j=1,2} |c_{ij} - c_{ij}'| \leq 1$, and $\sum_{1 \leq i < j \leq 2} |c_{ij} - c_{ij}'| \in \{0,2\}$.
\end{enumerate}
\end{proposition}
\begin{proof}
If there is some node with degree surplus greater or equal than three, then the total degree deficit is also at least three, which follows from the first condition defining $\mathcal{G}'(\gamma,d)$. This means that $\sum_{i=1}^n |d_i - d_i'| \geq 6$, which violates the second condition defining $\mathcal{G}'(c,d)$. A similar argument holds in case there is some node with degree deficit greater or equal than three. To see that the second property is true, assume without loss of generality that $c_{11}' \geq c_{11} + 2$ (similar arguments hold for $c_{22}'$). Because of the fact that $c_{12}' \in \{c_{12}-1,c_{12},c_{12}+1\}$, by the third condition defining $\mathcal{G}'(c,d)$, it must be that the total degree surplus of the nodes in $V_1$ is at least three. This gives a contradiction for similar reasons as before. An analogous argument holds in case $c_{11}' \leq c_{11} - 2$.  Finally, the last property is a direct consequence of $\max_{i,j=1,2} |c_{ij} - c_{ij}'| \leq 1$ and the fact that $\sum_{i,j=1,2} |c_{ij} - c_{ij}'|$ is an even number, because of the first property defining $\mathcal{G}'(c,d)$.
\end{proof}

\begin{remark}
As we focus on the case in which $V$ is partitioned into two classes $V_1$ and $V_2$ here, we will use some shorthand notation in this section.\footnote{Still, the notation in Proposition \ref{prop:basic_properties} is more convenient in Appendix \ref{sec:stable_jdm}.} Given a sequence $d$ the number $\gamma = c_{12}$ uniquely determines the matrix $c$, and the set $\mathcal{G}(c,d)$ is then denoted by $\mathcal{G}(V_1,V_2,\gamma,d)$. As before, we often leave out $V_1$ and $V_2$ from the tuple (for sake of readability). That is, we then write $\mathcal{G}'(\gamma,d)$ instead of $\mathcal{G}'(V_1,V_2,c,d)$. We do this in order to emphasize that the derived results only hold for two classes.
\end{remark}
 
We define the \emph{hinge flip Markov chain}\footnote{One could also choose a natural generalization of the JS chain here. Similarly, for the PAM model with one class as studied in Section \ref{sec:main_result}, one could also use a hinge flip based Markov chain. We choose to work with the JS chain in Section \ref{sec:main_result} as this is the formulation introduced in \cite{Jerrum1990} (and already claimed to be rapidly mixing for various degree sequences).} $\mathcal{M}(\gamma,d)$ on $\mathcal{G}'(\gamma,d)$ as follows. Suppose the current state of the chain is $G \in \mathcal{G}'(\gamma,d)$:
\begin{itemize}
\item With probability $1/2$, do nothing.
\item Otherwise, perform a \emph{hinge flip} operation: select an ordered triple $i,j,k$ of nodes uniformly at random. If $\{i,j\} \in E(G)$, $\{j,k\} \notin E(G)$, and $G - \{i,j\} + \{j,k\} \in \mathcal{G}'(\gamma,d)$, then delete $\{i,j\}$ and add $\{j,k\}$. 
\end{itemize}
Note that we can check if $G - \{i,j\} + \{j,k\} \in \mathcal{G}'(\gamma,d)$ in time polynomial in $n$ based on the  state $G$.

\begin{figure}[ht!]
\centering
\begin{tikzpicture}[scale=0.6]
\coordinate (i) at (0,0); 
\coordinate (k) at (4,0);
\coordinate (j) at (2,2);

\coordinate (P1) at (5,1);
\coordinate (P2) at (6,1);

\node at (i) [circle,scale=0.7,fill=black] {};
\node (i1) [left=0.1cm of i]  {$i$};
\node at (j) [circle,scale=0.7,fill=black] {};
\node (j1) [left=0.1cm of j]  {$j$};
\node at (k) [circle,scale=0.7,fill=black] {};
\node (k1) [left=0.1cm of k]  {$k$};

\path[every node/.style={sloped,anchor=south,auto=false}]
(i) edge[-,very thick] node {} (j)
(P1) edge[->,very thick] node {} (P2);   
\end{tikzpicture}
\quad
\begin{tikzpicture}[scale=0.6]
\coordinate (i) at (0,0); 
\coordinate (k) at (4,0);
\coordinate (j) at (2,2);

\node at (i) [circle,scale=0.7,fill=black] {};
\node (i1) [left=0.1cm of i]  {$i$};
\node at (j) [circle,scale=0.7,fill=black] {};
\node (j1) [left=0.1cm of j]  {$j$};
\node at (k) [circle,scale=0.7,fill=black] {};
\node (k1) [left=0.1cm of k]  {$k$};

\path[every node/.style={sloped,anchor=south,auto=false}]
(j) edge[-,very thick] node {} (k);  
\end{tikzpicture}
\caption{Example of a hinge flip operation for the ordered triple $i,j,k$.}
%\label{fig:switch}
\end{figure}
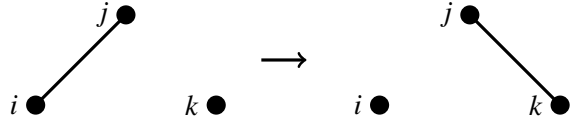

%This chain will henceforth be referred to as the \emph{JS chain on $\mathcal{G}'(d)$}.\footnote{A slightly different definition of stability is given by Jerrum, McKay and Sinclair \cite{Jerrum1989graphical}. Based on this variant, one could define the corresponding variant of the JS chain. Nevertheless, the definitions of stability in \cite{Jerrum1989graphical} and \cite{Jerrum1990} (and their corresponding definitions of strong stability) are equivalent. To avoid confusion, here we only use the definitions in \cite{Jerrum1990}.} 
Graphs $G,G' \in \mathcal{G}'(\gamma,d)$ are said to be \emph{adjacent} in $\mathcal{M}$ if $G$ can be obtained from $G'$ with positive probability in one transition of the chain $\mathcal{M}$. We say that two graphs $G, G'$ are within distance $r$ in $\mathcal{M}$ if there exists a path of at most length $r$ from $G$ to $G'$ in the state space graph of $\mathcal{M}$. By $\text{dist}(G,\gamma,d)$ we denote the  minimum distance of $G$ from an element in $\mathcal{G}(\gamma,d)$.
The following parameter is the analog of \eqref{eq:distance} for the current setting and will be used in a similar manner to define the appropriate variant of strong stability. %\pkblue{Different definition needed here! Give the combinatorial interpretation in terms of short alternating path(s).}  
We define 
\begin{equation}%\label{eq:distance}
k(\gamma,d) = \max_{G \in \mathcal{G}'(\gamma,d)} \text{dist}(G,\gamma,d).
\end{equation}
%Based on the parameter $k$, we define the notion of \emph{strong stability}. 
In the PAM model with two degree classes, a family $\mathcal{D}$ of graphical tuples $(\gamma,d)$ is called %\emph{strongly stable}\footnote{This is abuse of notation with respect to the notion of strong stability in Section \ref{sec:degree_sequence}, but both definitions are conceptually similar.} 
if there exists a constant $k$ such that $k(\gamma,d) \leq k$ for all $(\gamma,d) \in \mathcal{D}$. %The simple observation in Proposition \ref{prop:strong_stable} justifies the terminology.

\begin{theorem} %\yarem{Make poly(n) in the theorem explicit? \pkblue{Did $n^3$.}}
Let $\mathcal{D}$ be a family of graphical tuples that is strongly stable with respect to some constant $k$. Then for every $(\gamma,d) \in \mathcal{D}$, the chain $\mathcal{M}(\gamma,d)$ is irreducible, aperiodic and symmetric, and, hence, has uniform stationary distribution over $\mathcal{G}'(\gamma,d)$. Moreover, $P(G,G')^{-1} \leq n^3$ for all adjacent $G,G' \in \mathcal{G}'(\gamma,d)$, and also 
the maximum in- and out-degrees of the state space graph of the chain $\mathcal{M}(\gamma,d)$ are bounded by $n^3$.\footnote{It might be the case that the chain is always irreducible, even if $\mathcal{D}$ is not strongly stable, but this is not relevant at this point. The assumption of strong stability allows for a shortcut in the proof of irreducibility.}
\end{theorem}
\begin{proof}
The only claim that requires a detailed argument, and uses the assumption of strong stability, is that of the irreducibility of the chain. By definition of strong stability, we always know that every perturbed state is connected to some element in $\mathcal{G}(\gamma,d)$ so it suffices to show that there is a path between any two states in $\mathcal{G}(\gamma,d)$. This follows from the analysis in the remainder of this section. Aperiodicity follows from the holding probability in the description of the chain $\mathcal{M}$, and symmetry is straightforward. The bound on $P(G,G^{-1})$ follows directly from the description of the chain, as well as the bound on the in- and out-degrees of the state space graph. 
\end{proof}

The remainder of this section is devoted to proving Theorem \ref{thm:auxiliary} based on ideas introduced in \cite{Bhatnagar2008}. Throughout this section we always consider tuples $(\gamma,d)$ coming from strongly stable families.

\begin{theorem}\label{thm:auxiliary}
Let $\mathcal{D}$ be a strongly stable family of tuples $(\gamma,d)$ with respect to some constant $k$. Then there exist polynomials $p(n), r(n)$ such that for any $(\gamma,d) \in \mathcal{D}$, with $d = (d_1,\dots,d_n)$, there exists an efficient multicommodity flow $f$ for the auxiliary chain $\mathcal{M}(\gamma,d)$ on $\mathcal{G}'(\gamma,d)$   	
%\yarem{Make poly(n) explicit.\pkblue{This okay? For actual bound the proof in appendix has to be made explicit.}}
satisfying $\max_e f(e) \leq p(n)/ |\mathcal{G}'(\gamma,d)|$  and  $\ell(f) \leq r(n)$. Hence, the chain $\mathcal{M}(\gamma,d)$ is rapidly mixing for families of strongly stable tuples.
\end{theorem}

%\begin{corollary}
%There exists a \emph{fully polynomial randomized approximate sampler (fpras)} for uniformly sampling graphical realizations in $\mathcal{G}(\gamma,d)$ when $(\gamma,d)$ comes from a family of strongly stable tuples $\mathcal{D}$.
%\end{corollary}

We will use the following lemma in order to simplify the proof of Theorem \ref{thm:auxiliary}. It is the analog of Lemma \ref{lem:flow_simplification} in Subsection \ref{app:js}. %It is  similar to an idea used in \cite{Jerrum1989} for the analysis of a Markov chain that can be used to sample perfect matchings in a given dense graph. %in order to restrict ourselves to establishing flow between states in $\mathcal{G}(d)$, rather than between all states in $\mathcal{G}'(d)$. 

\begin{lemma}\label{lem:flow_simplification_jdm}
Let $f'$ be a flow that routes $1/|\mathcal{G}'(\gamma,d)|^2$ units of flow between any pair of states in $\mathcal{G}(\gamma,d)$ in the chain $\mathcal{M}(\gamma,d)$, so that $f'(e) \leq b/|\mathcal{G}'(\gamma,d)|$ for all $e$ in the state space graph of $\mathcal{M}(\gamma,d)$. Then $f'$ can be extended to a flow $f$ that routes $1/|\mathcal{G}'(\gamma,d)|^2$ units of flow between any pair of states in $\mathcal{G}'(\gamma,d)$ with the property that for all $e$
	\begin{equation}\label{eq:extending}
		f(e) \leq q(n) \frac{b}{|\mathcal{G}'(\gamma,d)|} \,,
	\end{equation}
	where $q(\cdot)$ is a polynomial whose degree only depends on $k(\gamma,d)\ (\leq k)$. Moreover, $\ell(f) \le \ell(f') + 2k(\gamma,d)$. %\footnote{We omit the proof of Lemma \ref{lem:flow_simplification} as the lemma is actually not needed. Careful consideration of the proof of Theorem \ref{thm:transformation} shows that we can only focus on flow between states in $\mathcal{G}(d)$, since the flow $h$ given in the proof of Theorem \ref{thm:transformation} only has positive flow between states corresponding to elements in $\mathcal{G}(d)$. That is, when defining the flow $h$, we essentially forget about all flow in $f$ between any pair of states where at least one state is an auxiliary state, i.e., an element of $\mathcal{G}'(d) \mysetminus \mathcal{G}(d)$. Said differently, in Theorem \ref{thm:transformation} we could start with the assumption that $f$ routes $1/|\mathcal{G}'(d)|^2$ units of flow between any pair of states in $\mathcal{G}(d)$ in the state space graph of the JS chain, and then the transformation still works. However, the formulations of Theorems \ref{thm:js_mixing} and \ref{thm:transformation} are more natural to describe a comparison between the JS and switch chain.}
\end{lemma}
%\pkrem{Do you consider this sketch sufficient? Formalizing everything is a bit tedious...}
\begin{proof}[Proof (sketch)]
We extend the flow $f'$ to $f$ as follows. For any $G \in \mathcal{G}'(\gamma,d) \setminus \mathcal{G}(\gamma,d)$ fix some $\phi(G) \in \mathcal{G}(\gamma,d)$ within distance $k$ of $G$ (which exists by assumption of strong stability), and fix some path in the state space graph from $G$ to $\phi(G)$. Moreover, define $\phi(H) = H$ for all $H \in \mathcal{G}(\gamma,d)$. The flow between $G$ and any given $G' \in \mathcal{G}'(\gamma,d)$ is now send as follows. 

First route $1/|\mathcal{G}'(\gamma,d)|^2$ units of flow from $G$ to $\phi(G)$ over the fixed path from $G$ to $\phi(G)$. Then use the flow-carrying paths used to send $1/|\mathcal{G}'(\gamma,d)|^2$ units of flow  between $\phi(G)$ and $\phi(G')$ as in the flow $f'$ (note that in general multiple paths might be used for this in the flow $f'$). Finally, use the reverse of the fixed path from $G'$ to $\phi(G')$ to route $1/|\mathcal{G}'(\gamma,d)|^2$ from $\phi(G')$ to $G'$.  For any $H \in \mathcal{G}(\gamma,d)$, we have $|\phi^{-1}(H)| \leq \text{poly}(n^k)$, as the in- and out-degrees of the nodes in the state space graph of $\mathcal{M}(\gamma,d)$ are polynomially bounded. It can then be shown (left to the reader) that this extension of $f'$, yielding the flow $f$, only gives an additional term of at most $\text{poly}(n^k) \frac{b}{|\mathcal{G}'(\gamma,d)|}$ to the congestion of every arc in the state space graph of the chain $\mathcal{M}(\gamma,d)$ in the flow $f'$. Hence, the extended flow $f$  satisfies (\ref{eq:extending}) for some appropriately chosen polynomial $q(n)$.
\end{proof}

\subsection{Proof of Theorem \ref{thm:auxiliary}}
Because of Lemma \ref{lem:flow_simplification_jdm} it now suffices to show that there exists a flow $f'$  that routes $1/|\mathcal{G}'(\gamma,d)|^2$ units of flow between any two pair of states in $\mathcal{G}(\gamma,d)$, in the state space graph of the chain $\mathcal{M}(\gamma,d)$, with the property that $f'(e) \leq p(n)/|\mathcal{G}'(\gamma,d)|$, and $\ell(f') \le q(n)$ for some polynomials $p(\cdot), q(\cdot)$ whose degrees may only depend on $k(\gamma,d)$. Note that $f'$ is not a feasible multi-commodity flow as defined in Section \ref{sec:preliminaries}, but should rather be interpreted as an intermediate auxiliary flow.
The proof of Theorem \ref{thm:auxiliary} will consist of multiple parts following, conceptually, the proof template in \cite{Cooper2007} developed for proving rapid mixing of the switch chain for regular graphs. The main difference is that for the so-called \emph{canonical paths} between states we rely on ideas introduced in \cite{Bhatnagar2008}. \medskip

\subsubsection{Canonical Paths}
We first introduce some basic terminology similar to that in \cite{Cooper2007}. Let  $V$ be a set of labeled vertices, let $\prec_{E}$ be a fixed total order on the set $\{ \{v,w\} : v,w \in V\}$ of edges, and let $\prec_{\mathcal{C}}$ be a total order on all \emph{circuits} on the complete graph $K_V$, i.e., $\prec_{\mathcal{C}}$ is  a total order on the closed walks in $K_V$ that visit every edge at most once. We fix for every circuit one of its vertices where the walk begins and ends.%\footnote{In \cite{Amanatidis2018} this is a bit imprecise as we only specify a starting point, but in the circuit processing procedure we never specify with with which blue edge around the starting point we start processing the cycle (one can just take the lexicographically smallest blue edge around the starting point for this).}

For given $G,G \in \mathcal{G}(\gamma,d)$, let $H = G \triangle G'$ be their symmetric difference. We refer to the edges in $G \mysetminus G'$ as \emph{blue}, and the edges in $G' \mysetminus G$ as \emph{red}.
A \emph{pairing} of red and blue edges in $H$ is a bijective mapping that, for each node $v \in V$, maps every red edge adjacent to $v$, to a blue edge adjacent to $v$. The set of all pairings is denoted by  $\Psi(G,G')$, and, with $\theta_v$ the number of red edges adjacent to $v$ (which is the same as the number of blue edges adjacent to $v$), we have 
$
|\Psi(G,G')| = \displaystyle{\Pi_{v \in V}} \theta_v!.
$ 

Remember that we are considering an instance of the PAM problem with two classes $V_1$ and $V_2$. For a given graphical realization $G \in \mathcal{G}(\gamma,d)$ we say that $e \in E(G)$ is a \emph{cut} edge if it has an endpoint in both $V_1$ and $V_2$. Otherwise we say that $e$ is an \emph{internal} edge, as both endpoints either lie both in the class $V_1$ or both in class $V_2$.

Similar to the approach in \cite{Cooper2007}, the goal is to construct for each pairing $\psi \in \Psi(G,G')$ a canonical path from $G$ to $G'$ that carries a fraction $|\Psi(G,G')|^{-1}$ of the total flow from $G$ to $G'$ in $f'$. For a given pairing $\psi$ and the total order $\prec_E$ given above, we first decompose $H$ into the edge-disjoint union of circuits in a canonical way. We start with the lexicographically least edge $w_0w_1$ in $E_H$ and follow the pairing $\psi$ until we reach the edge $w_kw_0$ that was paired with $w_0w_1$. This defines the circuit $C_1$ (which is indeed a closed walk). If $C_1 = E_H$, we are done. Otherwise, we pick the lexicographically least edge in $H \mysetminus C_1$ and repeat this procedure. We continue generating circuits until $E_H = C_1 \cup \dots \cup C_s$. Note that all circuits have even length and alternate between red and blue edges, and that they are pairwise edge-disjoint. 

We form a path from $G$ to $G'$ in the state space graph of the chain $\mathcal{M}(\gamma,d)$ by changing the blue edges of $G$ into the red edges of $G'$ using hinge flip operations. For certain parings this can be done in a straightforward way, but in general this is not the case. As a warm-up, we first consider a simple case (this case essentially describes how we would process the circuits in case there is only one class).

\paragraph{Warm-up Example.}
%We start explaining this processing step for a special case of the symmetric difference (and given pairing). 
If for every $i$, the circuit $C_i$ exclusively consist of internal edges, only within $V_1$ or only within $V_2$, or exclusively of cut edges, circuits can be processed according to the ordering $\prec_{\mathcal{C}}$ as follows.
%according to the total order $\prec_{\mathcal{C}}$.\footnote{The total order $\prec_{C}$ can be chosen in such a way that we always process circuit $C_i$ before $C_j$ if $i \leq j$ for any $G$ and $G'$, but this is not necessary.} The \emph{processing of a circuit $C$} means that all blue edges on $C$ are deleted, and all red edges of $C$ are added to the current graphical realization, using the three types of transitions in the JS chain mentioned at the beginning of this section. All other edges of the current graphical realization remain unchanged. In general, this can be done similar to the circuit processing procedure in \cite{Jerrum1989}.\footnote{This is the main difference between the switch chain analyses \cite{Cooper2007,Greenhill2017journal,Erdos2013,Erdos2015,Erdos2015decomposition,ErdosMMS2018} and our analysis. The processing of a circuit is much more complicated if performed directly in the switch chain.} 
Let $C = x_0x_1x_2\dots x_qx_0$ be a circuit, and assume w.l.o.g. that $x_0x_1$ is the lexicographically smallest blue edge adjacent to the starting node $x_0$ of the circuit. The processing of $C$ now consists of performing a sequence of hinge flips on the ordered pairs $(x_{i-1},x_i,x_{i+1})$ for $i = 1,\dots,q$ with the convention that $x_{q+1} = x_0$. 
%We first perform a type 0 transition in which we remove the blue edge $vx_1$; then we perform a sequence of type 1 transitions in which we add the red edge $x_ix_{i+1}$ and remove the blue edge $x_{i-1}x_{i}$ for $i = 1,3,\dots,q$; and finally we perform a type 2 transition in which we add the red edge $vx_q$. In particular, this means that the elements on the canonical path right before and after the processing of a circuit belong to $\mathcal{G}(d)$. All the intermediate elements that we visit during the processing of the circuit $C$ belong to $\mathcal{G}'(d) \mysetminus \mathcal{G}(d)$, i.e., every element has either precisely two nodes with degree deficit one, or one node with degree deficit two. 
This is illustrated in Figures \ref{fig:circuit_step1}, \ref{fig:step2_3} and \ref{fig:step4_5} for an example of $C$ as illustrated in Figure \ref{fig:circuit_step1} on the left.\footnote{This is similar to the procedure described in Figures \ref{fig:step1_2}, \ref{fig:step3_4} and \ref{fig:step5_6} for the JS chain (we give an example here as well for the sake of completeness).} We have also indicated the degree surplus and deficit at every step. By assumption, the edges of $C$ either are all internal edges, or all cut edges. Therefore, throughout the processing of $C$, we never violate the constraint that there should be $\gamma$ edges between the classes $V_1$ and $V_2$, and, in particular, this implies that every intermediate state is an element of $\mathcal{G}'(\gamma,d)$.

\begin{figure}[ht!]
\centering
\scalebox{0.8}{
\begin{tikzpicture}[
  ->,
  >=stealth',
  shorten >=1pt,
  auto,
  %node distance=2cm,
  semithick,
  every state/.style={circle,radius=0.1pt,text=black},
]
\begin{scope}
  \node[state]  (v)               					 		{$x_0/x_3$};
  \node[state]  (x1) [below left=0.7 and 1cm of v] 	 		{$x_1$};
  \node[state]  (x2) [above left=0.7 and 1cm of v]  			 		{$x_2$};
  \node[state]  (x4) [above right=0.7 and 1cm of v]					{$x_4$};
  \node[state]  (x5) [below right=0.7 and 1cm of x4] 			{$x_5/x_8$};
  \node[state]  (x6) [above right=0.7 and 1cm of x5] 			{$x_6$};
  \node[state]  (x7) [below right=0.7 and 1cm of x5] 	{$x_7$};
  \node[state]  (x9) [below right=0.7 and 1cm of v]				{$x_9$};       

\path[every node/.style={sloped,anchor=south,auto=false}]
%(v) edge[line width=1.5pt, bend left=25] 	node {} (x1)        
(v) edge[-, line width=2pt] 	node {} (x1)            
(x1) edge[-,dashed,line width=2pt] 	node {} (x2) 
(x2) edge[-,line width=2pt] 	node {} (v) 
(v) edge[-,dashed,line width=2pt] 	node {} (x4) 
(x4) edge[-,line width=2pt] 	node {} (x5) 
(x5) edge[-,dashed,line width=2pt] 	node {} (x6) 
(x6) edge[-,line width=2pt] 	node {} (x7) 
(x7) edge[-,dashed,line width=2pt] 	node {} (x5) 
(x5) edge[-,line width=2pt] 	node {} (x9) 
(x9) edge[-,dashed,line width=2pt] 	node {} (v);
        
\end{scope}
\end{tikzpicture}}
 \ \ \ \ \ \ \ \ \ \quad \ \ \ \ \ \ \ \ \ 
\scalebox{0.8}{
\begin{tikzpicture}[
  ->,
  >=stealth',
  shorten >=1pt,
  auto,
  %node distance=2cm,
  semithick,
  every state/.style={circle,radius=0.1pt,text=black},
]
\begin{scope}
  \node[state]  (v)               					 		{$x_0/x_3$};
  \node[state]  (x1) [below left=0.7 and 1cm of v] 	 		{$x_1$};
  \node[state]  (x2) [above left=0.7 and 1cm of v]  			 		{$x_2$};
  \node[state]  (x4) [above right=0.7 and 1cm of v]					{$x_4$};
  \node[state]  (x5) [below right=0.7 and 1cm of x4] 			{$x_5/x_8$};
  \node[state]  (x6) [above right=0.7 and 1cm of x5] 			{$x_6$};
  \node[state]  (x7) [below right=0.7 and 1cm of x5] 	{$x_7$};
  \node[state]  (x9) [below right=0.7 and 1cm of v]				{$x_9$};       

\node (1) [right=0.1cm of x2] {$+1$};  
\node (3) [right=0.1cm of v] {$-1$};  

\path[every node/.style={sloped,anchor=south,auto=false}]
%(v) edge[line width=1.5pt, bend left=25] 	node {} (x1)        
%(v) edge[-, line width=2pt] 	node {} (x1)            
(x1) edge[-,line width=3pt] 	node {} (x2) 
(x2) edge[-,line width=3pt] 	node {} (v) 
(v) edge[-,dashed,line width=3pt] 	node {} (x4) 
(x4) edge[-,line width=3pt] 	node {} (x5) 
(x5) edge[-,dashed,line width=3pt] 	node {} (x6) 
(x6) edge[-,line width=3pt] 	node {} (x7) 
(x7) edge[-,dashed,line width=3pt] 	node {} (x5) 
(x5) edge[-,line width=3pt] 	node {} (x9) 
(x9) edge[-,dashed,line width=3pt] 	node {} (v);
        
\end{scope}
\end{tikzpicture}}
\caption{The circuit $C = x_0x_1x_2x_3x_4x_5x_6x_7x_8x_9x_0$ with $x_0 = x_3$ and $x_5 = x_8$. The blue edges are represented by the solid edges, and the red edges by the dashed edges (left). The edge $x_0x_1$ is removed and $x_1x_2$ is added  (right).} 
\label{fig:circuit_step1}
\end{figure}
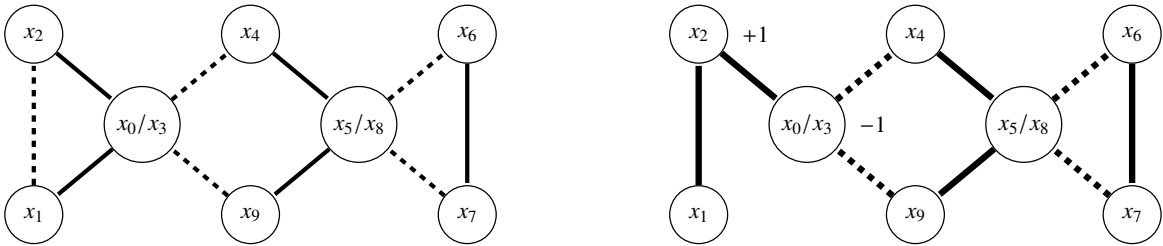 

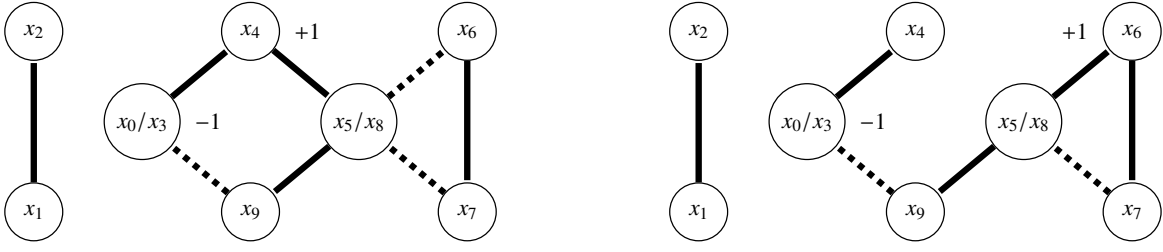
\begin{figure}[ht!]
\centering

%\caption{Edge $vx_1$ is removed.} 
%\label{fig:braess_5}

\scalebox{0.8}{
\begin{tikzpicture}[
  ->,
  >=stealth',
  shorten >=1pt,
  auto,
  %node distance=2cm,
  semithick,
  every state/.style={circle,radius=0.1pt,text=black},
]
\begin{scope}
  \node[state]  (v)               					 		{$x_0/x_3$};
  \node[state]  (x1) [below left=0.7 and 1cm of v] 	 		{$x_1$};
  \node[state]  (x2) [above left=0.7 and 1cm of v]  			 		{$x_2$};
  \node[state]  (x4) [above right=0.7 and 1cm of v]					{$x_4$};
  \node[state]  (x5) [below right=0.7 and 1cm of x4] 			{$x_5/x_8$};
  \node[state]  (x6) [above right=0.7 and 1cm of x5] 			{$x_6$};
  \node[state]  (x7) [below right=0.7 and 1cm of x5] 	{$x_7$};
  \node[state]  (x9) [below right=0.7 and 1cm of v]				{$x_9$};       

\node (3) [right=0.1cm of v] {$-1$};  
\node (4) [right=0.1cm of x4] {$+1$};  

\path[every node/.style={sloped,anchor=south,auto=false}]
%(v) edge[line width=1.5pt, bend left=25] 	node {} (x1)        
%(v) edge[-, line width=2pt] 	node {} (x1)            
(x1) edge[-,line width=3pt] 	node {} (x2) 
%(x2) edge[-,dashed,line width=0.5pt] 	node {} (v) 
(v) edge[-,line width=3pt] 	node {} (x4) 
(x4) edge[-,line width=3pt] 	node {} (x5) 
(x5) edge[-,dashed,line width=3pt] 	node {} (x6) 
(x6) edge[-,line width=3pt] 	node {} (x7) 
(x7) edge[-,dashed,line width=3pt] 	node {} (x5) 
(x5) edge[-,line width=3pt] 	node {} (x9) 
(x9) edge[-,dashed,line width=3pt] 	node {} (v);

\end{scope}
\end{tikzpicture}}
 \ \ \ \ \ \ \ \ \ \quad \ \ \ \ \ \ \ \ \ 
 \scalebox{0.8}{
\begin{tikzpicture}[
  ->,
  >=stealth',
  shorten >=1pt,
  auto,
  %node distance=2cm,
  semithick,
  every state/.style={circle,radius=0.1pt,text=black},
]
\begin{scope}
  \node[state]  (v)               					 		{$x_0/x_3$};
  \node[state]  (x1) [below left=0.7 and 1cm of v] 	 		{$x_1$};
  \node[state]  (x2) [above left=0.7 and 1cm of v]  			 		{$x_2$};
  \node[state]  (x4) [above right=0.7 and 1cm of v]					{$x_4$};
  \node[state]  (x5) [below right=0.7 and 1cm of x4] 			{$x_5/x_8$};
  \node[state]  (x6) [above right=0.7 and 1cm of x5] 			{$x_6$};
  \node[state]  (x7) [below right=0.7 and 1cm of x5] 	{$x_7$};
  \node[state]  (x9) [below right=0.7 and 1cm of v]				{$x_9$};       

\node (3) [right=0.1cm of v] {$-1$};  
\node (5) [left=0.1cm of x6] {$+1$};  

\path[every node/.style={sloped,anchor=south,auto=false}]
%(v) edge[line width=1.5pt, bend left=25] 	node {} (x1)        
%(v) edge[-, line width=2pt] 	node {} (x1)            
(x1) edge[-,line width=3pt] 	node {} (x2) 
%(x2) edge[-,dashed,line width=0.5pt] 	node {} (v) 
(v) edge[-,line width=3pt] 	node {} (x4) 
%(x4) edge[-,line width=3pt] 	node {} (x5) 
(x5) edge[-,line width=3pt] 	node {} (x6) 
(x6) edge[-,line width=3pt] 	node {} (x7) 
(x7) edge[-,dashed,line width=3pt] 	node {} (x5) 
(x5) edge[-,line width=3pt] 	node {} (x9) 
(x9) edge[-,dashed,line width=3pt] 	node {} (v);

\end{scope}
\end{tikzpicture}}
\caption{The edge $x_2x_3$ is removed and $x_3x_4$ is added (left). The edge $x_4x_5$ is removed and $x_5x_6$ is added (right).} 
\label{fig:step2_3}
\end{figure}

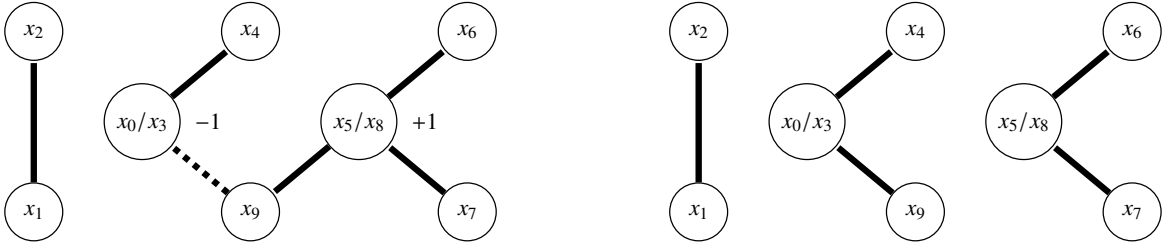
\begin{figure}[ht!]
\centering
\scalebox{0.8}{
\begin{tikzpicture}[
  ->,
  >=stealth',
  shorten >=1pt,
  auto,
  %node distance=2cm,
  semithick,
  every state/.style={circle,radius=0.1pt,text=black},
]
\begin{scope}
  \node[state]  (v)               					 		{$x_0/x_3$};
  \node[state]  (x1) [below left=0.7 and 1cm of v] 	 		{$x_1$};
  \node[state]  (x2) [above left=0.7 and 1cm of v]  			 		{$x_2$};
  \node[state]  (x4) [above right=0.7 and 1cm of v]					{$x_4$};
  \node[state]  (x5) [below right=0.7 and 1cm of x4] 			{$x_5/x_8$};
  \node[state]  (x6) [above right=0.7 and 1cm of x5] 			{$x_6$};
  \node[state]  (x7) [below right=0.7 and 1cm of x5] 	{$x_7$};
  \node[state]  (x9) [below right=0.7 and 1cm of v]				{$x_9$};       

\node (3) [right=0.1cm of v] {$-1$};  
\node (7) [right=0.1cm of x5] {$+1$};  

\path[every node/.style={sloped,anchor=south,auto=false}]
%(v) edge[line width=1.5pt, bend left=25] 	node {} (x1)        
%(v) edge[-, line width=2pt] 	node {} (x1)            
(x1) edge[-,line width=3pt] 	node {} (x2) 
%(x2) edge[-,dashed,line width=0.5pt] 	node {} (v) 
(v) edge[-,line width=3pt] 	node {} (x4) 
%(x4) edge[-,line width=3pt] 	node {} (x5) 
(x5) edge[-,line width=3pt] 	node {} (x6) 
%(x6) edge[-,line width=3pt] 	node {} (x7) 
(x7) edge[-,line width=3pt] 	node {} (x5) 
(x5) edge[-,line width=3pt] 	node {} (x9) 
(x9) edge[-,dashed,line width=3pt] 	node {} (v);

\end{scope}
\end{tikzpicture}}
 \ \ \ \ \ \ \ \ \ \quad \ \ \ \ \ \ \ \ \ 
 \centering
\scalebox{0.8}{
\begin{tikzpicture}[
  ->,
  >=stealth',
  shorten >=1pt,
  auto,
  %node distance=2cm,
  semithick,
  every state/.style={circle,radius=0.1pt,text=black},
]
\begin{scope}
  \node[state]  (v)               					 		{$x_0/x_3$};
  \node[state]  (x1) [below left=0.7 and 1cm of v] 	 		{$x_1$};
  \node[state]  (x2) [above left=0.7 and 1cm of v]  			 		{$x_2$};
  \node[state]  (x4) [above right=0.7 and 1cm of v]					{$x_4$};
  \node[state]  (x5) [below right=0.7 and 1cm of x4] 			{$x_5/x_8$};
  \node[state]  (x6) [above right=0.7 and 1cm of x5] 			{$x_6$};
  \node[state]  (x7) [below right=0.7 and 1cm of x5] 	{$x_7$};
  \node[state]  (x9) [below right=0.7 and 1cm of v]				{$x_9$};       

%\node (3) [below=0.1cm of v] {$-1$};  
%\node (9) [below=0.1cm of x9] {$-1$};  

\path[every node/.style={sloped,anchor=south,auto=false}]
%(v) edge[line width=1.5pt, bend left=25] 	node {} (x1)        
%(v) edge[-, line width=2pt] 	node {} (x1)            
(x1) edge[-,line width=3pt] 	node {} (x2) 
%(x2) edge[-,dashed,line width=0.5pt] 	node {} (v) 
(v) edge[-,line width=3pt] 	node {} (x4) 
%(x4) edge[-,line width=3pt] 	node {} (x5) 
(x5) edge[-,line width=3pt] 	node {} (x6) 
%(x6) edge[-,line width=3pt] 	node {} (x7) 
(x7) edge[-,line width=3pt] 	node {} (x5) 
%(x5) edge[-,line width=3pt] 	node {} (x9) 
(x9) edge[-,line width=3pt] 	node {} (v);

\end{scope}
\end{tikzpicture}}
\caption{The edge $x_6x_7$ is removed and $x_7x_8$ is added (left). The edge $x_8x_9$ is removed and $x_9x_0$ is added (right).}
\label{fig:step4_5}
\end{figure}

In general, however, it might happen that circuits contain both cut and internal edges, in which case we cannot use the circuit processing procedure explained above, as the processing of a circuit might result in a graphical realization for which the number of edges between the classes $V_1$ and $V_2$ lies outside the set $\{\gamma-1,\gamma,\gamma+1\}$. The latter condition is necessary for the intermediate states in the circuit processing procedure to be elements of $\mathcal{G}'(\gamma,d)$, by definition of that set. In order to overcome the issue described above, we will use the ideas in \cite{Bhatnagar2008}, and process a circuit at multiple places \emph{simultaneously} in case there is only one circuit in the canonical decomposition of a pairing, or, process multiple circuits \emph{simultaneously}  in case the decomposition yields multiple circuits.  At the core of this approach lies (a variation of) the \emph{mountain-climbing problem} \cite{Homma1952,Whittaker1966}. We begin with introducing this problem, and afterwards continue with the description of the circuit processing procedure, based on the solution to the mountain climbing problem. \medskip

\textit{Intermezzo: mountain climbing problem.} We first introduce some notation and terminology. For non-negative integers $a +1 < b$ we define an \emph{$\{a,b\}$-mountain} as a function $P : \{a,a+1,\dots,b\} \rightarrow \Z_{\geq 0}$ with the properties that (i) $P(a) = P(b) = 0$; (ii) $P(i) > 0$ for all $i \in \{a+1,\dots,b-1\}$; and (iii) $|P(i+1) - P(i)| = 1$ for all $i \in \{a,\dots,b-1\}$. A function $P : \{a,a+1,\dots,b\} \rightarrow \Z_{\leq 0}$ is called an \emph{$\{a,b\}$-valley} if the function $-P$ is an $\{a,b\}$-mountain. We subdivide a mountain into a left side $\{a,\dots,t\}$ and right side $\{t,\dots,b\}$ where $t$ is the smallest integer maximizing the function $P$. %We define $T = P(t)$ as the value of $P$ at the top $t$. 
For a valley function $P$, the left and right side are determined by the smallest integer $t$ minimizing the function $P$. 

\begin{definition}\label{def:transversal}
A \emph{traversal} of the mountain $P$ on  $\{a,\dots,b\}$ is a sequence $(0,t) = (i_1,j_1),\dots,(i_k,j_k) = (t,b)$ with the properties 
\begin{enumerate}[label=(\alph*)]
\item $|i_{r} - i_{r+1}| = |j_{r} - j_{r+1}| = 1$,
\item $P(i_r) + P(j_r) = P(t)$,
\item $a \leq i_r \leq t$ and $t \leq j_r \leq b$,
\end{enumerate}
for all $1 \leq r \leq k-1$. We always assume that a traversal is minimal, in the sense that there is no subsequence of $(0,t) = (i_1,j_1),\dots,(i_k,j_k) = (t,b)$ which is also a traversal.
\end{definition}

Roughly speaking, we place one person at the far left end of the mountain, and one at the first top. These persons now simultaneously traverse the mountain in such a way that the sum of their heights is always equal, and they always stay on their respective sides of the mountain that they started. The goal of the person on the left it to ascend to the top, whereas the goal of the player at the top is to descend to the far right of the mountain.

\begin{lemma}[\cite{Bhatnagar2008}]\label{lem:mountain} For any mountain or valley function $P$ on $\{a,\dots,b\}$ with first top $t$, there exists a traversal of $P$ of length at most $O((t-a)(b-t))$, that can be found in time $O((t-a)(b-t))$.
\end{lemma} \medskip

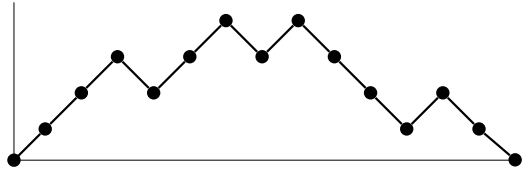
\begin{figure}[ht!]
\centering
\scalebox{0.7}{
\begin{tikzpicture}%[scale=0.1]
\coordinate (i0) (0,0) {};
%\node [below left=0.1cm and 0.1cm] {$\mathbf{(0,0)}$};

\node at (i0) [circle,scale=0.7,fill=black] {};
\node (i1) [above right=0.5cm and 0.5cm of i0, circle,scale=0.7,fill=black] {};
\node (i2) [above right=0.5cm and 0.5cm of i1, circle,scale=0.7,fill=black] {};
\node (i3) [above right=0.5cm and 0.5cm of i2, circle,scale=0.7,fill=black] {};
\node (i4) [below right=0.5cm and 0.5cm of i3, circle,scale=0.7,fill=black] {};
\node (i5) [above right=0.5cm and 0.5cm of i4, circle,scale=0.7,fill=black] {};
\node (i6) [above right=0.5cm and 0.5cm of i5, circle,scale=0.7,fill=black] {};
\node (i7) [below right=0.5cm and 0.5cm of i6, circle,scale=0.7,fill=black] {};
\node (i8) [above right=0.5cm and 0.5cm of i7, circle,scale=0.7,fill=black] {};
\node (i9) [below right=0.5cm and 0.5cm of i8, circle,scale=0.7,fill=black] {};
\node (i10) [below right=0.5cm and 0.5cm of i9, circle,scale=0.7,fill=black] {};
\node (i11) [below right=0.5cm and 0.5cm of i10, circle,scale=0.7,fill=black] {};
\node (i12) [above right=0.5cm and 0.5cm of i11, circle,scale=0.7,fill=black] {};
\node (i13) [below right=0.5cm and 0.5cm of i12, circle,scale=0.7,fill=black] {};
\node (i14) [below right=0.4cm and 0.5cm of i13, circle,scale=0.7,fill=black] {};

\path[every node/.style={sloped,anchor=south,auto=false}]
(i0) edge[-,very thick] node {} (i1)
(i1) edge[-,very thick] node {} (i2)
(i2) edge[-,very thick] node {} (i3)
(i3) edge[-,very thick] node {} (i4)
(i4) edge[-,very thick] node {} (i5)
(i5) edge[-,very thick] node {} (i6)
(i6) edge[-,very thick] node {} (i7)
(i7) edge[-,very thick] node {} (i8)
(i8) edge[-,very thick] node {} (i9)
(i9) edge[-,very thick] node {} (i10)
(i10) edge[-,very thick] node {} (i11)
(i11) edge[-,very thick] node {} (i12)
(i12) edge[-,very thick] node {} (i13)
(i13) edge[-,very thick] node {} (i14);

%Axis
\draw (0,0) -- (9.5,0);
\draw (0,0) -- (0,3);

\end{tikzpicture}}
\caption{Example of a mountain function $P$ on the integers in  $\{0,\dots,14\}$ with the first top at $t = 6$. The left side of the mountain is given by $\{0,\dots,6\}$ and the right side by $\{6,\dots,14\}$. A traversal of $P$ is given by the sequence $(0,6),(1,7),(0,8),(1,9),(2,10),(3,11),(4,12),(5,13),(6,14)$.}
%\label{fig:switch}
\end{figure} 

We finish this part with some additional notation that will be used later on. Let $P_j : \{a_j,\dots,b_j\} \rightarrow \Z$ for $j = 1,\dots,p$ be a collection of mountain and valley functions such that $a_1 = 0$, $a_j = b_{j+1}$ for $j = 1,\dots,p-1$, and every $P_j$ is either a mountain or a valley.  We define the \emph{landscape} $Q$ of the functions $P_1,\dots,P_l$ as the function $Q : \{0,1,\dots,b_l\} \rightarrow \Z$ given by $Q(i) = P_j(i)$ where $j = j(i)$ is such that $i \in \{a_j,\dots,b_j\}$. Note that  $Q(0) = Q(b_l) = 0$, and $|Q(i+1) - Q(i)| = 1$ for all $i \in \{0,\dots,b_j-1\}$. Moreover, for any function $R : \{0,\dots,r\} \rightarrow \Z$ satisfying the latter two conditions, there is a unique collection of mountain and valley functions so that $R$ is the landscape of those functions. We call functions satisfying these conditions \emph{landscape functions}.

\paragraph{General Case.}
We first partition every circuit into a collection of so-called sections, which in turn will be grouped into so-called segments. Let $C_1,\dots,C_s$ be the canonical circuit decomposition of the symmetric difference $G \triangle G'$ for some pairing $\psi$, and assume w.l.o.g. that $C_i \prec_{C} C_j$ whenever $i < j$. We write $C_i = x_0^ix_1^i\dots x_{q_i}^ix_0^i$ where $x_0^ix_1^i$ is the lexicographically smallest blue edge adjacent to the starting point $x_0^i$ of the circuit $C_i$ that contains $q_i + 1$ edges (and where $x_0 = x_{q_i+1}$). For any $i$, we define the function
\[
l_i\left(r\right) = \left\{ \begin{array}{ll}
-1 & \text{ if $\{x_{r-2}^i,x_{r-1}^i\}$ is cut edge and $\{x_{r-1}^i,x_{r}^i\}$ is internal edge}, \\
1 & \text{ if $\{x_{r-2}^i,x_{r-1}^i\}$ is internal edge and $\{x_{r-1}^ix_{r}^i\}$ is cut edge}, \\
0 & \text{ otherwise, }
\end{array}\right.
\]
for $r = 2,4,\dots,q_i + 1$. The function $l_i$ indicate what happens to number of cut edges of a graphical realization when we perform a hinge flip on a pair of consecutive edges $\{x_{r-2}^i,x_{r-1}^i\}$ and $\{x_{r-1}^i,x_{r}^i\}$ on the circuit $C_i$. \medskip

\textit{Decomposition into segments.} We subdivide every circuit $C_i$ into a sequence of (not necessarily closed) walks of even length, called \emph{sections}. Let $Z_i = \{r : l_i(r) \neq 0\} =  \{z_1,\dots,z_{u_i}\} \subseteq \{2,4,\dots,q_i + 1\}$ be the set of indices that represent a change in cut edges along the circuit, where we assume that $z_1 \leq z_2 \leq \dots \leq z_{u_i}$. We define $C_i^1 = x_0^ix_1^i\dots x_{z_1}^i$ and $C_i^j = x_{z_j}^i\dots x_{z_{j+1}}^i$ for $j = 2,\dots,u_i-1$. %Intuitively, the index $z_1$ is the first point at which the number of cut edges will change if we would apply hinge flip operations up to that point (as in the warm-up example) to remove the blue edges of $G \triangle G'$ and add the red edges of $G \triangle G'$, on that section.
%We continue this procedure, i.e., we define $C_i^2 = x_{z_1}^i\dots x_{z_2 + 2}^i$
%where $z_2$ is the second even index on the walk $C_i$ so that $l_i(z_2) \neq 0$. In general, $C_i^{w} = x_{z_{w-1}}^i\dots x_{z_w + 2}^i$ where $z_{w - 1}$ and $z_{w}$ are the $(w-1)$-th and $w$-th index on the walk $C_i$ with $l_i$ unequal to zero. 
If $l_i(q_i + 1) \neq 0$ this procedure partitions the circuit $C_i$ completely, with $C_i^{u_i}$ being the last section. Otherwise, we define $C_i^{u_1+1} = x_{z_{u_i}}^i \dots x_0^i$ as the final section, which is the remainder of the circuit $C_i$. We define $U_i$ as the total number of obtained sections, which is either $u_i$ or $u_i + 1$. %Note that $l_i(r) = 0$ for all even indices $r$ on the section in this case. 
Note that when $Z_i = \emptyset$, the whole circuit will form one section $C_i = C_i^1$. Also note that a section always starts with a blue edge. We extend the function $l_i$ to sections in the following way:
\[
l_i\left(C_i^j \right) = \sum_{r=2,4,\dots,z_j} l_i(r) = \left\{ \begin{array}{ll}
-1 & \text{ if } l_i(z_j) = -1, \\
1 & \text{ if } l_i(z_j) = 1, \\
0 & \text{ otherwise, }
\end{array}\right.
\]
for $j = 1,\dots,U_i$. Note that $l(C_i^j) \in \{-1,1\} \text{ for } j = 1,\dots,u_i$, and zero for $j = u_i+1$ if this term is present. An example is given in Figure \ref{fig:sections}. 

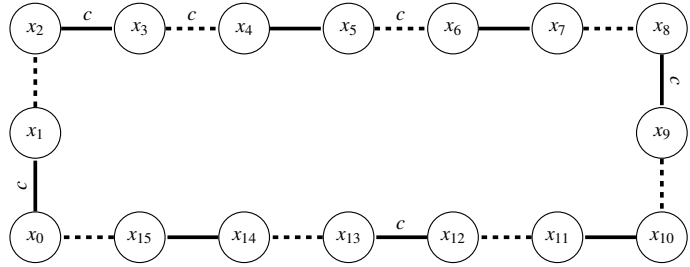
\begin{figure}[ht!]
\centering
\scalebox{0.7}{
\begin{tikzpicture}[
  ->,
  >=stealth',
  shorten >=1pt,
  auto,
  %node distance=2cm,
  semithick,
  every state/.style={circle,radius=0.1pt,text=black},
]
\begin{scope}
  \node[state]  (x0)               					 		{$x_0$};
  \node[state]  (x1) [above=1cm of x0] 	 		{$x_1$};
  \node[state]  (x2) [above=1cm of x1]  			 		{$x_2$};
  \node[state]  (x3) [right=1cm of x2]					{$x_3$};
  \node[state]  (x4) [right=1cm of x3]					{$x_4$};
  \node[state]  (x5) [right=1cm of x4]					{$x_5$};
  \node[state]  (x6) [right=1cm of x5]					{$x_6$};
  \node[state]  (x7) [right=1cm of x6]					{$x_7$};
  \node[state]  (x8) [right=1cm of x7]					{$x_8$};
  \node[state]  (x9) [below=1cm of x8]					{$x_9$};
  \node[state]  (x10) [below=1cm of x9]					{$x_{10}$};
  \node[state]  (x11) [left=1cm of x10]					{$x_{11}$};
  \node[state]  (x12) [left=1cm of x11]					{$x_{12}$};
  \node[state]  (x13) [left=1cm of x12]					{$x_{13}$};
  \node[state]  (x14) [left=1cm of x13]					{$x_{14}$};
  \node[state]  (x15) [left=1cm of x14]					{$x_{15}$};
 
\path[every node/.style={sloped,anchor=south,auto=false}]
%(v) edge[line width=1.5pt, bend left=25] 	node {} (x1)        

(x0) edge[-, line width=2pt] 	node {$c$} (x1)            
(x1) edge[-,dashed,line width=2pt] 	node {} (x2) 
(x2) edge[-,line width=2pt] 	node {$c$} (x3) 
(x3) edge[-,dashed,line width=2pt] 	node {$c$} (x4) 
(x4) edge[-, line width=2pt] 	node {} (x5)
(x5) edge[-,dashed,line width=2pt] 	node {$c$} (x6) 
(x6) edge[-, line width=2pt] 	node {} (x7)
(x7) edge[-,dashed,line width=2pt] 	node {} (x8) 
(x8) edge[-, line width=2pt] 	node {$c$} (x9)
(x9) edge[-,dashed,line width=2pt] 	node {} (x10) 
(x10) edge[-, line width=2pt] 	node {} (x11)
(x11) edge[-,dashed,line width=2pt] 	node {} (x12) 
(x12) edge[-, line width=2pt] 	node {$c$} (x13)
(x13) edge[-,dashed,line width=2pt] 	node {} (x14) 
(x14) edge[-, line width=2pt] 	node {} (x15)
(x15) edge[-,dashed,line width=2pt] 	node {} (x0); 
     
\end{scope}
\end{tikzpicture}}
\caption{The circuit $C_1 = x_0x_1\dots x_{15}x_0$ with $q_1 = 15$. The blue edges are represented by the solid edges, and the red edges by the dashed edges. A label $c$ on an edge indicates that it is a cut edge (all others are internal edges). We have $C_1^1 = x_0x_1x_2$ with $l_1(C_1^1) = -1$; $C_1^2 =x_2x_3x_4x_5x_6$ with $l_1(C_1^2) = 1$; $C_1^3 = x_6x_7x_8x_9x_{10}$ with $l_1(C_1^3) = -1$; $C_1^4 = x_{10}x_{11}x_{12}x_{13}x_{14}$ with $l_1(C_1^4) = -1$; and $C_1^5 = x_{14}x_{15}x_0$ with $l_1(C_1^5) = 0$ (note that $U_1 = 5$ in this example).} 
\label{fig:sections}
\end{figure}

We now continue by grouping the union of all sections into segments in a similar flavor. For sake of readability, we rename the sections $C_1^1,\dots,C_1^{U_1},C_2^{1},\dots,C_2^{U_2},\dots,C_s^1,\dots,C_s^{U_s}$ as $D_1,\dots,D_U$ in the obvious way, where $U = \sum_{i=1}^s U_i$, and we define $l(D_k) = l_i(C_i^j)$ if $C_i^j$ was renamed $D_k$.  
We define $W = \{k : l(D_k) \neq 0\}
 = \{w_1,\dots,w_{B}\}$ as the set of sections representing a change in cut edges along a circuit, where we assume that $w_1 \leq \dots \leq w_{B}$.
We define the \emph{segment} $S_1 = (D_1,\dots,D_{w_1})$, and $S_i = (D_{w_{i-1} + 1},\dots,D_{w_{i}})$ for $i = 2,\dots,w_{B} - 1$.  If $l(D_U) \neq 0$, i.e., when $w_B = U$, this procedure completely groups the collection of sections into segments. Otherwise, we redefine the last segment as $S_B = (D_{w_{B-1} + 1},\dots,D_U)$.  We can extend the function $l$ to segments in the following way:
\[
l\left(S_i \right) = \sum_{j=w_{i-1}+1}^{w_i} l(D_j) = \left\{ \begin{array}{ll}
-1 & \text{ if } l(D_{w_i}) = -1, \\
1 & \text{ if } l(D_{w_i}) = 1, \\
\end{array}\right.
\]
for $i = 1,\dots,B-1$, and $l\left(S_B \right) = \sum_{j=w_{B-1}+1}^{U} l(D_j)$. Note that 
\begin{equation}\label{eq:segment_nonzero}
l(S_i) \in \{-1,1\} \ \ \  \text{ for } i = 1,\dots,B.\footnote{Unless in the special case that there is only one segment $S_1$ covering all circuits, then $l(S_1) = 0$. This happens, e.g., in the situation of the warm-up example.}
\end{equation}

An example of a decomposition into segments is given in Figures \ref{fig:sym_dif} and \ref{fig:landscape_encoding} later on.
Roughly speaking, a segment is a maximal collection of edges that could be processed, using hinge flips operations as in the warm-up example, until the number of cut-edges changes. In particular, the first segment represents precisely the point up to where we could carry out the same processing steps as in the warm-up example until the number of cut edges will have changed for the first time. 
Note that a segment might contain sections from multiple circuits, in particular,  it might consist of a final section of a circuit $J_1$, then some full circuits $J_2,\dots,J_h$ (which all form a section on their own) and then the first section of some circuit $J_{h+1}$. The function $l$ is then zero on the last section of $J_1$ and all circuits (sections) $J_2,\dots,J_h$, and non-zero on the section of $J_{h+1}$. \medskip

 \textit{Unwinding/rewinding of a segment.} The \emph{unwinding} of a section $D = x_f\dots x_g$ consists of performing a number of hinge flip operations, that represent transitions in the Markov chain $\mathcal{M}'(\gamma,d)$. That is we perform a sequence of hinge flip operations replacing the (blue) edges $\{x_{r-2},x_{r-1}\}$ by (red) edges $\{x_{r-1},x_r\}$ for $r = f+2,\dots,g$, in increasing order of $r$. Sometimes, we need to temporarily undo the unwinding of a section, in which case we perform a sequence of hinge flip operations replacing the (red) edges $\{x_{r-1},x_{r}\}$ by (blue) edges $\{x_{r-2},x_{r-1}\}$ for $r = f+2,\dots,g$, in decreasing order of $r$ this time. That is, we reverse the operations done during the unwinding. This is called \emph{rewinding} a section. We say that a circuit is (currently) processed if all its sections have been unwound, and it is (currently) unprocessed if at least one section has not been unwound.
 
The unwinding of a segment $S_i = (D_{a_i},\dots,D_{a_i+1})$ consists of unwinding the sections $D_{a_i}, \dots, D_{a_i+1}$ in increasing order. The rewinding of $S_i$ consists of rewinding the section $D_{a_i}, \dots, D_{a_i+1}$ in decreasing order.

 \medskip
 
\begin{figure}[ht!]
\centering
%\scalebox{0.8}{
%\begin{tikzpicture}[
%  ->,
%  >=stealth',
%  shorten >=1pt,
%  auto,
%  %node distance=2cm,
%  semithick,
%  every state/.style={circle,radius=0.1pt,text=black},
%]
%\begin{scope}
%  \node[state]  (x0)               					 		{$x_0$};
%  \node[state]  (x1) [right=1cm of x0]					{$x_1$};
%  \node[state]  (x2) [right=1cm of x1]					{$x_2$};
%  \node[state]  (x3) [right=1cm of x2]					{$x_3$};
%  \node[state]  (x4) [right=1cm of x3]					{$x_4$};
%  \node[state]  (x5) [right=1cm of x4]					{$x_5$};
%  \node[state]  (x6) [right=1cm of x5]					{$x_6$};
% 
%\path[every node/.style={sloped,anchor=south,auto=false}]
%%(v) edge[line width=1.5pt, bend left=25] 	node {} (x1)        
%
%
%(x0) edge[-, line width=2pt] 	node {} (x1)            
%(x1) edge[-,dashed,line width=2pt] 	node {} (x2) 
%(x2) edge[-,line width=2pt] 	node {} (x3) 
%(x3) edge[-,dashed,line width=2pt] 	node {} (x4) 
%(x4) edge[-, line width=2pt] 	node {} (x5)
%(x5) edge[-,dashed,line width=2pt] 	node {} (x6);
%     
%\end{scope}
%\end{tikzpicture}}
%\quad \bigskip
\scalebox{0.7}{
\begin{tikzpicture}[
  ->,
  >=stealth',
  shorten >=1pt,
  auto,
  %node distance=2cm,
  semithick,
  every state/.style={circle,radius=0.1pt,text=black},
]
\begin{scope}
  \node[state]  (x0)               					 		{$x_0$};
  \node[state]  (x1) [right=1cm of x0]					{$x_1$};
  \node[state]  (x2) [right=1cm of x1]					{$x_2$};
  \node[state]  (x3) [right=1cm of x2]					{$x_3$};
  \node[state]  (x4) [right=1cm of x3]					{$x_4$};
  \node[state]  (x5) [right=1cm of x4]					{$x_5$};
  \node[state]  (x6) [right=1cm of x5]					{$x_6$};
  
  \node[state]  (y0) [below=2cm of x0]					{$x_0$};
  \node[state]  (y1) [right=1cm of y0]					{$x_1$};
  \node[state]  (y2) [right=1cm of y1]					{$x_2$};
  \node[state]  (y3) [right=1cm of y2]					{$x_3$};
  \node[state]  (y4) [right=1cm of y3]					{$x_4$};
  \node[state]  (y5) [right=1cm of y4]					{$x_5$};
  \node[state]  (y6) [right=1cm of y5]					{$x_6$};
  
  \node (P1) [below left=1cm and 0.5cm of x3] {$\text{unwinding}$};
  \node (P2) [below right=1cm and 0.5cm of x3] {$\text{rewinding}$};
  
  \node (A1) [below left=0.25cm and 0.25cm of x3] {};
  \node (B1) [below=1.5cm of A1] {};
  \node (A2) [below right=0.25cm and 0.25cm of x3] {};
  \node (B2) [below=1.5cm of A2] {};
  
\path[every node/.style={sloped,anchor=south,auto=false}]
%(v) edge[line width=1.5pt, bend left=25] 	node {} (x1)        

(x0) edge[-, line width=2pt] 	node {} (x1)            
%(x1) edge[-,dashed,line width=2pt] 	node {} (x2) 
(x2) edge[-,line width=2pt] 	node {} (x3) 
%(x3) edge[-,dashed,line width=2pt] 	node {} (x4) 
(x4) edge[-, line width=2pt] 	node {} (x5)
%(x5) edge[-,dashed,line width=2pt] 	node {} (x6);

(y1) edge[-, line width=2pt] 	node {} (y2)            
%(x1) edge[-,dashed,line width=2pt] 	node {} (x2) 
(y3) edge[-,line width=2pt] 	node {} (y4) 
%(x3) edge[-,dashed,line width=2pt] 	node {} (x4) 
(y5) edge[-, line width=2pt] 	node {} (y6)
(A1) edge[->,line width=2pt] node {} (B1)
(B2) edge[->,line width=2pt] node {} (A2);
%(x5) edge[-,dashed,line width=2pt] 	node {} (x6);
     
\end{scope}
\end{tikzpicture}}
\caption{A section $D = x_0x_1\dots x_{6}$. The blue edges are represented by the solid edges. The unwinding consists of performing first a hinge flip with $\{x_0,x_1\}$ to $\{x_1,x_2\}$; then $\{x_2,x_3\}$ to $\{x_3,x_4\}$; and finally $\{x_4,x_5\}$ to $\{x_5,x_6\}$. The rewinding consist of first a hinge flip with $\{x_5,x_6\}$ to $\{x_4,x_5\}$; then $\{x_3,x_4\}$ to $\{x_2,x_3\}$; and finally $\{x_1,x_2\}$ to $\{x_0,x_1\}$  } 
\label{fig:unwinding_rewinding}
\end{figure}
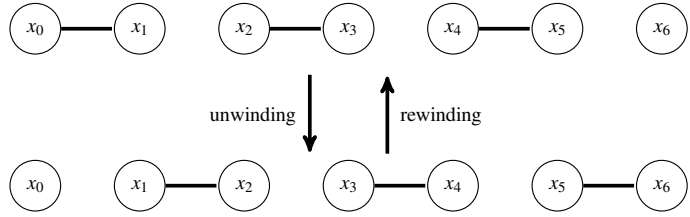 
 \medskip

\textit{Landscape processing.} Remember that $B$ is the number of segments obtained from the decomposition of circuits into segments. We define the function $P : \{0,1,\dots,B\} \rightarrow \Z$ by $P(0) = 0$ and $P(i) = \sum_{j=1}^i l(S_j)$ for $i = 1,\dots,B$.

\begin{lemma}\label{lem:landscape}
The function $P$ is a landscape function.
\end{lemma}
\begin{proof}
We have to check that $P(0) = P(B) = 0$ and that $|P(i+1) - P(i)| = 1$ for all $i = 0,\dots,B-1$, see the description of the mountain climbing problem. We have $P(0)$ by definition. Moreover, since both graphical realizations $G$ and $G'$ contain $\gamma$ cut edges, it holds that
$
P(B) = \sum_{i=1}^B l(S_i) = 0.
$
Finally, using (\ref{eq:segment_nonzero}) and the definition of $P$, it follows that
\[
|P(i+1) - P(i)| = \Bigg|\sum_{j=1}^{i+1}
 l(S_j) - \sum_{j=1}^i l(S_j)\Bigg| = |l(S_i)| = 1
\]
for all $i = 1,\dots,B-1$.
\end{proof}

Based on the segments $S_1,\dots,S_B$, we define the canonical path from $G$ to $G'$ in the state space graph of the chain $\mathcal{G}'(\gamma,d)$ that replaces all the blue edges in $G \triangle G'$ with the red edges in $G \triangle G'$. By Lemma \ref{lem:landscape} we know $P$ is a landscape function and therefore there is a unique decomposition into mountain and valley functions $P_1,\dots,P_p$ so that $P$ is the landscape function for this collection, where every function is of the form $P_j : \{a_j,\dots,b_j\} \rightarrow \Z$ with $a_1 = 0, b_{j} = a_{j+1}$ for $j = 1,\dots,p-1$, and $b_p = B$.\footnote{The function $P_1$ can be found by determining the first $j > 0$ so that $P(j) = 0$. The sign of $P(1)$ determines if it is a mountain or a valley. The remaining mountains and valleys can be found similarly.} 
The \emph{processing of a mountain/valley} $P_j$ means that all segments $S_{a_j + 1},\dots,S_{b_j}$ will be unwound (it might be that during this procedure segments are temporarily rewound). This processing will rely on a traversal of the mountain, see Definition \ref{def:transversal}. We say that the segments $S_{a_j + 1},\dots,S_{c_j}$ are on the left side of the mountain, and the segments $S_{c_j + 1},\dots,S_{b_j}$ on the right side. Let $P = P_j$ for some $j$ and assume that $P$ is a mountain function. For sake of notation, we write $a = a_j$ and $b = b_j$, and $t = t_j$ where $t_j$ is the first top of the mountain. 

Now, fix some traversal $(a,t) = (r_1,s_1),\dots,(r_k,s_k) = (t,b)$ of $P$. For $c = 1,\dots,k-1$ in increasing order, do the following:

\begin{itemize}
\item if $r_{c+1} > r_c$ and $s_{c+1} > s_c$: first unwind segment $S_{r_{c+1}}$, then unwind segment $S_{s_{c+1}}$;
\item if $r_{c+1} > r_c$ and $s_{c+1} < s_c$: first unwind segment $S_{r_{c+1}}$, then rewind segment $S_{s_{c}}$;
\item if $r_{c+1} < r_c$ and $s_{c+1} > s_c$: first  rewind segment $S_{r_{c}}$, then unwind segment $S_{s_{c+1}}$;
\item if $r_{c+1} < r_c$ and $s_{c+1} < s_c$: first rewind segment $S_{r_{c}}$, then rewind segment $S_{s_{c}}$.
\end{itemize}
This describes the processing of a mountain based on a traversal. Note that after the processing of a mountain, indeed all its segments have been unwound. If $P$ is a valley function, we can use essentially the same procedure performed on $-P$. The \emph{processing of a landscape} is done by processing the mountains/valleys $P_1,\dots, P_p$ in increasing order. 

This procedure generates a sequence $G = Z_1,Z_2,\dots,Z_l = G'$ of graphical realizations transforming $G$ into $G'$ where any two consecutive realizations differ by a hinge flip operation. The following lemma shows that this sequence indeed defines a (canonical) path from $G$ to $G'$ in the state space graph of $\mathcal{M}(\gamma,d)$, for a given pairing $\psi$. This lemma is essentially the motivation for the definition of $\mathcal{G}'(\gamma,d)$.

\begin{lemma}\label{lem:feasible_path}
Let $Z = Z_i$ be a graphical realization on the constructed path from $G$ to $G'$ for pairing $\psi$, with degree sequence $d'$ and $\gamma'$ cut edges. Then properties (i), (ii) and (iii) defining $\mathcal{G}'(\gamma,d)$ are satisfied. 

Moreover, there exists a polynomial $r(\cdot)$ such that the length of any constructed (canonical) path carrying flow is at most $r(n)$.
\end{lemma}
\begin{proof}[Proof (sketch).] Since hinge flip operations never add or remove edges, property (i) is clearly satisfied. Since the operations $(1)-(4)$ given above unwind and rewind at most two segments, and by construction of the trajectories describing the traversal, the property (ii) is also satisfied. Finally, the cases $(1)-(4)$, in combination with the second property of a traversal as in Definition \ref{def:transversal}, guarantee that property (iii) is satisfied. To see that all canonical paths have polynomial length, note that the traversal has polynomial length, and also every individual segment has polynomial length.
\end{proof}

%\newpage

\subsubsection{Encoding}
We continue with defining the notion of an \emph{encoding} that will be used in the next section to bound the congestion of an edge in the state space graph of $\mathcal{M}(\gamma,d)$.  
Let $\tau = (Z,Z')$ be a given transition of the Markov chain. Suppose that a canonical path from $G$ to $G'$ for some pairing $\psi \in \Psi(G,G')$, with canonical circuit decomposition $\{C_1,\dots,C_s\}$, uses the transition $\tau$. %Let $P$ be the landscape function of this canonical path, and $P_1,\dots,P_p$ its decomposition into mountain and valley functions. 
We define $L_\tau(G,G) = (G \triangle G') \triangle Z$.  An example is given in Figures \ref{fig:sym_dif}, \ref{fig:landscape_encoding}, \ref{fig:transition} and \ref{fig:encoding}.
%Let $P_j$ with top $t_j$ be the mountain/valley containing the transition $\tau$. 

\begin{lemma}\label{lem:recovery_jdm}
Given $\tau = (Z,Z')$, $\psi$, $L$, %$T_{\psi}$, and $\sigma_{\psi}$ be given (where $T$ and $\sigma$ are circuits from the set of circuits obtained from the canonical decomposition of the pairing into circuits).
if there is some pair $(G,G')$ so that $L = L_{\tau}(G,G')$, %$\sigma_{\psi} = \sigma_{\psi}(G,G')$, and $T_{\psi} = T_{\psi}(G,G')$, 
then there are at most $\frac{1}{8}n^4$ such pairs.
\end{lemma}
\begin{proof}%[\pkblue{Proof}] 

For any pair $(G,G')$, let $P$ be the landscape function of this canonical path between $G$ and $G'$ using the transition $\tau$, and $P_1,\dots,P_p$ its decomposition into mountain and valley functions.  Let $T_{\tau,\psi}(G,G') \in \{C_1,\dots,C_s\}$ be the circuit containing the first node of the first segment of the right part of the mountain/valley $P_j$ containing the transition $\tau$. Moreover, if $\tau$ is used in the processing of a segment on the left side of the mountain $P_j$ containing $\tau$, let $\sigma_{\psi}(G,G')$ be the circuit containing the last node of the segment with highest index on the right side of the mountain that is currently unwound. If $\tau$ lies on the right side of the mountain, we let $\sigma_{\psi}(G,G')$ be the circuit containing the last node of the segment with highest index on the left side of the mountain that is currently unwound.

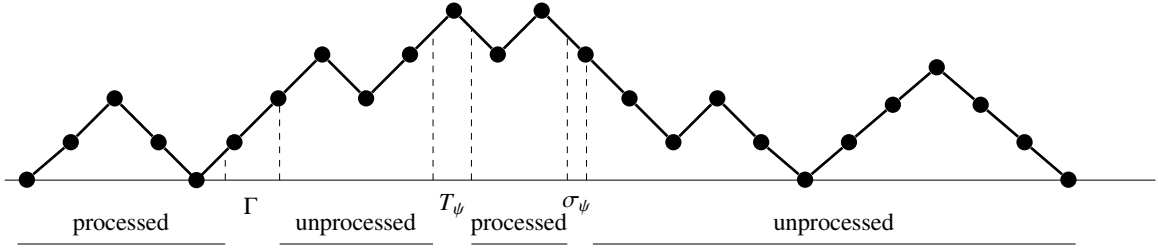
\begin{figure}[ht!]
\centering
\scalebox{0.85}{
\begin{tikzpicture}%[scale=0.1]
\coordinate (i0) (0,0) {};
%\node [below left=0.1cm and 0.1cm] {$\mathbf{(0,0)}$};

\node at (i0) [circle,scale=0.7,fill=black] {};
\node (i1) [above right=0.5cm and 0.5cm of i0, circle,scale=0.7,fill=black] {};
\node (i2) [above right=0.5cm and 0.5cm of i1, circle,scale=0.7,fill=black] {};
\node (i3) [above right=0.5cm and 0.5cm of i2, circle,scale=0.7,fill=black] {};
\node (i4) [below right=0.5cm and 0.5cm of i3, circle,scale=0.7,fill=black] {};
\node (i5) [above right=0.5cm and 0.5cm of i4, circle,scale=0.7,fill=black] {};
\node (i6) [above right=0.5cm and 0.5cm of i5, circle,scale=0.7,fill=black] {};
\node (i7) [below right=0.5cm and 0.5cm of i6, circle,scale=0.7,fill=black] {};
\node (i8) [above right=0.5cm and 0.5cm of i7, circle,scale=0.7,fill=black] {};
\node (i9) [below right=0.5cm and 0.5cm of i8, circle,scale=0.7,fill=black] {};
\node (i10) [below right=0.5cm and 0.5cm of i9, circle,scale=0.7,fill=black] {};
\node (i11) [below right=0.5cm and 0.5cm of i10, circle,scale=0.7,fill=black] {};
\node (i12) [above right=0.5cm and 0.5cm of i11, circle,scale=0.7,fill=black] {};
\node (i13) [below right=0.5cm and 0.5cm of i12, circle,scale=0.7,fill=black] {};
\node (i14) [below right=0.4cm and 0.5cm of i13, circle,scale=0.7,fill=black] {};
\node (i15) [above right=0.4cm and 0.5cm of i14, circle,scale=0.7,fill=black] {};
\node (i16) [above right=0.4cm and 0.5cm of i15, circle,scale=0.7,fill=black] {};
\node (i17) [above right=0.4cm and 0.5cm of i16, circle,scale=0.7,fill=black] {};
\node (i18) [below right=0.4cm and 0.5cm of i17, circle,scale=0.7,fill=black] {};
\node (i19) [below right=0.4cm and 0.5cm of i18, circle,scale=0.7,fill=black] {};
\node (i20) [below right=0.4cm and 0.5cm of i19, circle,scale=0.7,fill=black] {};

\node (j1) [above left=0.5cm and 0.5cm of i0, circle,scale=0.7,fill=black] {};
\node (j2) [above left=0.5cm and 0.5cm of j1, circle,scale=0.7,fill=black] {};
\node (j3) [below left=0.5cm and 0.5cm of j2, circle,scale=0.7,fill=black] {};
\node (j4) [below left=0.4cm and 0.5cm of j3, circle,scale=0.7,fill=black] {};

\path[every node/.style={sloped,anchor=south,auto=false}]
(i0) edge[-,very thick] node {} (i1)
(i1) edge[-,very thick] node {} (i2)
(i2) edge[-,very thick] node {} (i3)
(i3) edge[-,very thick] node {} (i4)
(i4) edge[-,very thick] node {} (i5)
(i5) edge[-,very thick] node {} (i6)
(i6) edge[-,very thick] node {} (i7)
(i7) edge[-,very thick] node {} (i8)
(i8) edge[-,very thick] node {} (i9)
(i9) edge[-,very thick] node {} (i10)
(i10) edge[-,very thick] node {} (i11)
(i11) edge[-,very thick] node {} (i12)
(i12) edge[-,very thick] node {} (i13)
(i13) edge[-,very thick] node {} (i14)
(i14) edge[-,very thick] node {} (i15)
(i15) edge[-,very thick] node {} (i16)
(i16) edge[-,very thick] node {} (i17)
(i17) edge[-,very thick] node {} (i18)
(i18) edge[-,very thick] node {} (i19)
(i19) edge[-,very thick] node {} (i20)
(i0) edge[-,very thick] node {} (j1)
(j1) edge[-,very thick] node {} (j2)
(j2) edge[-,very thick] node {} (j3)
(j3) edge[-,very thick] node {} (j4);

%Axis
\draw (-3,0) -- (15,0);

%T
\draw (3.7,0) edge[-,dashed] (3.7,2.35);
\draw (4.3,0) edge[-,dashed] (4.3,2.4);
\node at (4,-0.4) {$T_{\psi}$};
%Gamma
\draw (0.45,0) edge[-,dashed] (0.45,0.45);
\draw (1.3,0) edge[-,dashed] (1.3,1.45);
\node at (0.85,-0.4) {$\Gamma$};
%Sigma
\draw (5.8,0) edge[-,dashed] (5.8,2.3);
\draw (6.1,0) edge[-,dashed] (6.1,2);
\node at (5.95,-0.4) {$\sigma_{\psi}$};
%\draw (0,0) -- (0,3);

\draw (-2.8,-1) edge[-] node[anchor=south] {$\text{processed}$} (0.45,-1);
\draw (1.3,-1) edge[-] node[anchor=south] {$\text{unprocessed}$} (3.7,-1);
\draw (4.3,-1) edge[-] node[anchor=south] {$\text{processed}$} (5.8,-1);
\draw (6.2,-1) edge[-] node[anchor=south] {$\text{unprocessed}$} (13.75,-1);
\end{tikzpicture}}
\caption{The dashed vertical lines sketch the ranges of the circuits $T_{\psi}$, $\sigma_{\psi}$ and $\Gamma$. For every other circuit, contained in one of the four regions represented below the landscape, we know whether it has currently been processed or not.}
%\label{fig:switch}
\end{figure} 

We claim that, given $T_{\psi}, \sigma_{\psi} \in \{C_1,\dots,C_s\}$, it can be argued that there are at most $8$ pairs $(G,G')$ so that $T_{\psi} = T_{\psi}(G,G')$, $\sigma_{\psi} = \sigma_{\psi}(G,G')$. %, and  $\gamma$ is the circuit containing $\tau$ for the canonical path used by $(G,G')$.
This can be seen as follows. Note that we can infer for all other circuits in $\{C_1,\dots,C_s\} \setminus \{T_{\psi}$, $\sigma_{\psi},\Gamma\}$ which edges belong to $G$ and which to $G'$ using the (global) circuit ordering. To see this, assume that $\Gamma \preceq_{C}  T_{\psi} \preceq_{C} \sigma_{\psi}$ (the only other case $\sigma_{\psi} \preceq_{C}  T_{\psi} \preceq_{C} \Gamma$ is similar). Because the landscapes of the canonical paths always respect the circuit ordering, we know that all circuits in the canonical decomposition of $\psi$ appearing before $\Gamma$ have been unwound at this point.  All circuits lying strictly between $\Gamma$ and $T_{\psi}$ are not unwound. The circuits strictly between $T_{\psi}$ and $\sigma_{\psi}$ again have been unwound, and finally, all circuits appearing after $\sigma_{\psi}$ have not been unwound. By comparison with $Z$, it is uniquely determined with edges on these circuits belong to $G$ and which to $G'$. For the remaining three circuits $T_{\psi}$, $\sigma_{\psi}$ and $\gamma$ there are for every circuit two possible configurations of the edges of $G$ and $G'$, since every circuit alternates between edges of $G$ and $G'$.\footnote{Note that we cannot use the transition $\tau$ to infer which edges belong to $G$ and $G'$ on the circuit $\Gamma$, as we do not know (i.e., we do not encode) whether we are unwinding or rewinding the segment containing $\tau$.}
 Hence, there are at most $2^3 = 8$ possible pairs $(G,G')$ with the desired properties given $T_{\psi}$ and $\sigma_{\psi}$.

Finally, note that for any pairing $\psi$, there are at most $\frac{1}{4} \binom{n}{2}$ circuits in the canonical circuit decomposition $\{C_1,\dots,C_s\}$ of the pairing $\psi$, as every circuit has length at least four. Hence, for both $T_{\psi}$ and $\sigma_{\psi}$ there are at most $\frac{1}{4}\binom{n}{2}$ possible choices. Since $\Gamma$ is uniquely determined by the transition $\tau$, this implies that there are at most 
\[
8 \cdot \frac{1}{4} \binom{n}{2}\cdot \frac{1}{4} \binom{n}{2} \leq \frac{n^4}{8}
\] 
possible pairs $(G,G')$ with $L = L_{\tau}(G,G')$.\footnote{A canonical path uses every transition at most once, which follows from the fact that we assumed that a traversal is always minimal, see Definition \ref{def:transversal}.}
\end{proof}

\begin{figure}[ht!]
		\centering
		\scalebox{0.7}{
			\begin{tikzpicture}[
			->,
			>=stealth',
			shorten >=1pt,
			auto,
			%node distance=2cm,
			semithick,
			every state/.style={circle,radius=0.1pt,text=black},
			]
			\begin{scope}
			\node[state]  (a1)               					 		{$a_1$};
			\node[state]  (a2) [above=3cm of a1]            			{$a_2$};
			\node[state]  (a3) [right=2cm of a2]                		{$a_3$};
			\node[state]  (a4) [below=3cm of a3]               	    {$a_4$};
			
  \node[state]  (x0) [right=1cm of a4]             					 		{$x_0$};
  \node[state]  (x1) [above=1cm of x0] 	 		{$x_1$};
  \node[state]  (x2) [above=1cm of x1]  			 		{$x_2$};
  \node[state]  (x3) [right=1cm of x2]					{$x_3$};
  \node[state]  (x4) [right=1cm of x3]					{$x_4$};
  \node[state]  (x5) [right=1cm of x4]					{$x_5$};
  \node[state]  (x6) [right=1cm of x5]					{$x_6$};
  \node[state]  (x7) [right=1cm of x6]					{$x_7$};
  \node[state]  (x8) [right=1cm of x7]					{$x_8$};
  \node[state]  (x9) [below=1cm of x8]					{$x_9$};
  \node[state]  (x10) [below=1cm of x9]					{$x_{10}$};
  \node[state]  (x11) [left=1cm of x10]					{$x_{11}$};
  \node[state]  (x12) [left=1cm of x11]					{$x_{12}$};
  \node[state]  (x13) [left=1cm of x12]					{$x_{13}$};
  \node[state]  (x14) [left=1cm of x13]					{$x_{14}$};
  \node[state]  (x15) [left=1cm of x14]					{$x_{15}$};

			\node[state]  (b1) [right=1cm of x10]          			{$b_1$};
			\node[state]  (b2) [above=3cm of b1]            			{$b_2$};
			\node[state]  (b3) [right=2cm of b2]                		{$b_3$};
			\node[state]  (b4) [below=3cm of b3]               	    {$b_4$};
			\path[every node/.style={sloped,anchor=south,auto=false}]
			%(v) edge[line width=1.5pt, bend left=25] 	node {} (x1)        
(x0) edge[-, line width=2pt] 	node {$c$} (x1)            
(x1) edge[-,dashed,line width=2pt] 	node {} (x2) 
(x2) edge[-,line width=2pt] 	node {$c$} (x3) 
(x3) edge[-,dashed,line width=2pt] 	node {$c$} (x4) 
(x4) edge[-, line width=2pt] 	node {} (x5)
(x5) edge[-,dashed,line width=2pt] 	node {$c$} (x6) 
(x6) edge[-, line width=2pt] 	node {} (x7)
(x7) edge[-,dashed,line width=2pt] 	node {} (x8) 
(x8) edge[-, line width=2pt] 	node {$c$} (x9)
(x9) edge[-,dashed,line width=2pt] 	node {} (x10) 
(x10) edge[-, line width=2pt] 	node {} (x11)
(x11) edge[-,dashed,line width=2pt] 	node {} (x12) 
(x12) edge[-, line width=2pt] 	node {$c$} (x13)
(x13) edge[-,dashed,line width=2pt] 	node {} (x14) 
(x14) edge[-, line width=2pt] 	node {} (x15)
(x15) edge[-,dashed,line width=2pt] 	node {} (x0)
			%Path for the a and b cycles
			(a1) edge[-, line width=2pt] 	node {} (a2)   
			(a2) edge[-,dashed, line width=2pt] 	node {} (a3)   
			(a3) edge[-, line width=2pt] 	node {} (a4)   
			(a4) edge[-,dashed, line width=2pt] 	node {} (a1)   
			(b1) edge[-, line width=2pt] 	node {} (b2)   
			(b2) edge[-,dashed, line width=2pt] 	node {$c$} (b3)   
			(b3) edge[-, line width=2pt] 	node {} (b4)   
			(b4) edge[-,dashed, line width=2pt] 	node {$c$} (b1);
			
			\end{scope}
			\end{tikzpicture}}
		\caption{Symmetric difference $H = G \triangle G'$ where the solid edges represent the (blue) edges $G$ and the dashed edges the (red) edges of $G'$. From left to right the circuit are numbered $C_1 = a_1a_2a_3a_4a_1$, $C_2 = x_0\cdots x_{15}x_0$ and $C_3 = b_1b_2b_3b_4b_1$, and assume that this is also the order in which they are processed. Cut edges are indicated with the label $c$.} 
		\label{fig:sym_dif}
	\end{figure}
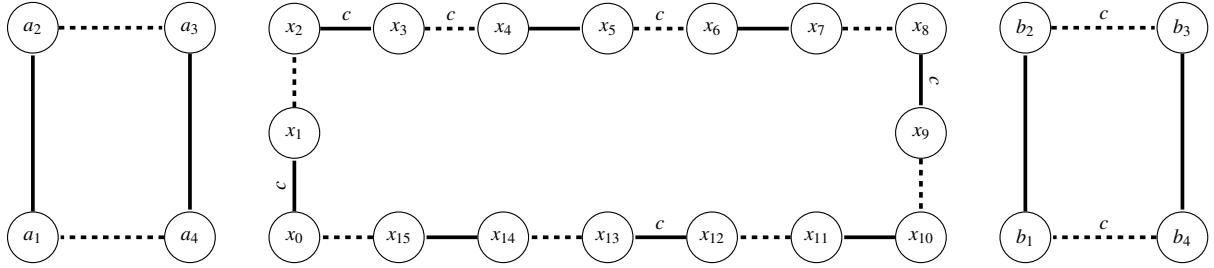

	\begin{figure}[ht!]
\centering
\scalebox{1}{
\begin{tikzpicture}%[scale=0.1]
\coordinate (i0) (0,0) {};
%\node [below left=0.1cm and 0.1cm] {$\mathbf{(0,0)}$};

\node at (i0) [circle,scale=0.7,fill=black] {};
\node (i1) [below right=0.75cm and 0.75cm of i0, circle,scale=0.7,fill=black] {};
\node (i2) [above right=0.7cm and 0.75cm of i1, circle,scale=0.7,fill=black] {};
\node (i3) [below right=0.75cm and 0.75cm of i2, circle,scale=0.7,fill=black] {};
\node (i4) [below right=0.75cm and 0.75cm of i3, circle,scale=0.7,fill=black] {};
\node (i5) [above right=0.75cm and 0.75cm of i4, circle,scale=0.7,fill=black] {};
\node (i6) [above right=0.75cm and 0.75cm of i5, circle,scale=0.7,fill=black] {};
%\node (i7) [below right=0.5cm and 0.5cm of i6, circle,scale=0.7,fill=black] {};
%\node (i8) [above right=0.5cm and 0.5cm of i7, circle,scale=0.7,fill=black] {};
%\node (i9) [below right=0.5cm and 0.5cm of i8, circle,scale=0.7,fill=black] {};
%\node (i10) [below right=0.5cm and 0.5cm of i9, circle,scale=0.7,fill=black] {};
%\node (i11) [below right=0.5cm and 0.5cm of i10, circle,scale=0.7,fill=black] {};
%\node (i12) [above right=0.5cm and 0.5cm of i11, circle,scale=0.7,fill=black] {};
%\node (i13) [below right=0.5cm and 0.5cm of i12, circle,scale=0.7,fill=black] {};
%\node (i14) [below right=0.4cm and 0.5cm of i13, circle,scale=0.7,fill=black] {};

\path[every node/.style={sloped,anchor=north,auto=false}]
(i0) edge[-,very thick] node {$S_1$} (i1)
(i1) edge[-,very thick] node {$S_2$} (i2)
(i2) edge[-,very thick] node {$S_3$} (i3)
(i3) edge[-,very thick] node {$S_4$} (i4)
(i4) edge[-,very thick] node {$S_5$} (i5)
(i5) edge[-,very thick] node {$S_6$} (i6);
%(i6) edge[-,very thick] node {} (i7)
%(i7) edge[-,very thick] node {} (i8)
%(i8) edge[-,very thick] node {} (i9)
%(i9) edge[-,very thick] node {} (i10)
%(i10) edge[-,very thick] node {} (i11)
%(i11) edge[-,very thick] node {} (i12)
%(i12) edge[-,very thick] node {} (i13)
%(i13) edge[-,very thick] node {} (i14);

%Axis
\draw (-1,0) -- (6,0);
%\draw (0,0) -- (0,-1.5);

\end{tikzpicture}}
\caption{The landscape, consisting of two valleys, corresponding to the symmetric difference in Figure \ref{fig:sym_dif}. The segments are given by $S_1 = (a_1a_2a_3a_4a_1,x_0x_1x_2)$, $S_2 = (x_2x_3x_4x_5x_6)$, $S_3 = (x_6x_7x_8x_9x_{10})	$, $S_4 = (x_{10}x_{11}x_{12}x_{13}x_{14})$, $S_5 = (x_{14}x_{15}x_{16}, b_1b_2)$, and $S_6 = (b_3b_4b_1)$.}
\label{fig:landscape_encoding}
\end{figure}
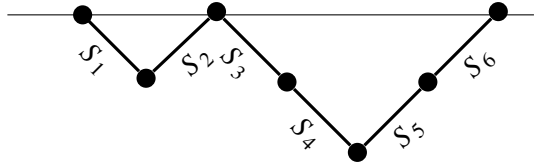

	\begin{figure}[ht!]
		\centering
		\scalebox{0.7}{
\begin{tikzpicture}[
			->,
			>=stealth',
			shorten >=1pt,
			auto,
			%node distance=2cm,
			semithick,
			every state/.style={circle,radius=0.1pt,text=black},
			]
			\begin{scope}
			\node[state]  (a1)               					 		{$a_1$};
			\node[state]  (a2) [above=3cm of a1]            			{$a_2$};
			\node[state]  (a3) [right=2cm of a2]                		{$a_3$};
			\node[state]  (a4) [below=3cm of a3]               	    {$a_4$};
			
  \node[state]  (x0) [right=1cm of a4]             					 		{$x_0$};
  \node[state]  (x1) [above=1cm of x0] 	 		{$x_1$};
  \node[state]  (x2) [above=1cm of x1]  			 		{$x_2$};
  \node[state]  (x3) [right=1cm of x2]					{$x_3$};
  \node[state]  (x4) [right=1cm of x3]					{$x_4$};
  \node[state]  (x5) [right=1cm of x4]					{$x_5$};
  \node[state]  (x6) [right=1cm of x5]					{$x_6$};
  \node[state]  (x7) [right=1cm of x6]					{$x_7$};
  \node[state]  (x8) [right=1cm of x7]					{$x_8$};
  \node[state]  (x9) [below=1cm of x8]					{$x_9$};
  \node[state]  (x10) [below=1cm of x9]					{$x_{10}$};
  \node[state]  (x11) [left=1cm of x10]					{$x_{11}$};
  \node[state]  (x12) [left=1cm of x11]					{$x_{12}$};
  \node[state]  (x13) [left=1cm of x12]					{$x_{13}$};
  \node[state]  (x14) [left=1cm of x13]					{$x_{14}$};
  \node[state]  (x15) [left=1cm of x14]					{$x_{15}$};
  
  	\node (arrow1) [below=0.2cm of x13]               	    {};
	\node (arrow2) [below=1.2cm of arrow1]               	  {};
	\node (arrow3) [below=0cm of arrow2]               	  {};

			\node[state]  (b1) [right=1cm of x10]          			{$b_1$};
			\node[state]  (b2) [above=3cm of b1]            			{$b_2$};
			\node[state]  (b3) [right=2cm of b2]                		{$b_3$};
			\node[state]  (b4) [below=3cm of b3]               	    {$b_4$};
			\path[every node/.style={sloped,anchor=south,auto=false}]
			%(v) edge[line width=1.5pt, bend left=25] 	node {} (x1)        
%(x0) edge[-, line width=2pt] 	node {$c$} (x1)            
(x1) edge[-,line width=2pt] 	node {} (x2) 
%(x2) edge[-,line width=2pt] 	node {$c$} (x3) 
(x3) edge[-,line width=2pt] 	node {} (x4) 
%(x4) edge[-, line width=2pt] 	node {} (x5)
(x5) edge[-,line width=2pt] 	node {} (x6) 
%(x5) edge[-,dashed,line width=0pt,color=white] 	node {$\textcolor{black}{c}$} (x6) 
%(x6) edge[-, line width=2pt] 	node {} (x7)
(x7) edge[-,line width=2pt] 	node {} (x8) 
%(x8) edge[-, line width=2pt] 	node {} (x9)
(x9) edge[-,line width=2pt] 	node {} (x10) 
(x10) edge[-, line width=2pt] 	node {} (x11)
%(x11) edge[-,dashed,line width=2pt] 	node {} (x12) 
(x12) edge[-, line width=2pt] 	node {} (x13)
%(x13) edge[-,dashed,line width=2pt] 	node {} (x14) 
(x14) edge[-, line width=2pt] 	node {} (x15)
%(x15) edge[-,dashed,line width=2pt] 	node {} (x0)
			%Path for the a and b cycles
			%(a1) edge[-, line width=2pt] 	node {} (a2)   
			(a2) edge[-,line width=2pt] 	node {} (a3)   
			%(a3) edge[-, line width=2pt] 	node {} (a4)   
			(a4) edge[-,line width=2pt] 	node {} (a1)   
			%(b1) edge[-, line width=2pt] 	node {} (b2)   
			(b2) edge[-,line width=2pt] 	node {} (b3)   
			(b3) edge[-, line width=2pt] 	node {} (b4)
			(arrow1) edge[line width=3.8pt] 	node {} (arrow2); ;
			%(b4) edge[-,dashed, line width=2pt] 	node {$c$} (b1);
			
			\end{scope}
			\end{tikzpicture}}	
		\quad
		\scalebox{0.7}{
			\begin{tikzpicture}[
			->,
			>=stealth',
			shorten >=1pt,
			auto,
			%node distance=2cm,
			semithick,
			every state/.style={circle,radius=0.1pt,text=black},
			]
			\begin{scope}
			\node[state]  (a1)               					 		{$a_1$};
			\node[state]  (a2) [above=3cm of a1]            			{$a_2$};
			\node[state]  (a3) [right=2cm of a2]                		{$a_3$};
			\node[state]  (a4) [below=3cm of a3]               	    {$a_4$};
			
  \node[state]  (x0) [right=1cm of a4]             					 		{$x_0$};
  \node[state]  (x1) [above=1cm of x0] 	 		{$x_1$};
  \node[state]  (x2) [above=1cm of x1]  			 		{$x_2$};
  \node[state]  (x3) [right=1cm of x2]					{$x_3$};
  \node[state]  (x4) [right=1cm of x3]					{$x_4$};
  \node[state]  (x5) [right=1cm of x4]					{$x_5$};
  \node[state]  (x6) [right=1cm of x5]					{$x_6$};
  \node[state]  (x7) [right=1cm of x6]					{$x_7$};
  \node[state]  (x8) [right=1cm of x7]					{$x_8$};
  \node[state]  (x9) [below=1cm of x8]					{$x_9$};
  \node[state]  (x10) [below=1cm of x9]					{$x_{10}$};
  \node[state]  (x11) [left=1cm of x10]					{$x_{11}$};
  \node[state]  (x12) [left=1cm of x11]					{$x_{12}$};
  \node[state]  (x13) [left=1cm of x12]					{$x_{13}$};
  \node[state]  (x14) [left=1cm of x13]					{$x_{14}$};
  \node[state]  (x15) [left=1cm of x14]					{$x_{15}$};

			\node[state]  (b1) [right=1cm of x10]          			{$b_1$};
			\node[state]  (b2) [above=3cm of b1]            			{$b_2$};
			\node[state]  (b3) [right=2cm of b2]                		{$b_3$};
			\node[state]  (b4) [below=3cm of b3]               	    {$b_4$};
			\path[every node/.style={sloped,anchor=south,auto=false}]
			%(v) edge[line width=1.5pt, bend left=25] 	node {} (x1)        
%(x0) edge[-, line width=2pt] 	node {$c$} (x1)            
(x1) edge[-,line width=2pt] 	node {} (x2) 
%(x2) edge[-,line width=2pt] 	node {$c$} (x3) 
(x3) edge[-,line width=2pt] 	node {} (x4) 
%(x4) edge[-, line width=2pt] 	node {} (x5)
(x5) edge[-,line width=2pt] 	node {} (x6) 
%(x5) edge[-,dashed,line width=0pt,color=white] 	node {$\textcolor{black}{c}$} (x6) 
%(x6) edge[-, line width=2pt] 	node {} (x7)
(x7) edge[-,line width=2pt] 	node {} (x8) 
%(x8) edge[-, line width=2pt] 	node {} (x9)
(x9) edge[-,line width=2pt] 	node {} (x10) 
%(x10) edge[-, line width=2pt] 	node {} (x11)
(x11) edge[-,line width=2pt] 	node {} (x12) 
(x12) edge[-, line width=2pt] 	node {} (x13)
%(x13) edge[-,dashed,line width=2pt] 	node {} (x14) 
(x14) edge[-, line width=2pt] 	node {} (x15)
%(x15) edge[-,dashed,line width=2pt] 	node {} (x0)
			%Path for the a and b cycles
			%(a1) edge[-, line width=2pt] 	node {} (a2)   
			(a2) edge[-,line width=2pt] 	node {} (a3)   
			%(a3) edge[-, line width=2pt] 	node {} (a4)   
			(a4) edge[-,line width=2pt] 	node {} (a1)   
			%(b1) edge[-, line width=2pt] 	node {} (b2)   
			(b2) edge[-,line width=2pt] 	node {} (b3)   
			(b3) edge[-, line width=2pt] 	node {} (b4);
			%(b4) edge[-,dashed, line width=2pt] 	node {$c$} (b1);
			
			\end{scope}
			\end{tikzpicture}}
		\caption{The transition $\tau = (Z,Z')$ that is the hinge flip operation that removes the edge $\{x_{10},x_{11}\}$ and adds the edge $\{x_{11},x_{12}\}$ as part of the unwinding of $S_4$. Note that the segments $S_1$ and as $S_2$, forming the first valley, have been processed already. Also, the first segment $S_3$ of the left part of the second valley, as well as the segment $S_5$ being the first segment of the right part of the second valley, have been processed already. The edges in $(E(G) \cup E(G')) \setminus E(H)$ are left out.} 
		\label{fig:transition}
	\end{figure}
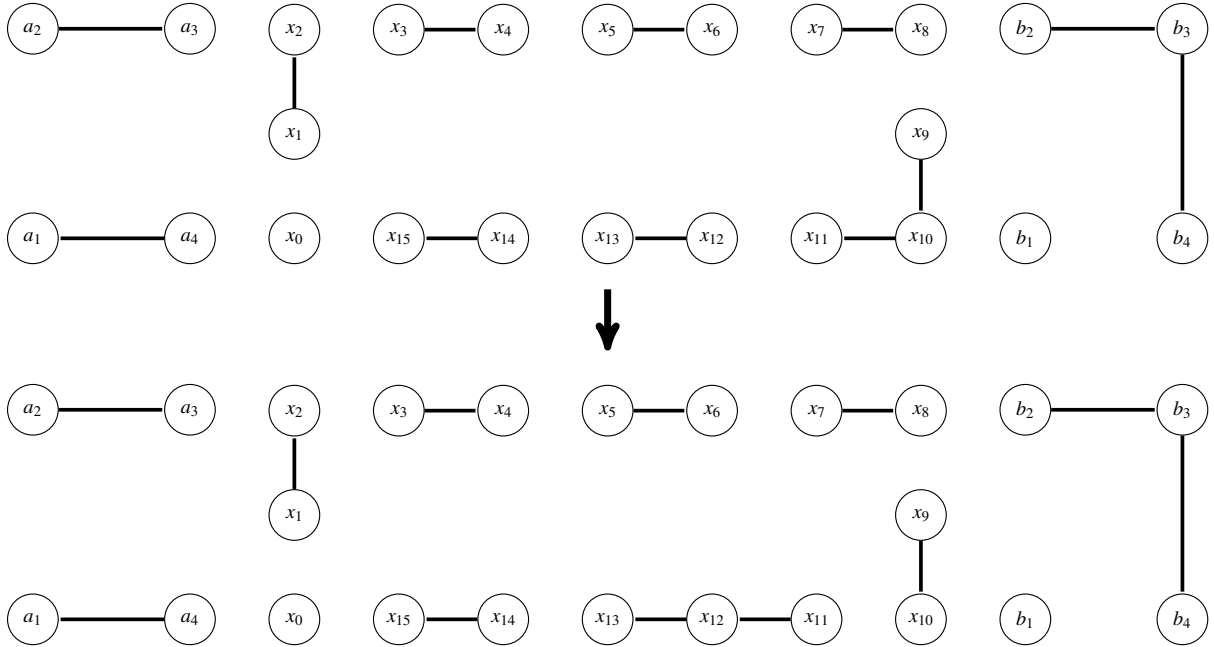

	\begin{figure}[ht!]
		\centering
		\scalebox{0.7}{
			\begin{tikzpicture}[
			->,
			>=stealth',
			shorten >=1pt,
			auto,
			%node distance=2cm,
			semithick,
			every state/.style={circle,radius=0.1pt,text=black},
			]
			\begin{scope}
			\node[state]  (a1)               					 		{$a_1$};
			\node[state]  (a2) [above=3cm of a1]            			{$a_2$};
			\node[state]  (a3) [right=2cm of a2]                		{$a_3$};
			\node[state]  (a4) [below=3cm of a3]               	    {$a_4$};
			
  \node[state]  (x0) [right=1cm of a4]             					 		{$x_0$};
  \node[state]  (x1) [above=1cm of x0] 	 		{$x_1$};
  \node[state]  (x2) [above=1cm of x1]  			 		{$x_2$};
  \node[state]  (x3) [right=1cm of x2]					{$x_3$};
  \node[state]  (x4) [right=1cm of x3]					{$x_4$};
  \node[state]  (x5) [right=1cm of x4]					{$x_5$};
  \node[state]  (x6) [right=1cm of x5]					{$x_6$};
  \node[state]  (x7) [right=1cm of x6]					{$x_7$};
  \node[state]  (x8) [right=1cm of x7]					{$x_8$};
  \node[state]  (x9) [below=1cm of x8]					{$x_9$};
  \node[state]  (x10) [below=1cm of x9]					{$x_{10}$};
  \node[state]  (x11) [left=1cm of x10]					{$x_{11}$};
  \node[state]  (x12) [left=1cm of x11]					{$x_{12}$};
  \node[state]  (x13) [left=1cm of x12]					{$x_{13}$};
  \node[state]  (x14) [left=1cm of x13]					{$x_{14}$};
  \node[state]  (x15) [left=1cm of x14]					{$x_{15}$};

			\node[state]  (b1) [right=1cm of x10]          			{$b_1$};
			\node[state]  (b2) [above=3cm of b1]            			{$b_2$};
			\node[state]  (b3) [right=2cm of b2]                		{$b_3$};
			\node[state]  (b4) [below=3cm of b3]               	    {$b_4$};
			\path[every node/.style={sloped,anchor=south,auto=false}]
			%(v) edge[line width=1.5pt, bend left=25] 	node {} (x1)        
(x0) edge[-, line width=2pt] 	node {} (x1)            
%(x1) edge[-,line width=2pt] 	node {} (x2) 
(x2) edge[-,line width=2pt] 	node {} (x3) 
%(x3) edge[-,line width=2pt] 	node {} (x4) 
(x4) edge[-, line width=2pt] 	node {} (x5)
%(x5) edge[-,line width=2pt] 	node {} (x6) 
%(x5) edge[-,dashed,line width=0pt,color=white] 	node {$\textcolor{black}{c}$} (x6) 
(x6) edge[-, line width=2pt] 	node {} (x7)
%(x7) edge[-,line width=2pt] 	node {} (x8) 
(x8) edge[-, line width=2pt] 	node {} (x9)
%(x9) edge[-,line width=2pt] 	node {} (x10) 
%(x10) edge[-, line width=2pt] 	node {} (x11)
(x11) edge[-,line width=2pt] 	node {} (x12) 
%(x12) edge[-, line width=2pt] 	node {} (x13)
(x13) edge[-,line width=2pt] 	node {} (x14) 
%(x14) edge[-, line width=2pt] 	node {} (x15)
(x15) edge[-,line width=2pt] 	node {} (x0)
			%Path for the a and b cycles
			(a1) edge[-, line width=2pt] 	node {} (a2)   
			%(a2) edge[-,line width=2pt] 	node {} (a3)   
			(a3) edge[-, line width=2pt] 	node {} (a4)   
			%(a4) edge[-,line width=2pt] 	node {} (a1)   
			(b1) edge[-, line width=2pt] 	node {} (b2)   
			%(b2) edge[-,line width=2pt] 	node {} (b3)   
			%(b3) edge[-, line width=2pt] 	node {} (b4);
			(b4) edge[-,line width=2pt] 	node {} (b1);
			
			\end{scope}
			\end{tikzpicture}}
		\caption{The encoding $L = L_t(G,G') = (G \triangle G') \triangle Z$ for the symmetric difference in Figure \ref{fig:sym_dif} and transition as in Figure \ref{fig:transition}, where again the edges in $(E(G) \cup E(G')) \setminus E(H)$ are left out.} 
		\label{fig:encoding}
	\end{figure}
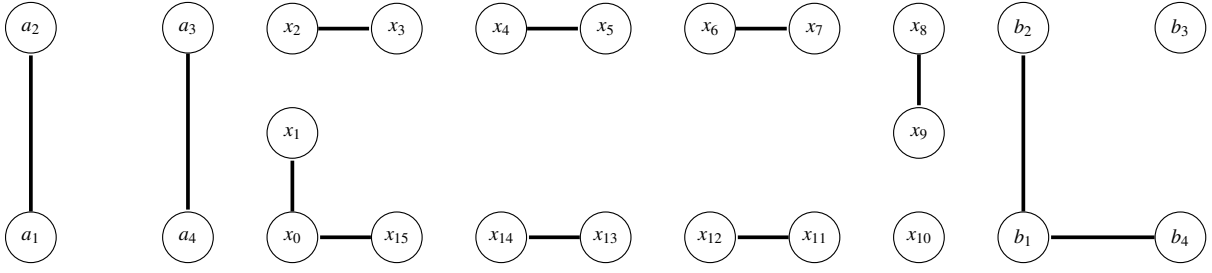

\subsubsection{Bounding the Congestion}
For a tuple $(G,G',\psi)$, let $p_{\psi}(G,G')$ denote the canonical path from $G$ to $G'$ for pairing $\psi$.
Let 
\[
\mathcal{L}_\tau = \cup_{(G,G',\psi) \in \mathcal{F}_\tau} L_\tau(G,G')
\]
be the union of all (distinct) encodings $L_\tau$, where 
$
\mathcal{F}_\tau = \{(G,G',\psi) : \tau \in p_{\psi}(G,G')\}
$ 
is the set of all tuples $(G,G',\psi)$ such that the canonical path from $G$ to $G'$ under pairing $\psi$ uses the transition $\tau$. A crucial observation is the following.

\begin{lemma}
If $L_{\tau}(G,G') = (G \triangle G') \triangle Z$ for transition $\tau = (Z,Z')$ used by a canonical path between $G$ and $G'$, then $L \in \mathcal{G}'(\gamma,d)$. This implies that 
\begin{equation}\label{eq:injective1_jdm}
|\mathcal{L}_\tau| \leq |\mathcal{G}'(\gamma,d)|.
\end{equation}
\end{lemma}
\begin{proof}[Proof (sketch)]
We check that the properties (i), (ii) and (iii) defining the set $\mathcal{G}'(\gamma,d)$ are satisfied by $L$. Note that $L \triangle Z = G \triangle G$. As every individual hinge flip operations adds and removes an arc from the symmetric difference, it follows that $L$ and $Z$ have the same number of edges. This proves property (i). Also, if $Z$ has a perturbation of $\alpha_v \in \{-2,-1,0,-1,-2\}$ (see Proposition \ref{prop:basic_properties}) at node $v$, then $L$ has a perturbation of $-\alpha_v$ at node $v$, which shows that property (ii) is satisfied for $L$. Finally, with $\beta \in \{-1,0,1\}$, if $Z$ contains $\gamma - \beta$ cut edges, then $L$ contains $\gamma + \beta$ cut edges (using implicitly that $G$ and $G'$ contain the same number of cut edges). This implies that property (iii) is satisfied.
\end{proof}

	Moreover, with $H = G \triangle G'$ and $L=L_t(G,G')$, the pairing $\psi$ has the property that it pairs up the edges of $E(H)\mysetminus E(L)$ and $E(H)\cap E(L)$ in such a way that for every node $v$  each edge in $E(H)\mysetminus E(L)$ that is incident to $v$ is paired up with an edge in $E(H)\cap E(L)$ that is incident to $v$, except for at most four pairs.\footnote{These are the nodes $x_0, x_{12}, b_1$ and $b_3$ in Figure \ref{fig:encoding}.} %However, there are either two nodes for which the incident edges in $E(H)\mysetminus E(L)$ exceed by 2 the incident edges in $E(H)\cap E(L)$, or one node for which the incident edges in $E(H)\mysetminus E(L)$ exceed by 4 the incident edges in $E(H)\cap E(L)$. These are exactly the two nodes with degree deficit 1 or the one node with degree deficit 2 in $L$; for the example in Figure \ref{fig:encoding3} these are nodes $x_1$ and $x_6$. There $\psi$ pairs up each edge of $E(H)\cap E(L)$ to an edge of $E(H)\mysetminus E(L)$ but also two edges of $E(H)\mysetminus E(L)$ with each other; or in the case of one node with degree deficit 2 $\psi$ pairs up each edge of $E(H)\cap E(L)$ to an edge of $E(H)\mysetminus E(L)$ but also makes two pairs out of the remaining 4 edges in $E(H)\mysetminus E(L)$. 
	Let $\Psi'(L)$ be the set of all pairings with this property. Remember that we do not need to know $G$ and $G'$ in order to determine the set $H = L_\tau \triangle Z$. Note that not every such pairing has to correspond to a tuple $(G,G',\psi)$ for which $t \in p_{\psi}(G,G')$. 
Using a counting argument,\footnote{This can be done similarly as the argument used towards the end of Appendix \ref{app:js}.} we can upper bound $|\Psi'(L)|$ in terms of $|\Psi(H)|$. In particular, there exists a polynomial $q(n)$ such that 
	\begin{equation}\label{eq:injective2_jdm}
	%|\Psi'(L)| = \left( \Pi_{v \in V\mysetminus\{u, %w\}}\theta_v! \right) \cdot \frac{(\theta_u+1)!}{2} %\cdot \frac{(\theta_w+1)!}{2} = |\Psi(H)| \cdot 
	%\frac{(\theta_u+1)(\theta_w+1)}{4} \leq n^2 \cdot |
	%\Psi(H)| \,.
	|\Psi'(L)| \leq q(n)\cdot \Psi(H)|.\footnote{A very rough choice is $q(n) = n^{20}$.}
	\end{equation}
	Putting everything together, we have
	\begin{eqnarray}\label{eq:injective3}
	|\mathcal{G}'(\gamma,d)|^2 f'(\tau)  & = & \sum_{(G,G')} \sum_{\psi \in \Psi(G,G')} \mathbf{1}(e \in p_{\psi}(H)) |\Psi(H)|^{-1}  \nonumber \\
	& \leq & \frac{1}{8}n^4 \sum_{L \in \mathcal{L}_\tau} \sum_{\psi' \in \Psi'(L)}  |\Psi(H)|^{-1} \ \ \ \ \ (\text{using Lemma } \ref{lem:recovery}) \nonumber \\
	& \leq & \frac{1}{8}n^4 \cdot q(n) \sum_{L \in \mathcal{L}_\tau} 1 \ \  \ \ \ \ \ \ \ \ \ \ \ \ \ \ \ \ (\text{using } (\ref{eq:injective2_jdm})) \nonumber \\
	& \leq & \frac{1}{8}n^4\cdot q(n) \cdot|\mathcal{G}'(\gamma,d)|  \ \  \ \ \ \ \ \ \ \ \ \ \ \ (\text{using } (\ref{eq:injective1_jdm}))
	\end{eqnarray}
	The usage of Lemma \ref{lem:recovery_jdm} for the first inequality works as follows. Every tuple $(G,G',\psi)\in \mathcal{F}_t$ with encoding $L_t(G, G')$ generates a unique tuple in $\{L_t(G, G')\} \times \Psi'(L_t(G, G'))$. But since, by Lemma \ref{lem:recovery_jdm}, there are at most $\frac{1}{8}n^4$ pairs $(G,G')$ with $L = L_{\tau}(G,G')$ for given $L$, $\tau$ and $\psi$, we have that $\frac{1}{8}n^4 \sum_{L \in \mathcal{L}_\tau} |\{L\} \times \Psi'(L)| = \frac{1}{8}n^4 \sum_{L \in \mathcal{L}_\tau} \sum_{\psi' \in \Psi'(L)} 1$ is an upper bound on the number of canonical paths that use $\tau$. 
	
	By rearranging \eqref{eq:injective3} we get the upper bound for  $f'$ required in Lemma \ref{lem:flow_simplification_jdm}. We already observed that the length of any canonical path is polynomially bounded as well. This then completes the proof of Theorem \ref{thm:auxiliary}. % What is left to show is that $\ell(f')$ is not too large. This, however, is determined by the way we defined the canonical paths. It is easy to see that for any canonical path between any two graphs $G, G' \in \mathcal{G}(d)$ has length at most $\frac{3}{4} |E(G \triangle G')|$ and, therefore, $\ell(f')\le n^2$.

%\begin{remark}[Bipartite case]\label{rem:app_bip}
%The proof for the bipartite case is very similar. The only difference is that in the circuit processing procedure there will never be an auxiliary state where one node has degree deficit two. For this to occur there necessarily has to be a simple cycle of odd length in a circuit, see, e.g., Figure \ref{fig:step1_2}. This explains the adjusted definition of stability for the bipartite case as in Section \ref{sec:bipartite}. Moreover, we can use the same encoding and injective mapping arguments.
%\end{remark}

%\newpage
\bigskip

%%%%%%%%%%%%%%%%%%%%%%%%%%%%%%%%%%%%%%%%
%%%%%%%%%%%%%%%%%%%%%%%%%%%%%%%%%%%%%%%%
\section{Strongly Stable Families for the PAM Model}\label{sec:stable_jdm}
%%%%%%%%%%%%%%%%%%%%%%%%%%%%%%%%%%%%%%%%
%%%%%%%%%%%%%%%%%%%%%%%%%%%%%%%%%%%%%%%%
In Appendix \ref{sec:auxiliary} we have shown that the hinge flip Markov chain for PAM instances with two classes is rapidly mixing on $\mathcal{G}'(\gamma,d)$ in case $(\gamma,d)$ comes from a family of strongly stable tuples. In this section we give two explicit families of sequences that are strongly stable. When dealing with a family of instances, even when this is not explicitly mentioned, we  only consider the tuples $(c,d)$ for which there is at least one graphical realization.

\begin{theorem}[Regular classes]\label{cor:stable_jdm}
Let $\mathcal{D}$ be the family of instances of the joint degree matrix model, i.e., where for every tuple $(V_1,V_2,\gamma,d)$ it holds that 
$
1 \leq \beta_1,\beta_2 \leq |V| - 1,
$
and $1 \leq \gamma \leq |V_1||V_2| - 1$,
%$3 \leq \beta_i \leq |V_i| - 3$ for $i = 1,2$. 
where $\beta_1$ and $\beta_2$ are the common degrees in the classes $V_1$ and $V_2$, respectively. The family $\mathcal{D}$ is strongly stable for $k = 6$, and, hence, the hinge flip chain is rapidly mixing for all tuples in $\mathcal{D}$.
\end{theorem}
\begin{proof}
We first show that this family is strongly stable for $k = 6$. For convenience, we will work with the notation $\mathcal{G}'(c,d)$ instead of $\mathcal{G}'(\gamma,d)$. Remember that 
\[
c_{ii} = \left( \sum_{j \in V_i} d_j \right) - \gamma
\]
for $i = 1,2$ is the number of internal edges that $V_i$ has in any graphical realization in $\mathcal{G}(\gamma,d)$, and that $\gamma = c_{12} = c_{21}$.  For sake of readability, we define the notion of a cancellation hinge flip. For either $i = 1$ or $i = 2 $, suppose nodes $v,w \in V_i$, are such that $v$ has a degree deficit of at least one, and $w$ a degree surplus of at least one. Then $w$ has a neighbor $z \in V$ that is not a neighbor of $v$ (using that $v$ and $w$ have the same degree $\beta_i$). The hinge flip operation that removes the edge $\{z,w\}$ and adds the edge $\{z,v\}$ is called a \emph{cancellation flip on $v$ and $w$}. Note that the number of internal edges in $V_1$ and $V_2$ as well as the number of cut edges does not change with such an operation.\footnote{That is, either $z$ lies in the other class, in which case the cancellation flip removes and adds a cut edge, or, $z$ lies in the same class as $v$ and $w$ in which case an internal edge in $V_i$ is removed and added.} Moreover, we say that an edge $\{a,b\}$ is a non-edge of a graphical realization $G$ if $\{a,b\} \notin E(G)$.

Let $G \in \mathcal{G}'(c,d)$ for some tuple $(c',d')$ as in the definition of $\mathcal{G}'(c,d)$ at the start of Section \ref{sec:auxiliary}. We first show that with at most four hinge flip operations, we can obtain a perturbed auxiliary state $G^* \in \mathcal{G}'(c,d)$ for which its tuple $(c^*,d^*)$  is edge-balanced. That is, it satisfies $c^* = c$. Remember that the value $c'_{12}$ uniquely determines the matrix $c'$, and, by assumption of $\mathcal{G}'(c,d)$, we have $c_{12}' \in \{c_{12}-1,c_{12},c_{12}+1\}$. We can therefore distinguish the following cases.

%We first show that if $(c_{11}',\gamma',c_{22}') \neq (c_{11},\gamma,c_{22})$, with $c_{ii}'$ the number of internal edges of $V_i$ in the realization $G$, then we can find some $G^* \in \mathcal{G}'(\gamma,d)$ at distance at most four from $G$ that has partition adjacency matrix $c^*$ satisfying $c^* = c$. That is, we can obtain such a $G^*$ from $G$ using at most four hinge flip operations.

\begin{itemize}
\item \textbf{Case 1: $c_{12}' = c_{12} + 1.$} Then, by Proposition \ref{prop:basic_properties}, either $c_{11}' = c_{11} - 1$ and $c_{22}' = c_{22}$, or, $c_{22}' = c_{22} - 1$ and $c_{11}' = c_{11}$. Assume without loss of generality that we are in the first case. Then it holds that 
\begin{equation}\label{eq:regular1}
\sum_{j \in V_1} d_j' = \left( \sum_{j \in V_1} d_j \right) - 1 \ \ \ \text{ and } \ \ \ \sum_{j \in V_2} d_j' = \left( \sum_{j \in V_2} d_j \right) + 1. 
\end{equation}
Moreover, there must be at least one node $v_2 \in V_2$ with a degree surplus (of either one or two), and there is at least one non-edge $\{a,b\}$ with both endpoints in $V_1$. If $v_2$ is adjacent to either $a$ or $b$, we can perform a hinge flip to make the realization $G$ edge-balanced, so assume this is not the case. Also, if the total deficit of $a$ and $b$ is $-2$, there must be a node in $V_1$ with degree surplus, otherwise (\ref{eq:regular1}) is violated. Then we can perform a cancellation flip in $V_1$ to remove the deficit at either $a$ or $b$. Hence, we may assume without loss of generality that $a$ does not have a degree deficit. 
\begin{itemize}
\item \textbf{Case A: $v_2$ has a neighbor $v_1 \in V_1$.} If $v_1$ has a degree surplus we can perform a cancellation flip in $V_1$ to remove it, which must exist by (\ref{eq:regular1}). So assume $v_1$ has no degree surplus. As node $a$ has no deficit, and is not connected to $v_2$, whereas $v_1$ is, there must be some neighbor $p$ of $a$ which is not a neighbor of $v_1$. This holds since $v_1$ and $b$ have the same degree $\beta_1$ in the sequence $d$. Then the path $v_2 - v_1 - p - a - b$ alternates between edges and non-edges of $G$, and with two hinge flips we can obtain an edge-balanced realization in $\mathcal{G}'(\gamma,d)$.
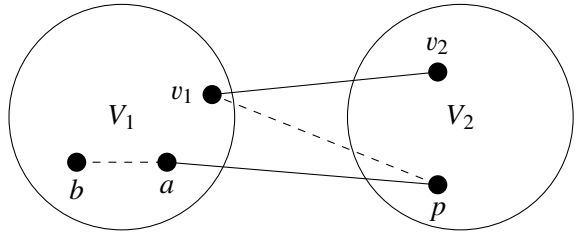
\begin{figure}[ht!]
\centering
\begin{tikzpicture}[scale=0.6]
\coordinate (a) at (0,0.5); 
\coordinate (b) at (2,0.5);
\coordinate (p) at (8,0);
\coordinate (v1) at (3,2);
\coordinate (v2) at (8,2.5);

\node at (b) [circle,scale=0.7,fill=black] {};
\node (B) [below=0.1cm of a]  {$b$};
\node at (a) [circle,scale=0.7,fill=black] {};
\node (A) [below=0.1cm of b]  {$a$};
\node at (p) [circle,scale=0.7,fill=black] {};
\node (P) [below=0.1cm of p]  {$p$};
\node at (v1) [circle,scale=0.7,fill=black] {};
\node (V1) [left=0.1cm of v1]  {$v_1$};
\node at (v2) [circle,scale=0.7,fill=black] {};
\node (V2) [above=0.1cm of v2]  {$v_2$};

\path[every node/.style={sloped,anchor=south,auto=false}]
(a) edge[-,dashed] node {} (b)
(b) edge[-,] node {} (p)
(p) edge[-,dashed] node {} (v1)
(v1) edge[-] node {} (v2);

\draw (1,1.5) circle (2.5cm) node{$V_1$};
\draw (8.5,1.5) circle (2.5cm) node{$V_2$};
\end{tikzpicture}
\caption{Sketch of first case with subcase A.}
\end{figure} 

\item \textbf{Case B: $v_2$ has no neighbors in $V_1$.} We know that there is at least one cut edge $\{q,r\}$ in the realization $G$, since $c_{12}' = c_{12} + 1$. If $q$ has a degree surplus, we are in the situation of Case A. Otherwise $v_2$ has a neighbor $u$ which is not a neighbor of $q$, since $q$ and $v_2$ have the same degree $\beta_2$ in the sequence $d$. We can then perform the hinge flip that removes $\{v_2,u\}$ and adds $\{u,q\}$. If $q$ now has a degree surplus, we are in Case A. Otherwise, in case this hinge flip cancelled out a degree deficit at $q$, there must be at least one other node in $V_2$ with a degree surplus, because of (\ref{eq:regular1}). We can then perform the same step again, which will now result in a degree surplus at $q$. This is true since the node $q$ cannot have a deficit of $-2$, since (\ref{eq:regular1}) would then imply that the the total degree surplus of nodes in $V_2$ is at least three, which violates the second property defining $\mathcal{G}'(c,d)$. That is, we can always reduce to the situation in Case A.
\end{itemize}

Summarizing, we can always find an edge-balanced realization $G^*$ using at most four hinge flip operations in case $c_{12}' = c_{12} + 1$.

\item \textbf{Case 2: $c_{12}' = c_{12} - 1.$} Using complementation, it can be seen that this case is similar to Case 1. That is, we consider the tuple $(\bar{c},\bar{d})$ in which all nodes in $V_1$ have degree $|V| - \beta_1$, all nodes in $V_2$ have degree $|V| - \beta_2$, and where all feasible graphical realizations have $\bar{c}_{12} = |V_1||V_2| - c_{12}$ cut edges. The case $c_{12} = c_{12} - 1$ then corresponds to the case $\bar{c}_{12}' = \bar{c}_{12} + 1$.
 %Then, by Proposition \ref{prop:basic_properties}, either $c_{11}' = c_{11} + 1$ and $c_{22}' = c_{22}$, or, $c_{22}' = c_{22} + 1$ and $c_{11}' = c_{11}$. Assume without loss of generality that we are in the first case. Then it holds that 
%\begin{equation}\label{eq:regular1}
%\sum_{j \in V_1} d_j' = \left( \sum_{j \in V_1} d_j \right) + 1 \ \ \ \text{ and } \ \ \ \sum_{j \in V_2} d_j' = \left( \sum_{j \in V_2} d_j \right) - 1. 
%\end{equation}
%Moreover, there is at least one non-edge $\{a,b\}$ with $a \in V_1$ and $b \in V_2$.

\item \textbf{Case 3: $c_{12}' = c_{12}$}. If also $c_{11} = c_{11}'$ we are done. Otherwise, suppose that $c_{11} = c_{11}' + 1$. Then it must be that $c_{22} = c_{22}' - 1$, as $c_{12}' = c_{12}$, and it holds that 
\begin{equation}\label{eq:regular2}
\sum_{j \in V_1} d_j' = \left( \sum_{j \in V_1} d_j \right) + 2 \ \ \ \text{ and } \ \ \ \sum_{j \in V_2} d_j' = \left( \sum_{j \in V_2} d_j \right) - 2. 
\end{equation}
Then there is at least one edge $\{a,b\}$ in the graphical realization with $a,b \in V_1$. Moreover, we may assume that $a$ has a degree surplus. If not, then there is at least one other node $u$ with a degree surplus because of (\ref{eq:regular2}). Performing a cancellation flip then gives the node $a$ a degree surplus (it could not be that $a$ had a degree deficit, as this would imply, in combination with (\ref{eq:regular2}), that the total degree surplus of nodes in $V_1$ is at least three).
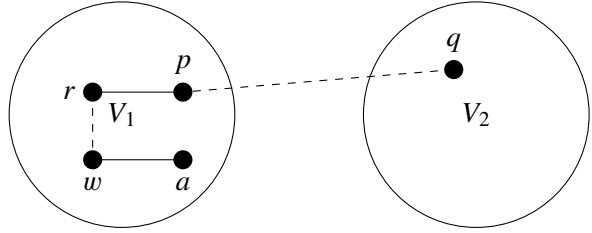
\begin{figure}[ht!]
\centering
\begin{tikzpicture}[scale=0.6]
\coordinate (w) at (0,0.5); 
\coordinate (a) at (2,0.5);
\coordinate (r) at (0,2);
\coordinate (p) at (2,2);
\coordinate (q) at (8,2.5);

\node at (a) [circle,scale=0.7,fill=black] {};
\node (A) [below=0.1cm of a]  {$a$};
\node at (w) [circle,scale=0.7,fill=black] {};
\node (W) [below=0.1cm of w]  {$w$};
\node at (r) [circle,scale=0.7,fill=black] {};
\node (R) [left=0.1cm of r]  {$r$};
\node at (p) [circle,scale=0.7,fill=black] {};
\node (P) [above=0.1cm of p]  {$p$};
\node at (q) [circle,scale=0.7,fill=black] {};
\node (Q) [above=0.1cm of q]  {$q$};

\path[every node/.style={sloped,anchor=south,auto=false}]
(q) edge[-,dashed] node {} (p)
(p) edge[-,] node {} (r)
(r) edge[-,dashed] node {} (w)
(w) edge[-] node {} (a);

\draw (0.65,1.5) circle (2.5cm) node{$V_1$};
\draw (8.5,1.5) circle (2.5cm) node{$V_2$};
\end{tikzpicture}
\caption{Sketch of last situation in Case 3.}
\end{figure}

 Now, if there is a non-edge of the form $\{b,v_2\}$ for some $v_2 \in V_2$, we can perform a hinge flip operation removing $\{a,b\}$ and adding $\{b,v_2\}$ in order to end up in Case 1. Otherwise, assume that $b$ is adjacent to all $v_2 \in V_2$. As $b$ is also adjacent to $a$, and $a$ has a degree surplus of at least one,\footnote{That is, $b$ can have a degree surplus of at most one. A degree surplus of two at $b$ would only give a bound of $\beta_1 \geq |V_2| - 1$.} it follows that $\beta_1 \geq |V_2|$.
 Now, by the assumption that $c_{12} \leq |V_1||V_2| - 1$, there is at least one non-edge $\{p,q\}$ with $p \in V_1$ and $q \in V_2$. As $p$ is not adjacent to $q$, but has degree at least $\beta_1 \geq |V_2|$, it must be that $p$ is adjacent to some $r \in V_1$. If $r$ has a degree surplus, then we can perform a hinge flip that removes $\{p,r\}$ and adds $\{p,q\}$ in order to end up in the situation of Case 1. Otherwise, node $a$, which has a degree surplus, has some neighbor $w$ which is not a neighbor of $r$. This implies the path $a - w - r - p - q$ alternates between edges and non-edges of $G$. Performing two hinge flips then brings us in the situation of Case 1.
\end{itemize}
%From now on we may assume that $G \in \mathcal{G}(\gamma,d)$ is such that the number of cut edges is precisely $\gamma$, and the number of internal edges in $V_{i}$ is $c_{ii}$ for $i = 1,2$. Now suppose there is some node $v_1 \in V_1$ that has a degree surplus. Then, by assumption that $G$ has the correct partition adjacency matrix, there must also be some node $w_1 \in V_1$ with degree deficit. As both $v_1$ and $w_1$ have degree $\beta_1$, it follows that $v_1$ has a neighbor $a$ that is not a neighbor of $w_1$. Performing a hinge flip operation that removes $\{v_1,a\}$ and add $\{a,w_1\}$ reduces the total degree deficit. Note that it is not relevant whether $a$ lies in $V_1$ or $V_2$. Moreover, a similar argument holds in case there is some node with degree surplus in $V_2$. Repeating this step one more time if there is still a node with degree surplus, we can reduce the total (absolute) degree deficit to zero. This means that in at most six hinge flip operations in total, we can always transform any $G \in \mathcal{G}'(\gamma,d)$ to a state in $\mathcal{G}(\gamma,d)$.

We have shown that with at most four hinge flips we can always obtain some $G^* \in \mathcal{G}'(c,d)$ that is edge-balanced. This implies that
\begin{equation}\label{eq:regular3}
\sum_{j \in V_1} d_j^* =  \sum_{j \in V_1} d_j \ \ \ \text{ and } \ \ \ \sum_{j \in V_2} d_j^* =  \sum_{j \in V_2} d_j. 
\end{equation}
Now, if $v \in V_1$ has a degree surplus, there must be some $w \in V_1$ that has a degree deficit, because of (\ref{eq:regular3}). We can then perform a cancellation flip to decrease the sum of the total degree deficit and degree surplus. A similar statement is true if $v \in V_2$ has a degree surplus. By performing this step at most twice, we obtain a realization $H \in \mathcal{G}(c,d)$. That is, with at most six hinge flip operations in total we can transform $G$ into a graphical realization in $\mathcal{G}(c,d)$. This shows that $\mathcal{D}$ is strongly stable for $k = 6$.
\end{proof}

We next show that certain families of sparse instances are strongly stable as well. Sparsity here refers to the fact that the maximum degree in a class in significantly smaller than the size of the class, as well as the fact that the number of cut edges is (much) smaller than the total number of edges in a graphical realization.

%\pkrem{Originally there were lower bounds on $c_{11},c_{22}$, but this is replaced by upper bound on $c_{12}$. That is, not only the degree sequence is sparse, but also the number of cut edges has to be sufficiently sparse.}
\begin{theorem}[Sparse irregular families]\label{cor:stable_sparse}
Let $0 < \alpha < 1/2$ be fixed, and let  $\mathcal{D}_{\alpha}$ be the family of tuples $(V_1,V_2,c,d)$ for which
\begin{enumerate}[label=(\roman*)]
\item $|V_1|, |V_2| \geq \alpha n$, 
\item $2 \leq d_i \leq  \sqrt{\frac{\alpha n}{4}} \ \ $\  for $i \in V_1 \cup V_2$, and, 
%\item $c_{11},c_{22} \geq\frac{2n}{\alpha}$, and $\frac{\alpha n}{2} \geq c_{12} \geq 1$.
\item  $1 \leq c_{12} \leq \frac{\alpha n}{2}$.
\end{enumerate}
The class $\mathcal{D}_{\alpha}$ is strongly stable for $k = 9$, and, hence, the hinge flip chain is rapidly mixing for all tuples in $\mathcal{D}_{\alpha}$.
\end{theorem}
\begin{proof}
We proceed in a similar fashion as the proof of Corollary \ref{cor:stable_jdm}. We will use the notion of an alternating path. For a given graph $H = (W,E)$, an alternating path $(x_1,\dots,x_q)$ is an odd sequence of nodes so that $\{x_i,x_{i+1}\} \in E$ for $i$ even, and $\{x_i,x_{i+1}\} \notin E$ for $i$ odd.% We first give a lemma following from \cite{Jerrum1989} that will be useful in the analysis.

\begin{lemma}[following from \cite{Jerrum1989}]\label{lem:alternating_path}
Let $\delta = (\delta_1,\dots,\delta_r)$ be a degree sequence with $1 \leq \delta_i \leq \sqrt{r/2}\ $ for all $i = 1,\dots,r$. %and let $\mathcal{G}(d)$ be the set of all labelled graphical realizations of this sequence on $r$ nodes. Let $\mathcal{G}'(d) = \cup_{d'} \mathcal{G}(d')$ with $d'$ satisfying $\sum_{i = 1}^r |d_i - d_i'| = 2$, i.e., either $d' = d$ or there exist two nodes $x,y$ so that $d_x' = d_x - 1$, $d_y' = d_y + 1$ and $d_i' = d_i$ for all $i \neq x,y$.
Fix $x,y \in [r]$ and let $H = ([r],E)$ be a graphical realization of the degree sequence $\delta'$ where $\delta_x' = \delta_x + 1$, $\delta_y' = \delta_y - 1$ and $\delta_i' = \delta_i$ for all $i \in [r] \setminus \{x,y\}$. Then there exists an alternating path of length at most four starting at $x$ and ending at $y$.
\end{lemma}

Now, let $G \in \mathcal{G}'(c,d)$ for some tuple $(c',d')$ as in the definition of $\mathcal{G}'(c,d)$ at the start of Section \ref{sec:auxiliary}. We show that with at most four hinge flip operations, we can transform $G$ into a perturbed auxiliary state $G^* \in \mathcal{G}'(c,d)$ for which its tuple $(c^*,d^*)$  is edge-balanced.

\begin{itemize}
\item \textbf{Case 1: $c_{12}' = c_{12} + 1.$} Then, by Proposition \ref{prop:basic_properties}, either $c_{11}' = c_{11} - 1$ and $c_{22}' = c_{22}$, or, $c_{22}' = c_{22} - 1$ and $c_{11}' = c_{11}$. Assume without loss of generality that we are in the first case. Then it holds that 
\begin{equation}\label{eq:sparse1}
\sum_{j \in V_1} d_j' = \left( \sum_{j \in V_1} d_j \right) - 1 \ \ \ \text{ and } \ \ \ \sum_{j \in V_2} d_j' = \left( \sum_{j \in V_2} d_j \right) + 1. 
\end{equation}
Moreover, there must be at least one node $v_2 \in V_2$ with a degree surplus (of either one or two). 

\begin{itemize}
\item \textbf{Case A: $v_2$ has a neighbor $v_1 \in V_1$.} 
By assumptions $i)$ and $ii)$ it must be that there is some $b \in V_1$ so that $v_1$ is not a neighbor of $b$. Then we can perform the hinge flip that removes $\{v_2,v_1\}$ and adds $\{v_1,b\}$, resulting in an edge-balanced realization. 
\item \textbf{Case B: $v_2$ has no neighbors in $V_1$.} Since $c_{12} \geq 1$ by assumption $iii)$, there is some edge $\{a,b\}$ in $G$ with $a \in V_1$ and $b \in V_2 \setminus \{v_2\}$. We consider the induced subgraph $H$ on the nodes in $V_2$ and use $\delta'$ to denote its degree sequence. We next apply Lemma \ref{lem:alternating_path} with $v_2 = x$ and $b = y$. To see that this is possible, note that $c_{12} \leq \alpha n/2$ by assumption $iii)$. This implies, in combination with assumption $i)$ that at least $\alpha n/2$ nodes in $V_1$ have degree at least one in $H$. Thus, we may apply Lemma \ref{lem:alternating_path} with $r = |V(H)| \geq \alpha n /2$, since $d_i \leq \sqrt{\alpha n/4}$ by assumption $ii)$ (which is less or equal than $\sqrt{r/2}$ if $r \geq \alpha n/2$) . Hence, there exists an alternating path of length at most four starting at $v_2$ and ending at $b$. Then by performing two hinge flips, resulting in the removal of $\{v_2,f\}, \{g,h\}$ and addition of $\{f,g\}$, $\{h,b\}$, we are in the situation of Case A where $b$ now plays the role of $v_2$.

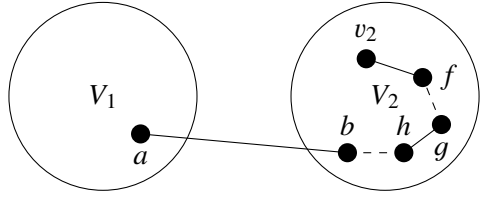
\begin{figure}[ht!]
\centering
\begin{tikzpicture}[scale=0.5]
\coordinate (a) at (2,0.5); 
\coordinate (b) at (7.5,0);
\coordinate (v2) at (8,2.5);
\coordinate (f) at (9.5,2);
\coordinate (g) at (10,0.75);
\coordinate (h) at (9,0);

\node at (a) [circle,scale=0.7,fill=black] {};
\node (A) [below=0.1cm of a]  {$a$};
\node at (f) [circle,scale=0.7,fill=black] {};
\node (F) [right=0.1cm of f]  {$f$};
\node at (g) [circle,scale=0.7,fill=black] {};
\node (G) [below=0.1cm of g]  {$g$};
\node at (h) [circle,scale=0.7,fill=black] {};
\node (H) [above=0.1cm of h]  {$h$};
\node at (b) [circle,scale=0.7,fill=black] {};
\node (B) [above=0.1cm of b]  {$b$};
\node at (v2) [circle,scale=0.7,fill=black] {};
\node (V2) [above=0.1cm of v2]  {$v_2$};

\path[every node/.style={sloped,anchor=south,auto=false}]
(a) edge[-] node {} (b)
(b) edge[-,dashed] node {} (h)
(h) edge[-] node {} (g)
(g) edge[-,dashed] node {} (f)
(f) edge[-] node {} (v2);

\draw (1,1.5) circle (2.5cm) node{$V_1$};
\draw (8.5,1.5) circle (2.5cm) node{$V_2$};
\end{tikzpicture}
\caption{Sketch of subcase B in Case 1.}
\end{figure} 
\end{itemize}

\item \textbf{Case 2: $c_{12}' = c_{12} - 1.$} Suppose without loss of generality that $c_{11}' = c_{11} + 1$ and $c_{22}' = c_{22}$, and note that there must be at least one node $v_1 \in V_1$ with a degree surplus. If $v_1$ has a neighor in $V_1$ we can perform a hinge flip operation to obtain an edge-balanced realization as desired (as this neighbor has at least one non-neighbor in $V_2$ by assumptions $i)$ and $ii)$). Therefore, assume that all neighbors of $v_1$ lie in $V_2$. Pick some neighbor $b \in V_2$ of $v_1$. Since all nodes in $V_1$ have degree at least one, and $c_{12} \leq \alpha n/2$, it must be that $b$ is not adjacent to some node $a \in V_1$ that has degree at least one in the induced subgraph $H$ on the nodes of $V_1$, as $d_b \leq \sqrt{\alpha n/4}$ and $|V(H)| \geq \alpha n/2$. Let $c$ be some neighbor of $a$ in $V_1$ that exists by assumption that $a$ has degree at least one in $H$. Also, $c$ is not adjacent to some $d \in V_2$ for similar reasons as that $b$ was not adjacent to $a$. This means that with two hinge flip operations we can obtain an edge-balanced realization.

\begin{figure}[ht!]
\centering
\begin{tikzpicture}[scale=0.5]
\coordinate (v1) at (2,0.5); 
\coordinate (b) at (7.5,0);
\coordinate (a) at (2,2);
\coordinate (c) at (2,3);
\coordinate (d) at (9,2);

\node at (v1) [circle,scale=0.7,fill=black] {};
\node (V1) [below=0.1cm of v1]  {$v_1$};
\node at (a) [circle,scale=0.7,fill=black] {};
\node (A) [below=0.1cm of a]  {$a$};
\node at (b) [circle,scale=0.7,fill=black] {};
\node (B) [right=0.1cm of b]  {$b$};
\node at (c) [circle,scale=0.7,fill=black] {};
\node (C) [left=0.1cm of c]  {$c$};
\node at (d) [circle,scale=0.7,fill=black] {};
\node (D) [above=0.1cm of d]  {$d$};

\path[every node/.style={sloped,anchor=south,auto=false}]
(v1) edge[-] node {} (b)
(b) edge[-,dashed] node {} (a)
(a) edge[-] node {} (c)
(c) edge[-,dashed] node {} (d);

\draw (1,1.5) circle (2.5cm) node{$V_1$};
\draw (8.5,1.5) circle (2.5cm) node{$V_2$};
\end{tikzpicture}
\caption{Sketch of Case 2.}
\end{figure}
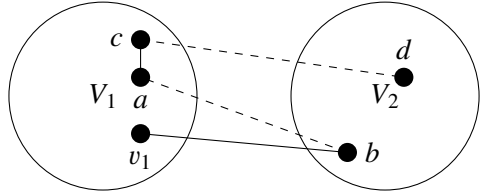 

\item \textbf{Case 3: $c_{12}' = c_{12}$}. If also $c_{11} = c_{11}'$ we are done. Otherwise, suppose that $c_{11} = c_{11}' + 1$. Then it must be that $c_{22} = c_{22}' - 1$, as $c_{12}' = c_{12}$, and it holds that 
\begin{equation}\label{eq:sparse2}
\sum_{j \in V_1} d_j' = \left( \sum_{j \in V_1} d_j \right) + 2 \ \ \ \text{ and } \ \ \ \sum_{j \in V_2} d_j' = \left( \sum_{j \in V_2} d_j \right) - 2. 
\end{equation}
There is at least one node $v_1 \in V_1$ with a degree surplus. If $v_1$ has a neighbor $a \in V_1$, which in turn has a non-neighbor $b \in V_2$ by assumptions $i)$ and $ii)$, then we can reduce to the situation of Case 1 by performing a hinge flip removing $\{v_1,a\}$ and adding $\{a,b\}$. Therefore, assume that all neighbors of $v_1$ are in $V_2$. We can then pick some neighbor $b \in V_2$ and perform a similar step as in Case 2 to find an alternating path $(v_1,b,a,c)$ with $v_1,a,c \in V_1$ and $b \in V_2$. Then we can perform two hinge flip operations to reduce to Case 1 again.
\end{itemize}

We have shown that with at most three hinge flips we can always obtain some $G^* \in \mathcal{G}'(c,d)$ that is edge-balanced. This implies that
\begin{equation}\label{eq:sparse3}
\sum_{j \in V_1} d_j^* =  \sum_{j \in V_1} d_j \ \ \ \text{ and } \ \ \ \sum_{j \in V_2} d_j^* =  \sum_{j \in V_2} d_j. 
\end{equation}
Now, if $v \in V_1$ has a degree surplus, there must be some $w \in V_1$ that has a degree deficit, because of (\ref{eq:sparse3}). Moreover, if all neighbors in of $v$ are in $V_2$, we can transfer the degree deficit to some node with degree at least one in the subgraph induced on $V_1$ at the cost of one hinge flip operation (similar as the analysis in Case 2). That is, we may assume that $v$ has a neighbor in $V_1$. Then, using similar arguments as in Case 1.B, it follows that there exists an alternating path from $v$ to $w$ in $V_1$, which allows us to decrease the total degree surplus/deficit using two hinge flip operations. This can be repeated to obtain a feasible realization $H \in \mathcal{G}(\gamma,d)$. 

Overall, this can be done using at most nine hinge flip operations.
\end{proof}

In particular, Theorem \ref{cor:stable_sparse} directly implies the following which, to the best of our knowledge, is the first result of its kind for the PAM model.
\begin{corollary}\label{cor:first_sparse}
Let $\mathcal{D}_{\alpha}$ be as in Theorem \ref{cor:stable_sparse}. Then there is an fully polynomial almost uniform generator for tuples in the family $\mathcal{D}_{\alpha}$.
\end{corollary}

%In the next section we show that, in the case of two regular degree classes, we can even prove that the switch chain mixes rapidly (and not only the hinge flip chain). However, for the PAM model it is not known when the (restricited) switch chain is irreducible. In particular, in \cite{ErdosHIM2017} it is shown that the switch chain is in general not irreducible. However, there in \cite{ErdosHIM2017} it is also shown that for the PAM model with two classes, the set of all realizations is connected under so-called \emph{double switches}.

%\newpage
\bigskip

%%%%%%%%%%%%%%%%%%%%%%%%%%%%%%%%%%%%%%%%
%%%%%%%%%%%%%%%%%%%%%%%%%%%%%%%%%%%%%%%%
\section{Switch Markov Chain for Two Regular Degree Classes}\label{sec:switch_for_jdm}
%%%%%%%%%%%%%%%%%%%%%%%%%%%%%%%%%%%%%%%%
%%%%%%%%%%%%%%%%%%%%%%%%%%%%%%%%%%%%%%%%
In this section we will use an embedding argument similar to that in the proof of Theorem \ref{thm:transformation} to show the switch chain is rapidly mixing in case both classes are regular, i.e., for instances that are essentially JDM instances with two degree classes. Recall the  \emph{restricted switch Markov chain} for sampling graphical realizations of $\mathcal{G}(c,d)$ defined in Section \ref{sec:preliminaries}: Let $G$ be the current state of the \emph{switch chain}. 
\begin{itemize}
	\item With probability $1/2$, do nothing.
	\item Otherwise, perform a \emph{switch move}: select two edges $\{a,b\}$ and $\{x,y\}$ uniformly at random, and select a perfect matching $M$ on nodes $\{x,y,a,b\}$ uniformly at random (there are three possible options). If $M \cap E(G) = \emptyset$ and $G + M - \{a,b\} \cup \{x,y\} \in \mathcal{G}(c,d)$, then delete $\{a,b\}, \{x,y\}$ from $E(G)$ and add the edges of $M$. This local operation is called a \emph{switch}.
	%\item With probability $1/4$ perform a \emph{double switch move}: \pkblue{to do...}
\end{itemize}
  As usual, graphs $G,G' \in \mathcal{G}(c,d)$ are said to be \emph{switch adjacent} if $G$ can be obtained from $G'$ with positive probability in one transition of this chain. The \emph{switch-distance} $\text{dist}_{\mathcal{G}(c,d)}(G,G')$ is the length of a shortest path between $G$ and $G'$ in the state space graph of the switch chain. Also recall, $P(G,G')^{-1} \leq n^4$ for all adjacent $G, G' \in \mathcal{G}'(c,d)$, and also the maximum in- and out-degrees of the state space graph of the switch chain are bounded by $n^4$.
%  
%\begin{assumption}[Irreducibility and feasibility]\label{ass:irreducibility}
%For given $c$ and $d$, we assume that $\mathcal{G}(c,d)$ is non-empty and that an element from it can be computed in time polynomial in $n$. Moreover, we assume that the switch chain is irreducible, i.e., the partition adjacency matrix $c$ and degree sequence $d$ are such that for any $H,H' \in \mathcal{G}(c,d)$ there is a directed path from $H$ to $H'$ in the state space graph of the switch chain.
%\end{assumption}
%
%\begin{theorem} \label{thm:switch_properties}
%The switch chain is aperiodic and symmetric, and, hence (under Assumption \ref{ass:irreducibility}), has uniform stationary distribution over $\mathcal{G}(c,d)$. Moreover, $P(G,G')^{-1} \leq n^4$ for all adjacent $G, G' \in \mathcal{G}'(c,d)$, and also 
%the maximum in- and out-degrees of the state space graph of the switch chain are bounded by $n^4$.
%\end{theorem}

While this chain is known to be irreducible for the instances of the JDM model \cite{Amanatidis2015,CzabarkaDEM15}, in general this is not true \cite{ErdosHIM2017}. To the best of our knowledge, there is no clear understanding for which pairs $c$ and $d$ it is irreducible in general. Nevertheless, we present the following meta-result for the rapid mixing of the switch chain, which in particular applies in case the degrees are component-wise regular (Theorem \ref{thm:stable_jdm1}).  %For the case of two classes, however, the realization problem is polynomial time solvable, and (an adjusted version of) the switch chain is irreducible if one also allows so-called double switches \cite{ErdosHIM2017}. 

%\pkrem{More/different references here.}
%One special case for which Assumption \ref{ass:irreducibility} is satisfied, is when all degrees within a given set of the partition are regular, that is, if $d_a = d_b$ for $a,b \in V_i$ for any $i = 1,\dots,q$. We denote the common degree in class $i$ by $\beta_i$. This corresponds to the \emph{Joint Degree Matrix model} introduced in \cite{Amanatidis2015}. Moreover, for this model it is also well-known that an initial state for the switch chain above can be computed in time polynomial in $n$, and that the switch chain is irreducible \cite{Amanatidis2015}. The following is the main result of this section.

\begin{theorem}\label{thm:switch_2} Let $\mathcal{D}$ be a strongly stable family of tuples $(\gamma,d)$ with respect to some constant $k$, and suppose there exists a function $p_0 : \N \rightarrow \N$ with the property that, for any fixed $x \in \N$: if $(\gamma,d) \in \mathcal{D}$, and $G, G' \in \mathcal{G}(\gamma,d)$ so that $|E(G)\triangle E(G')| \leq x$, the switch-distance satisfies
$
\text{dist}_{\mathcal{G}(\gamma,d)}(G,G') \leq p_0(x).
$\label{prop:switch2}
Then the switch chain is rapidly mixing for all tuples in the family $\mathcal{D}$ with respect to the uniform stationary distribution over $\mathcal{G}(\gamma,d)$.
\end{theorem}
\begin{proof}[Proof (sketch)]
First note that by definition of the function $p_0$ the switch chain is irreducible. Moreover, it is not hard to see that the switch chain is aperiodic and symmetric as well. This implies that it has a unique stationary distribution which is the uniform distribution over $\mathcal{G}(\gamma,d)$. Moreover, by assumption of strong stability, we know that the auxiliary chain $\mathcal{M}(\gamma,d)$ is rapidly mixing. In particular, from Theorem \ref{thm:auxiliary}, we know there exists a flow $f$  that routes $1/|\mathcal{G}'(\gamma,d)|^2$ units of flow between any pair of states in $\mathcal{G}(\gamma,d)$ in the state space graph of the chain $\mathcal{M}(\gamma,d)$, with the property that $f(e) \leq p(n)/|\mathcal{G}'(\gamma,d)|$, and $\ell(f) \le r(n)$, for some polynomials $p(\cdot), r(\cdot)$ whose degrees may only depend on $k(\gamma,d)$.

One can then use exactly the same embedding argument as in the proof of Theorem \ref{thm:transformation}. The existence of the function $p_0$, together with the notion of strong stability as defined in Appendix \ref{sec:auxiliary}, are sufficient for reproducing all the arguments.
\end{proof}

%\subsection{Regular families}
%We show that when both classes have regular degrees, the switch chain is rapidly mixing. We do this by showing that the properties in the statement of Theorem \ref{thm:switch} are satisfied.
%%We obtain the following for the Joint Degree Matrix model with two regular classes. We assume that the regular degrees in the classes are so that for any graphical realization $G \in \mathcal{G}'(\gamma,d)$ no node $v$ can be adjacent to all nodes in $V_1$ and also not adjacent to all nodes in $V_2$.

\begin{rtheorem}{Theorem}{\ref{thm:stable_jdm1}}
	Let $\mathcal{D}$ be the family of instances of the joint degree matrix model with two degree classes. Then the switch chain is rapidly mixing for  instances in $\mathcal{D}$.
\end{rtheorem}

\begin{proof}
Strong stability was shown in the previous section in Theorem \ref{cor:stable_jdm}. Moreover, from the proof of Lemma 7 in \cite{Amanatidis2015} it follows that for any two graphs $H, H' \in \mathcal{G}(\gamma,d)$, $H$ can be transformed into $H'$ using at most $\frac{3}{2}|E(H) \triangle E(H')|$ switches of the restricted switch chain. That is, we may take $p_0(x) = \frac{3}{2}x$. Then the statement follows from Theorem \ref{thm:switch_2}.
\end{proof}

%\newpage

\bigskip

%%%%%%%%%%%%%%%%%%%%%%%%%%%%%%%%%%%%%%%%
%%%%%%%%%%%%%%%%%%%%%%%%%%%%%%%%%%%%%%%%
\addcontentsline{toc}{section}{References}
\bibliographystyle{plain}
%\bibliography{references-03,references_04}
\bibliography{references_04}
%%%%%%%%%%%%%%%%%%%%%%%%%%%%%%%%%%%%%%%%
%%%%%%%%%%%%%%%%%%%%%%%%%%%%%%%%%%%%%%%%

\end{document}